\documentclass[10pt]{article}

\usepackage[
  top=0.85in,
  bottom=0.85in,
  left=0.9in,
  right=0.9in
]{geometry}

\usepackage[noadjust]{cite}
\usepackage{amsmath,amssymb,amsthm,mathtools}
\usepackage{bm}
\usepackage{enumitem}
\usepackage{xcolor}
\usepackage{listings}
\usepackage{silence}
\usepackage{titlesec}

\allowdisplaybreaks[4]

\setlist{itemsep=1pt,topsep=2pt,parsep=0pt,partopsep=0pt}
\setlist[enumerate,1]{label=(\roman*),ref=\roman*}
\setlist[enumerate,2]{label=(\alph*),ref=\theenumi(\alph*)}

\titlespacing*{\section}
{0pt}{10pt}{5pt}

\titlespacing*{\subsection}
{0pt}{7pt}{3pt}

\makeatletter

\@addtoreset{equation}{section}
\makeatother

\newtheorem{theorem}{Theorem}[section]
\newtheorem{lem}[theorem]{Lemma}
\newtheorem{cor}[theorem]{Corollary}
\newtheorem{prop}[theorem]{Proposition}
\newenvironment{lemma}{\begin{lem}}{\end{lem}}
\newenvironment{proposition}{\begin{prop}}{\end{prop}}
\theoremstyle{definition}
\newtheorem{Definition}[theorem]{Definition}
\newtheorem*{example}{Example}

\newtheorem*{remarknonumber}{Remark}

\newcommand{\F}{\mathbb{F}}
\newcommand{\Hull}{\operatorname{Hull}}
\newcommand{\N}{\operatorname{N}}
\newcommand{\Span}{\operatorname{Span}}

\lstdefinestyle{magma}{
  basicstyle=\ttfamily\footnotesize,
  keywordstyle=\color{blue!60!black},
  commentstyle=\color{green!45!black},
  stringstyle=\color{purple!70!black},
  breaklines=true,
  frame=single,
  columns=fullflexible,
  showstringspaces=false
}

\begin{document}
	
\title{Subcodes with Prescribed Hull Dimensions in Linear Codes: Enumeration, Containment, and Asymptotic Distributions\thanks{E-mail addresses: xlh@mails.ccnu.edu.cn (L. Xiong), hwliu@ccnu.edu.cn (H. Liu), l\_r\_lin86@163.com (L. Lin)}}
\author{
Linhan Xiong$^{1}$,~Hongwei Liu$^{1}$,~Liren Lin$^{2}$
}

\date{
\small
$^{1}$School of Mathematics and Statistics,
Central China Normal University,
Wuhan 430079, China\\
$^{2}$School of Cyber Science and Technology,
Hubei University,
Wuhan 430062, China
}
\maketitle
\begin{abstract}
Mass formulas for prescribed hull dimensions have mainly been studied for linear codes viewed as subspaces of the whole space. In this paper, we consider the corresponding internal enumeration problem for an arbitrary fixed parent code. Let $C$ be an $[n,k]_q$ linear code with an $\ell$-dimensional hull, and let $s$ denote the dimension of the subcodes being counted. We determine the number of subcodes with prescribed dimension, hull dimension, and type over finite fields of both even and odd characteristic. Our approach is based on decomposing the parent code into an {\rm LCD} complement and its hull, and then using projection to reduce the enumeration problem to the corresponding problem inside the {\rm LCD} complement. This yields explicit formulas refined by complete type and shows that the resulting counts depend only on the relevant parameters, rather than on the particular parent code or its length. The known whole-space mass formulas and the enumeration of {\rm LCD} subcodes of {\rm LCD} codes arise as special cases.

Motivated by code extension, we also study the dual containment problem: for a fixed code, we enumerate larger codes containing it with prescribed hull dimension and type. By duality, this problem is converted into the subcode enumeration problem already solved above, leading to corresponding conditional mass formulas.

Finally, we investigate the asymptotic distributions arising from the exact enumeration formulas. When both the subcode dimension $s$ and $k-\ell-s$ tend to infinity, the hull distribution is governed by the {\rm LCD} complement, and the limiting proportions are explicit constants determined by the field, the prescribed hull dimension, and the relevant type. When $s$ approaches or exceeds the dimension $k-\ell$ of the {\rm LCD} complement, we determine the asymptotic proportion of subcodes attaining the minimum possible hull dimension and establish the associated concentration behavior. Analogous asymptotic results are obtained for the containment problem.
\end{abstract}

\noindent\textbf{Keywords:} hull; {\rm LCD} codes; subcodes; mass formula; containment enumeration; asymptotic distribution

2020 \emph{Mathematics Subject Classification}:  94B05, 05A15

\section{Introduction}\label{sec-intro}

The hull of a linear code was introduced by Assmus and Key~\cite{hull}. 
For a linear code $C\subseteq \mathbb{F}_q^n$, its Euclidean hull is
\[
\operatorname{Hull}(C)=C\cap C^\perp.
\]
It records the self-orthogonal part contained in $C$. The two extreme
cases are classical: $\operatorname{Hull}(C)=\{0\}$ gives a linear
complementary dual (LCD) code, whereas $C\subseteq C^\perp$ gives a
self-orthogonal code. Prescribing the hull dimension therefore provides
a natural way to interpolate between LCD and self-orthogonal codes.
Hulls also arise in the construction of entanglement-assisted quantum
error-correcting codes; see, for example,~\cite{Dcc-k}. In addition, LCD
codes have applications in countermeasures against side-channel and
fault-injection attacks~\cite{C2,lcd-application}.

LCD codes were introduced by Massey~\cite{LCD-Massey}. Sendrier investigated
hull dimensions and proved that LCD codes meet the asymptotic
Gilbert--Varshamov bound~\cite{sendrier hull,GV bound}. A renewed interest in LCD codes
was stimulated by Carlet et al.~\cite{C1}, who showed that every linear
code over $\mathbb{F}_q$ with $q>3$ is equivalent to a Euclidean LCD
code. Since then, LCD codes and linear codes with prescribed hull
dimensions have been studied from structural, constructive,
classification, enumerative, and distance-optimal viewpoints; see, for
instance,~\cite{A1,A2,A3,An-LCD-embed,BLCD,binaryLCD,Galvez-LCD,Li cc,FFA,Wang}.

Subcodes of structured codes form another natural object in coding
theory. They occur in the study of optimum distance profiles and in
constructions that pass from a known code to smaller codes with
prescribed properties. Luo et al.~\cite{ODPs} introduced optimum distance
profiles. Freibert and Kim studied optimal subcodes and optimum distance
profiles first for formally self-dual codes and then for self-dual
codes~\cite{subcode-formal,subcode}, while Harada and Saito~\cite{Hcc} considered
subcodes of LCD codes. These works motivate studying how a
coding-theoretic property is distributed inside a fixed parent code,
rather than only among all subspaces of an ambient vector space.

A related line of work concerns the largest possible minimum distance
under orthogonality or hull constraints. For binary LCD codes, Galvez
et al.~\cite{Galvez-LCD} determined exact optimal distances for all lengths up
to twelve. Kim~\cite{Kim-hull-opt} studied binary optimal linear codes with
prescribed hull dimensions, and Li, Shi and Kim~\cite{Li-Shi-Kim-h1} developed a
detailed theory for the one-dimensional-hull case. At the
self-orthogonal boundary, shortest embeddings and optimal minimum
distances were investigated in~\cite{Kim-Kim-Lee-SO,Kim-Choi-SO,Shi-Li-Kim-SO,Shi-Li-Helleseth-Kim-SO}, while related
classifications of binary self-orthogonal codes were obtained
in~\cite{Shi-Liu-Kim-class}. More recently, An et al.~\cite{An-LCD-embed} developed shortest LCD
embeddings over binary, ternary, and quaternary fields and used them to
construct distance-optimal LCD codes. Additive self-dual codes over
fields of even order provide a broader related setting in which several
notions of duality and self-duality can be considered~\cite{Dougherty-Kim-Lee-additive}. These
distance-oriented results are not needed for our enumeration formulas,
but they provide useful context for the computational examples below,
where the best minimum distance attained in each prescribed-hull
subcode layer is compared with the corresponding global fixed-hull
optimum.

Mass formulas form an important enumerative line of research in coding
theory. Pless obtained classical mass formulas for self-orthogonal
codes using finite-geometric methods~\cite{Pless-mass-F2,Pless-mass-F}.
Sendrier subsequently developed a general mass formula for linear codes
with prescribed Euclidean hull dimension~\cite{sendrier hull}. The LCD
case, corresponding to trivial hull, was further studied by Carlet
et al.~\cite{binaryLCD}, who derived explicit mass formulas and related
structural and asymptotic results, while Liu and Wang~\cite{FFA}
obtained further structural descriptions and enumerative formulas.
Building on the prescribed-hull enumeration initiated by Sendrier, Li
and Shi~\cite{Li-hull-even} and Li, Shi and Ling~\cite{Li-hull-odd}
derived more explicit closed mass formulas in even and odd
characteristic, respectively, together with refinements and
classification consequences.

A broader finite-geometric formulation was developed by
Gluesing-Luerssen and Ravagnani~\cite{bilinear-space}, who enumerated
$\ell$-complementary subspaces in arbitrary nondegenerate finite
bilinear spaces. Thus, their framework contains, at the aggregate
level, prescribed-hull enumeration whenever the ambient code itself
carries a nondegenerate restricted bilinear form, as is the case for
an LCD code. Bouyuklieva, Bouyukliev and
Ozbudak~\cite{BBO-hull-sequence} studied the sequence of numbers of
linear codes of fixed length and dimension as the hull dimension
varies. They proved that, for $n\geq 2k$, the number of $[n,k]_q$
linear codes decreases as the hull dimension increases, and also
presented classification results for binary and ternary LCD and
self-orthogonal codes for several parameter sets.

The coding-theoretic mass formulas cited above count subspaces of the
whole Euclidean space $\mathbb{F}_q^n$. By contrast, the problem in this
paper is internal to a fixed parent code. We fix an arbitrary linear
code $C$, which is not assumed to be LCD, and ask how many of its
subcodes have prescribed dimension, hull dimension, and type. From this
viewpoint, the usual whole-space mass formula corresponds to the special
choice $C=\mathbb{F}_q^n$, whereas the LCD-subcode enumeration studied
in~\cite{x-lcd} corresponds to the special situation in which both the
parent code and the subcodes being counted have zero hull.

More precisely, let $C$ be an $[n,k]_q$ linear code with an
$\ell$-dimensional hull and type $\varepsilon$. For every admissible
subcode type $\tau$, we determine
\[
N_q\bigl[(s,t)_\tau;(k,\ell)_\varepsilon\bigr],
\]
the number of $s$-dimensional subcodes $D\subseteq C$ satisfying
\[
\dim\operatorname{Hull}(D)=t
\]
and having type $\tau$. The formulas show that this number depends only
on
\[
q,\qquad (s,t)_\tau,\qquad (k,\ell)_\varepsilon,
\]
and is independent of the particular representative of the parent code
and, in particular, of its length. Thus the prescribed-hull
distribution inside a fixed parent code is determined entirely by the
displayed parameters.

The main structural ingredient is the decomposition
\[
C={\rm L}\oplus \operatorname{Hull}(C),
\]
where ${\rm L}$ is an LCD complement of the hull. Let
\[
\Pi_{\rm L}:C\longrightarrow {\rm L}
\]
denote the projection onto ${\rm L}$ along $\operatorname{Hull}(C)$. For every
subcode $D\subseteq C$, the difference between the dimension of $D$ and
the dimension of its hull is preserved under projection. Equivalently,
\[
\dim\operatorname{Hull}(D)
=
\dim\bigl(D\cap\operatorname{Hull}(C)\bigr)
+
\dim\operatorname{Hull}\bigl(\Pi_{\rm L}(D)\bigr).
\]
Consequently, if
\[
m=\dim\bigl(D\cap\operatorname{Hull}(C)\bigr),
\]
then $\Pi_{\rm L}(D)$ has dimension $s-m$ and hull dimension $t-m$. For a
fixed projected subcode, the number of corresponding subcodes of $C$
is given by a Gaussian-coefficient factor. The fixed-parent problem is
therefore reduced to prescribed-hull enumeration inside the LCD
complement together with the enumeration of possible intersections
with $\operatorname{Hull}(C)$.

We first carry out this program in even characteristic. In this case we
retain the complete types ${\bf oo}$, ${\bf oe}$, ${\bf eo}$, and ${\bf ee}$. Starting from the
whole-space mass formulas and the LCD-subcode formulas, we determine
the prescribed-hull distributions inside LCD parent codes and then use
the projection described above to obtain the corresponding formulas
inside arbitrary parent codes. We then treat odd characteristic, where
the two types $+$ and $-$ are retained. The same projection principle
again reduces the arbitrary-parent problem to the corresponding
enumeration inside an LCD complement. In both characteristics, the
resulting formulas recover the known whole-space prescribed-hull mass
formulas when $C=\mathbb{F}_q^n$ and recover the LCD-subcode problem
when $t=0$.

Motivated by code extension, we also consider the dual containment
problem. Let $D\subseteq\mathbb{F}_q^n$ be a fixed code, and suppose that
one wishes to count $k$-dimensional codes $C$ satisfying
\[
D\subseteq C,\qquad \dim\operatorname{Hull}(C)=\ell,
\]
with a prescribed type. This problem is naturally related to the
subcode enumeration problem through
\[
D\subseteq C
\quad\Longleftrightarrow\quad
C^\perp\subseteq D^\perp.
\]
Since duality preserves hull dimension and transforms the relevant
types explicitly, the fixed-parent subcode formulas yield corresponding
conditional mass formulas for codes containing a prescribed subcode.
If $D$ has parameters $(s,t)_\tau$, we denote the resulting extension
count by
\[
N_{q,n}\bigl[(k,\ell)_\varepsilon\mid (s,t)_\tau\bigr].
\]

A further part of the paper is devoted to the asymptotic consequences
of the exact enumeration formulas. For a parent code with parameters
$(k,\ell)$, put
$
a=k-\ell-s.
$
The parameter $a$ measures the difference between the dimension of the
LCD complement and the dimension of the subcodes being counted, and it
organizes the asymptotic behavior into several regimes.

When both $s$ and $a$ tend to infinity, subcodes meeting
$\operatorname{Hull}(C)$ nontrivially become asymptotically negligible
after normalization. The hull distribution is therefore governed by
the LCD complement, yielding explicit limiting constants and the
corresponding dominant-type statements. For an LCD parent, this reduces
to the usual regime in which both the subcode dimension and its
codimension tend to infinity.

The behavior changes near the threshold $s=k-\ell$. If $a>0$ is fixed
while both $k-\ell$ and $\ell$ grow, we determine the asymptotic
proportion of subcodes attaining the minimum possible hull dimension.
At the threshold and above it, namely when $a\leq 0$, every
$s$-dimensional subcode has hull dimension at least $-a$. We determine
the exact number of subcodes attaining this minimum, obtain explicit
limits when $a$ is fixed, and show that when $a\to-\infty$ almost all
subcodes attain the minimum possible hull dimension. Under additional
dimension assumptions, the complete type also concentrates. Thus the exact formulas lead to a three-regime asymptotic picture for the principal scaling limits, organized by the difference
$
a=k-\ell-s.
$

The same asymptotic picture transfers to the containment problem by
duality. If the prescribed subcode $D$ has dimension $s$ and hull
dimension $t$, the corresponding parameter is
$
b=k-s-t.
$
When the number of added dimensions is sufficiently large relative to
$t$, the extension distribution is governed by the LCD complement of
$D^\perp$. Near the transition range one obtains explicit
minimum-hull proportions, whereas for $b\to-\infty$ almost all
extensions attain the minimum possible hull dimension $-b$.

The paper is organized as follows. Section~2 introduces the notation, type classifications, puncturing and projection tools, and LCD-complement decompositions used throughout the paper. Sections~3 and~4 derive the exact subcode enumeration formulas in even and odd characteristic, respectively. Section~5 develops the containment enumeration through duality and then studies the asymptotic behavior of prescribed-hull subcodes and extensions across the finite-dimensional, large-dimensional, fixed-difference, and concentration regimes. Section~6 concludes the paper. Appendix~A contains the proofs of the main results from Subsections~5.2 and~5.3, while Appendix~B presents binary and ternary computational examples illustrating the prescribed-hull subcode distributions and the best minimum distances attained in the corresponding subcode layers.

\section{Preliminaries}\label{sec-1}

In this section, we fix the notation used throughout the paper and
collect several structural results on hulls, types, puncturing maps,
projection maps, and LCD complements of hulls.

\subsection{Basic notation and definitions}

Throughout this paper, we adopt the following notation. Let $O$ denote a zero matrix of an appropriate size, $\mathbf{1}_n$ the all-one vector of length $n$, and $E_n$ the identity matrix of order $n$. For any matrix $A=(a_{ij})$ over $\F_q$, its transpose is denoted by $A^T=(a_{ji})$.

Let $\F_q$ be the finite field with $q$ elements, where $q$ is a prime power, and let $\F_q^*=\F_q\setminus\{0\}$ denote its multiplicative group. The vector space spanned by $\bm\alpha_1,\bm\alpha_2,\ldots,\bm\alpha_r$ is denoted by
\[
\Span\{\bm\alpha_1,\bm\alpha_2,\ldots,\bm\alpha_r\}.
\]
For a finite set $S$, its cardinality is denoted by $|S|$.
For a real number $x$, the floor $\lfloor x\rfloor$ denotes the greatest
integer not exceeding $x$, while the ceiling $\lceil x\rceil$ denotes the
least integer not smaller than $x$.

For integers $n$ and $k$, the $q$-binomial coefficient, also called the Gaussian coefficient, is defined by
\[
\genfrac{[}{]}{0pt}{}{n}{k}_q=
\begin{cases}
\displaystyle
\frac{(q^n-1)(q^n-q)\cdots(q^n-q^{k-1})}
{(q^k-1)(q^k-q)\cdots(q^k-q^{k-1})},
&\text{if $1\leq k\leq n$},\\[3mm]
1,&\text{if $k=0$ and $n\geq0$},\\
0,&\text{otherwise.}
\end{cases}
\]
The Gaussian coefficient $\genfrac{[}{]}{0pt}{}{n}{k}_q$ is the number of $k$-dimensional subspaces of $\F_q^n$. It satisfies
\[
\genfrac{[}{]}{0pt}{}{n}{k}_q=\genfrac{[}{]}{0pt}{}{n}{n-k}_q.
\]

Let $\operatorname{GL}_n(\F_q)$ denote the general linear group of degree $n$ over $\F_q$, consisting of all $n\times n$ invertible matrices over $\F_q$. The orthogonal group is defined by
\[
\mathbb{O}_n(q)
=
\left\{
Q\in\operatorname{GL}_n(\F_q)
\mid
QQ^T=E_n
\right\}.
\]

An $[n,k]_q$ linear code $C$ is a $k$-dimensional subspace of $\F_q^n$, and a subcode of $C$ is a linear subspace of $C$. A generator matrix of $C$ is a $k\times n$ matrix whose rows form a basis of $C$. For a matrix $G$, its row space is denoted by $\operatorname{Row}(G)$.

For a vector $\mathbf{c}\in\F_q^n$, its Hamming weight, denoted by $\operatorname{wt}(\mathbf{c})$, is the number of its nonzero coordinates. The minimum distance of a nonzero linear code $C$ is defined by
\[
d(C)
=
\min\left\{
\operatorname{wt}(\mathbf{c})
\mid
\mathbf{c}\in C\setminus\{\mathbf{0}\}
\right\}.
\]

If $d(C)=d$, then $C$ is also called an $[n,k,d]_q$ linear code.

For two linear codes $C_1,C_2\subseteq\F_q^n$, their sum is defined by
\(
C_1+C_2
=
\left\{
\mathbf{c}_1+\mathbf{c}_2
\mid
\mathbf{c}_1\in C_1,\ 
\mathbf{c}_2\in C_2
\right\}.
\)
If $C_1\cap C_2=\{\mathbf{0}\}$, then the sum is called a direct sum and is denoted by $C_1\oplus C_2$.

For any linear code $C\subseteq\F_q^n$ and any $n\times n$ matrix $Q$ over $\F_q$, the linear code $CQ$ is defined by
\(
CQ
=
\left\{
\mathbf{c}Q
\mid
\mathbf{c}\in C
\right\}.
\)

The Euclidean inner product of two vectors $\mathbf{x}=(x_1,\ldots,x_n)$ and $\mathbf{y}=(y_1,\ldots,y_n)$ in $\F_q^n$ is defined by
\(
\langle\mathbf{x},\mathbf{y}\rangle
=
\sum_{i=1}^n x_i y_i.
\)
A vector $\mathbf{x}$ is self-orthogonal if
\(
\langle\mathbf{x},\mathbf{x}\rangle=0,
\)
and is non-self-orthogonal otherwise.

For any linear code $C\subseteq\F_q^n$, its Euclidean dual code is defined by
\[
C^\perp
=
\left\{
\mathbf{x}\in\F_q^n
\mid
\langle\mathbf{x},\mathbf{c}\rangle=0
\text{ for all }\mathbf{c}\in C
\right\}.
\]
A linear code $C$ is self-orthogonal if $C\subseteq C^\perp$.

The hull of a linear code $C$ is defined by
\(
\Hull(C)
=
C\cap C^\perp.
\)
If $\Hull(C)=\{\mathbf{0}\}$, then $C$ is called a linear complementary dual ({\rm LCD}) code.

When $q$ has even characteristic, a vector $\mathbf{x}=(x_1,\ldots,x_n)$ is even-like if
\(
\sum_{i=1}^n x_i=0,
\)
and is odd-like otherwise. A linear code is even-like if all its codewords are even-like, and is odd-like otherwise.

We use $C$ for a fixed parent code and $D$ for a subcode being counted. When only the parameters and type of the parent matter, $C^{k,\ell}_{\varepsilon}$ denotes an arbitrary fixed $[n,k]_q$ code of type $\varepsilon$ with an $\ell$-dimensional hull; the length is suppressed when it is clear from the context. The notation ${\rm L}^k_\varepsilon$ is reserved for an {\rm LCD} code of type $\varepsilon$. Whenever a decomposition $C={\rm L}\oplus\Hull(C)$ is used, ${\rm L}$ denotes an {\rm LCD} complement of the hull. The letters $V$ and $W$ are used only in the general direct-sum and projection statements of Section~\ref{sec-1}; in their application to a hull decomposition one takes $W=\Hull(C)$ and $V={\rm L}$. Thus generic statements are written in terms of $C$ and $D$, whereas $C^{k,\ell}_{\varepsilon}$ and ${\rm L}^k_\varepsilon$ are used when a particular type or a standard representative is required. The admissible
type sets are
\[
\{{\bf oo},\quad {\bf oe},\quad {\bf eo},\quad {\bf ee}\}
\quad\text{if $q$ is even},
\qquad
\{+,\quad -\}
\quad\text{if $q$ is odd}.
\]
The precise type definitions are given in
Subsections~\ref{subsec:even-types} and~\ref{subsec:odd-types}.

Let
\[
\mathcal{S}^t(s;C)
=
\left\{
D\subseteq C
\mid
\dim(D)=s,\ 
\dim(\Hull(D))=t
\right\}
\]
denote the set of all $s$-dimensional subcodes of $C$ with a $t$-dimensional hull, and let
\(
\N^t(s;C)
=
|\mathcal{S}^t(s;C)|.
\)
For a complete type $\tau$, let
\[
\mathcal{S}^t_{\tau}(s;C)
=
\left\{
D\in\mathcal{S}^t(s;C)
\mid
D\text{ is of type }\tau
\right\},
\]
and define
\(
\N^t_{\tau}(s;C)
=
|\mathcal{S}^t_{\tau}(s;C)|.
\)
These set-based symbols will be retained in proofs when a specific
representative is needed. At that level,
\(
\N^t(s;C)
=
\sum_{\tau}\N^t_{\tau}(s;C),
\)
where the sum is taken over all complete types.

For the final parameterized counts, let $C$ be an $[n,k]_q$ linear
code with an $\ell$-dimensional hull and type $\varepsilon$. We use
\[
\N_q\bigl[(s,t)_\tau;(k,\ell)_\varepsilon\bigr]
\]
to denote the number of $s$-dimensional subcodes $D\subseteq C$ having
a $t$-dimensional hull and type $\tau$. Here $(s,t)_\tau$ describes
the subcode being counted, whereas $(k,\ell)_\varepsilon$ describes
the parent code; the two parameter groups are separated by a semicolon.
The formulas established in Sections~\ref{sec-2} and~\ref{sec-3} show
that this number depends only on
\[
q,\qquad (s,t)_\tau,\qquad (k,\ell)_\varepsilon,
\]
and is independent of the particular choice and the length of $C$.
When the subcode type is not distinguished, we use the type-summed
shorthand
\[
\N_q\bigl[(s,t);(k,\ell)_\varepsilon\bigr]
:=
\sum_\tau
\N_q\bigl[(s,t)_\tau;(k,\ell)_\varepsilon\bigr],
\]
where the sum is over all admissible complete subcode types. Thus, for
a particular representative $C^{k,\ell}_\varepsilon$, this shorthand
equals $\N^t(s;C^{k,\ell}_\varepsilon)$.

For the containment problem, let $D\subseteq\F_q^n$ be a fixed
$[n,s]_q$ linear code with a $t$-dimensional hull and type $\tau$. We
use
\[
\N_{q,n}\bigl[(k,\ell)_\varepsilon\bigm|(s,t)_\tau\bigr]
\]
to denote the number of $k$-dimensional codes $C\subseteq\F_q^n$
having an $\ell$-dimensional hull and type $\varepsilon$ such that
$D\subseteq C$. In this notation, $(k,\ell)_\varepsilon$ describes
the extension code being counted, while $(s,t)_\tau$ describes the
fixed contained code. The vertical bar indicates the containment
condition. Unlike the subcode count, the containment count generally
depends on the length $n$, and hence $n$ is retained in the subscript.
The corresponding type-summed shorthand is
\[
\N_{q,n}\bigl[(k,\ell)\bigm|(s,t)_\tau\bigr]
:=
\sum_\varepsilon
\N_{q,n}\bigl[(k,\ell)_\varepsilon\bigm|(s,t)_\tau\bigr].
\]
The duality theorem in Section~\ref{sec-4} shows that these extension
counts depend on the fixed code $D$ only through
$n,q,(s,t)_\tau$.

For a fixed parent
code $C\subseteq\F_q^n$, we define
\[
d_q(s,t;C)
=
\max\left\{
d(D):
D\subseteq C,\ 
\dim(D)=s,\ 
\dim\Hull(D)=t
\right\}.
\] Thus $d_q(s,t;C)$ is the largest minimum distance among all
$[n,s]_q$ subcodes of $C$ with a $t$-dimensional hull.
This notation is used only when \(\mathcal{S}^t(s;C)\neq\varnothing\).
\subsection{Gram matrices and hull dimensions}
\label{subsec:gram-hull}

The following standard result characterizes the hull dimension of a
linear code in terms of the rank of the Gram matrix of its generator
matrix. It will be used repeatedly to compare hull dimensions under
puncturing and projection and to interpret the type classifications
below.

\begin{lem}[{\cite[Proposition 3.1]{Dcc-k}}]\label{GG^T}
Let $C$ be an $[n,k]_q$ linear code with generator matrix $G$. Then $C$
has an $\ell$-dimensional hull if and only if
\(
\operatorname{rank}(GG^T)=k-\ell.
\)
\end{lem}

\subsection{Types in even characteristic}
\label{subsec:even-types}

Assume that $q$ is even. For a linear code $C$, we use the following
complete type terminology:
\begin{itemize}
    \item $C$ is of type ${\bf oo}$ if $C$ is odd-like and $C^\perp$ is odd-like;
    \item $C$ is of type ${\bf oe}$ if $C$ is odd-like and $C^\perp$ is even-like;
    \item $C$ is of type ${\bf eo}$ if $C$ is even-like and $C^\perp$ is odd-like;
    \item $C$ is of type ${\bf ee}$ if $C$ is even-like and $C^\perp$ is even-like.
\end{itemize}
The four symbols above are the complete types in even characteristic.
The single symbols $\mathbf{o}$ and $\mathbf{e}$ will be used only as
parity labels, indicating that the code itself is odd-like or even-like,
respectively, without specifying the parity of its dual. Accordingly,
for $0\leq t\leq s$, we write
\[
\N^t_\mathbf{o}(s;C)
=
\sum_{\tau\in\{{\bf oo},{\bf oe}\}}
\N^t_\tau(s;C),
\quad
\N^t_\mathbf{e}(s;C)
=
\sum_{\tau\in\{{\bf eo},{\bf ee}\}}
\N^t_\tau(s;C).
\]
Thus $\N^t_\mathbf{o}(s;C)$ and
$\N^t_\mathbf{e}(s;C)$ count odd-like and even-like
subcodes, respectively. These two symbols are abbreviations for
parity-label counts, not complete-type counts. Similarly,
$C_{\mathbf{o}}^{k,\ell}$ and
$C_{\mathbf{e}}^{k,\ell}$ denote odd-like and even-like $[n,k]_q$
linear codes with an $\ell$-dimensional hull, respectively, with no
condition imposed on the parity of their dual codes unless this is
stated explicitly.
In parameterized formulas we use the corresponding type-summed
abbreviations
\[
\N_q\bigl[(s,t)_{\mathbf{o}};(k,\ell)_\varepsilon\bigr]
:=
\sum_{\tau\in\{{\bf oo},{\bf oe}\}}
\N_q\bigl[(s,t)_\tau;(k,\ell)_\varepsilon\bigr],
\quad
\N_q\bigl[(s,t)_{\mathbf{e}};(k,\ell)_\varepsilon\bigr]
:=
\sum_{\tau\in\{{\bf eo},{\bf ee}\}}
\N_q\bigl[(s,t)_\tau;(k,\ell)_\varepsilon\bigr].
\]

\begin{lem}[{\cite{Huffman}}]\label{1_n}
Assume that $C$ is an $[n,k]_q$ linear code and $q$ is even. Then
$C$ is even-like if and only if $\mathbf{1}_n\in C^\perp$.
\end{lem}

In particular, every self-orthogonal code over a field of even
characteristic is even-like.

The next lemma relates the even-like/odd-like distinction to canonical
Gram matrices and separates the nondegenerate part of a code from its
hull.

\begin{lem}[{\cite[Theorem 13]{FFA}}]\label{gmatrix}
Let $C$ be an $[n,k]_q$ linear code over a finite field of even
characteristic, and suppose that
$
\dim\Hull(C)=\ell.
$
Then the following hold.
\begin{enumerate}
    \item If $C$ is odd-like, then there exists a generator matrix $G$ of
    $C$ such that
    \[
    GG^T
    =
    \operatorname{diag}\{\underbrace{1,1,\ldots,1}_{k-\ell},0,\ldots,0\}.
    \]

    \item If $C$ is even-like, then $k-\ell$ is even, and there exists a
    generator matrix $G$ of $C$ such that
    \[
    GG^T
    =
    \operatorname{diag}\{\underbrace{J_2,J_2,\ldots,J_2}_{(k-\ell)/2},0,\ldots,0\},
    \qquad
    J_2=
    \begin{pmatrix}
    0&1\\
    1&0
    \end{pmatrix}.
    \]
\end{enumerate}
\end{lem}

Taking $\ell=0$ in Lemma~\ref{gmatrix}{\rm (ii)} shows, in particular,
that every even-like {\rm LCD} code has even dimension. 

\subsection{Types in odd characteristic}
\label{subsec:odd-types}

Throughout this subsection, $q$ is a power of an odd prime. We begin
with a standard fact concerning squares in finite fields.

\begin{lem}[{\cite{Finite fields}}]\label{minus-one-square}
The element $-1$ is a square in $\F_q^*$ if and only if
$q\equiv1\pmod4$.
\end{lem}

The following characterization will be used to define the type of a
linear code with a prescribed hull dimension.

\begin{prop}[{\cite[Theorem 14]{FFA}}]\label{hull-matrix}
Let $C$ be an $[n,k]_q$ linear code. Then $C$ has an
$\ell$-dimensional hull if and only if it has a generator matrix $G$
such that
\[
GG^T=\operatorname{diag}
(\underbrace{1,1,\ldots,\delta}_{k-\ell},0,\ldots,0),
\]
where $\ell<k$, $\delta=1$ or $\delta$ is a fixed nonsquare in $\F_q^*$.
\end{prop}

Accordingly, an $[n,k]_q$ linear code with an $\ell$-dimensional hull
is said to be of type $+$ if $\delta=1$, and of type $-$ if $\delta$
is a nonsquare in $\F_q^*$. If $\ell=k$, then $C$ is self-orthogonal,
and we define it to be of type $+$.

For odd characteristic, the two types are denoted throughout by
$+$ and $-$, and their corresponding counts are written as
$\N_+^t(s;C)$ and $\N_-^t(s;C)$, respectively. The same signs are
used in both representative-level and parameterized formulas.

\subsection{Puncturing maps}

Let $C$ be a linear code of length $n$, and let
\(
S\subseteq\{1,2,\ldots,n\}.
\)
The punctured code of $C$ on $S$, denoted by $C^S$, is obtained by deleting all coordinates indexed by $S$ from each codeword of $C$. Similarly, for a codeword $\mathbf{c}\in C$, we denote by $\mathbf{c}^S$ the vector obtained by deleting the coordinates indexed by $S$.

Let
\[
p_S:C\longrightarrow C^S,\qquad
p_S(\mathbf{c})=\mathbf{c}^S,
\]
be the puncturing map. Let $\mathcal{L}(C)$ denote the set of all subcodes of $C$. The map $p_S$ induces a surjective map
\[
f_S:\mathcal{L}(C)\longrightarrow\mathcal{L}(C^S),
\qquad
D\longmapsto D^S.
\]
If $\dim(C)=\dim(C^S)$, then $p_S$ is a linear isomorphism. Consequently, $f_S$ is bijective and
\(
\dim(D)=\dim(D^S)
\)
for every $D\in\mathcal{L}(C)$.

\begin{lem}\label{pun matrix}
Let $C$ be an $[n,k]_q$ linear code, and let $C^S$ be its punctured code on $S\subseteq\{1,2,\ldots,n\}$. Suppose that
\(
\langle\mathbf{c}_1,\mathbf{c}_2\rangle
=
\langle\mathbf{c}_1^S,\mathbf{c}_2^S\rangle
\)
for all $\mathbf{c}_1,\mathbf{c}_2\in C$ and that
\(
\dim(C)=\dim(C^S).
\)
Then, for every subcode $D\subseteq C$, we have
\[
\dim(\Hull(D))
=
\dim(\Hull(D^S)).
\]
In particular, $C$ has an $\ell$-dimensional hull if and only if $C^S$ has an $\ell$-dimensional hull.
\end{lem}

\begin{proof}
Since $\dim(C)=\dim(C^S)$, the puncturing map
$
p_S:C\longrightarrow C^S
$
is a linear isomorphism. Consequently, the induced map $f_S$ is bijective and preserves the dimension of every subcode.

Let $D$ be a subcode of $C$, and let $G_D$ be a generator matrix of $D$. Denote by $G_D^S$ the matrix obtained from $G_D$ by deleting the columns indexed by $S$. Since
$
\dim(D)=\dim(D^S),
$
the rows of $G_D^S$ are linearly independent. Hence, $G_D^S$ is a generator matrix of $D^S$.

By the assumption that puncturing preserves the Euclidean inner product on $C$, the inner product of any two rows of $G_D$ is equal to the inner product of their corresponding punctured rows. Therefore,
$
G_DG_D^T
=
G_D^S(G_D^S)^T.
$
It follows that
$
\operatorname{rank}(G_DG_D^T)
=
\operatorname{rank}\left(G_D^S(G_D^S)^T\right).
$
By Lemma \ref{GG^T}, we have
\[
\begin{aligned}
\dim(\Hull(D))
&=
\dim(D)-\operatorname{rank}(G_DG_D^T)\\
&=
\dim(D^S)
-
\operatorname{rank}\left(G_D^S(G_D^S)^T\right)\\
&=
\dim(\Hull(D^S)).
\end{aligned}
\]
Taking $D=C$ gives the final assertion.
\end{proof}

\subsection{Projection maps}

Suppose that a linear code $C$ admits a direct-sum decomposition
\(
C=V\oplus W,
\)
where $V$ and $W$ are subcodes of $C$. Then every codeword $\mathbf{c}\in C$ can be uniquely written as
\(
\mathbf{c}=\mathbf{v}+\mathbf{w},
\)
where $\mathbf{v}\in V$ and $\mathbf{w}\in W$. We define the projection mapping from $C$ onto $V$ along $W$ by
\[
\Pi_V:C\longrightarrow V,\qquad
\Pi_V(\mathbf{v}+\mathbf{w})=\mathbf{v}.
\]
Clearly,
\[
\ker(\Pi_V)=W
\qquad\text{and}\qquad
\operatorname{Im}(\Pi_V)=V.
\]
For any subcode $D\subseteq C$, the image
$
\Pi_V(D)
=
\left\{
\Pi_V(\mathbf{c})
\mid
\mathbf{c}\in D
\right\}
$
is a subcode of $V$, and the restriction of $\Pi_V$ to $D$ has kernel $D\cap W$. Therefore,
\[
\dim(D)
=
\dim(\Pi_V(D))
+
\dim(D\cap W).
\]

The following lemma is a projection-form version of the standard formula for counting subspaces with a prescribed intersection dimension with a fixed subspace; see, for example, {\cite[p.~162]{Oxley}} or {\cite[Lemma~34]{Gao}}. We give a self-contained proof because the projection formulation and its refinements will be used repeatedly.

\begin{lem}\label{proj}
Let $C$ be a linear code over $\F_q$ admitting a direct-sum decomposition
$C=V\oplus W$, where $V$ and $W$ are subcodes of $C$ and
$\dim(W)=\ell$. Let $s$ and $m$ be nonnegative integers satisfying
$0\leq m\leq\min\{s,\ell\}$, and let $U\subseteq V$ be an
$(s-m)$-dimensional subcode. Then the number of subcodes
$D\subseteq C$ satisfying $\Pi_V(D)=U$, $\dim(D)=s$, and
$\dim(D\cap W)=m$ is
\[
q^{(s-m)(\ell-m)}
\genfrac{[}{]}{0pt}{}{\ell}{m}_q.
\]
\end{lem}

\begin{proof}
We construct a bijection between the set of such subcodes $D$ and
the set of pairs $(K,f)$, where $K$ is an $m$-dimensional subcode of
$W$ and $f:U\longrightarrow W/K$ is a linear map.

Let $D\subseteq C$ satisfy $\Pi_V(D)=U$, $\dim(D)=s$, and
$\dim(D\cap W)=m$. Put $K=D\cap W$. Then $K$ is an
$m$-dimensional subcode of $W$.

For each $\mathbf{u}\in U$, there exists $\mathbf{w}\in W$ such that
$\mathbf{u}+\mathbf{w}\in D$. If both
$\mathbf{u}+\mathbf{w}_1$ and $\mathbf{u}+\mathbf{w}_2$ belong to
$D$, then $\mathbf{w}_1-\mathbf{w}_2\in D\cap W=K$. Hence
$\mathbf{w}_1+K=\mathbf{w}_2+K$, and we may define
\[
f:U\longrightarrow W/K,\qquad
f(\mathbf{u})=\mathbf{w}+K.
\]
The map $f$ is well-defined and linear.

Conversely, let $K\subseteq W$ be an $m$-dimensional subcode and let
$f:U\longrightarrow W/K$ be a linear map. Define
\[
D
=
\left\{
\mathbf{u}+\mathbf{w}
\mid
\mathbf{u}\in U,\ 
\mathbf{w}\in W,\ 
\mathbf{w}+K=f(\mathbf{u})
\right\}.
\]
Then $D$ is a subcode of $C$ satisfying
$\Pi_V(D)=U$, $D\cap W=K$, and
$\dim(D)=\dim(U)+\dim(K)=(s-m)+m=s$.
Thus, the required subcodes $D$ are in bijection with the pairs
$(K,f)$.

The number of $m$-dimensional subcodes $K$ of $W$ is
$\genfrac{[}{]}{0pt}{}{\ell}{m}_q$. For each fixed $K$, we have
$\dim(W/K)=\ell-m$. Since $\dim(U)=s-m$, the number of linear maps
from $U$ to $W/K$ is $q^{(s-m)(\ell-m)}$. Consequently, the number
of such subcodes $D$ is
\(
q^{(s-m)(\ell-m)}
\genfrac{[}{]}{0pt}{}{\ell}{m}_q.
\)
\end{proof}

The next two corollaries are refinements of Lemma~\ref{proj} obtained
by imposing an additional containment condition.

\begin{cor}\label{projee}
Under the assumptions of Lemma~\ref{proj}, suppose that
$0\leq m\leq\min\{s-1,\ell\}$, and let
$\bm\alpha\in U$ with $\bm\alpha\neq\mathbf{0}$. Then the number of
subcodes $D\subseteq C$ satisfying $\Pi_V(D)=U$,
$\dim(D)=s$, $\dim(D\cap W)=m$, and $\bm\alpha\in D$ is
\[
q^{(s-m-1)(\ell-m)}
\genfrac{[}{]}{0pt}{}{\ell}{m}_q.
\]
\end{cor}

\begin{proof}
As in the proof of Lemma~\ref{proj}, every such subcode $D$
corresponds to a pair $(K,f)$, where $K$ is an $m$-dimensional
subcode of $W$ and $f:U\longrightarrow W/K$ is a linear map.

Since $\bm\alpha\in U$, the condition $\bm\alpha\in D$ is
equivalent to $f(\bm\alpha)=0$. Thus, $f$ induces a linear map
\[
\overline{f}:
U/\Span\{\bm\alpha\}
\longrightarrow W/K.
\]
Since
$\dim(U/\Span\{\bm\alpha\})=s-m-1$ and
$\dim(W/K)=\ell-m$, the number of such maps for each fixed $K$ is
$q^{(s-m-1)(\ell-m)}$. The number of choices for $K$ is
$\genfrac{[}{]}{0pt}{}{\ell}{m}_q$. Hence, the number of the required subcodes is
\(
q^{(s-m-1)(\ell-m)}
\genfrac{[}{]}{0pt}{}{\ell}{m}_q.
\)
\end{proof}

\begin{cor}\label{projeee}
Under the assumptions of Lemma~\ref{proj}, suppose that
$1\leq m\leq\min\{s,\ell\}$, and let
$\bm\alpha\in W$ with $\bm\alpha\neq\mathbf{0}$. Then the number of
subcodes $D\subseteq C$ satisfying $\Pi_V(D)=U$,
$\dim(D)=s$, $\dim(D\cap W)=m$, and $\bm\alpha\in D$ is
\[
q^{(s-m)(\ell-m)}
\genfrac{[}{]}{0pt}{}{\ell-1}{m-1}_q.
\]
\end{cor}

\begin{proof}
As in the proof of Lemma~\ref{proj}, every such subcode $D$
corresponds to a pair $(K,f)$, where $K=D\cap W$ is an
$m$-dimensional subcode of $W$ and
$f:U\longrightarrow W/K$ is a linear map.

Since $\bm\alpha\in W$, the condition $\bm\alpha\in D$ is
equivalent to $\bm\alpha\in K$. Thus, $K$ must be an
$m$-dimensional subcode of $W$ containing
$\Span\{\bm\alpha\}$. The number of such subcodes is
$\genfrac{[}{]}{0pt}{}{\ell-1}{m-1}_q$, because they correspond to the
$(m-1)$-dimensional subspaces of
$W/\Span\{\bm\alpha\}$.

For each fixed $K$, we have $\dim(U)=s-m$ and
$\dim(W/K)=\ell-m$. Therefore, the number of linear maps from $U$
to $W/K$ is $q^{(s-m)(\ell-m)}$. Consequently, the number of the
required subcodes is
\(
q^{(s-m)(\ell-m)}
\genfrac{[}{]}{0pt}{}{\ell-1}{m-1}_q.
\)
\end{proof}

Having established the enumeration of subcodes with a prescribed projection, we now consider the behavior of the hull dimension under the projection mapping.

\begin{prop}\label{invariance}
Let $C$ be an $[n,k]_q$ linear code admitting a direct-sum decomposition $C=V\oplus W$, where $V$ and $W$ are subcodes of $C$. Suppose that $W$ is an $\ell$-dimensional self-orthogonal code and $V\subseteq W^\perp$. Then, for any subcode $D\subseteq C$, we have
\[
\dim(D)-\dim(\Hull(D))
=
\dim(\Pi_V(D))
-
\dim(\Hull(\Pi_V(D))).
\]
\end{prop}
\begin{proof}
Put $K=D\cap W$ and let $\dim K=r$. Choose a basis
$\mathbf{w}_1,\ldots,\mathbf{w}_r$ of $K$ and extend it to a basis
$\mathbf{w}_1,\ldots,\mathbf{w}_r,\bm\alpha_1,\ldots,\bm\alpha_m$
of $D$. Write
$
\bm\alpha_i=\mathbf{v}_i+\mathbf{u}_i
$
with $\mathbf{v}_i\in V$ and $\mathbf{u}_i\in W$. Since the kernel of
$\Pi_V|_{D}$ is $K$, the vectors
$\mathbf{v}_1,\ldots,\mathbf{v}_m$ form a basis of $\Pi_V(D)$.

Let $G_D$ be the generator matrix of $D$ with rows
$\mathbf{w}_1,\ldots,\mathbf{w}_r,\bm\alpha_1,\ldots,\bm\alpha_m$,
and let $G_{\rm proj}$ be the generator matrix of $\Pi_V(D)$ with rows
$\mathbf{v}_1,\ldots,\mathbf{v}_m$. Since $W$ is self-orthogonal and
$V\subseteq W^\perp$, all inner products involving a row from $K$ vanish,
and
$
\langle\bm\alpha_i,\bm\alpha_j\rangle
=
\langle\mathbf{v}_i,\mathbf{v}_j\rangle.
$
Hence
\[
G_DG_D^T
=
\begin{pmatrix}
O&O\\
O&G_{\rm proj}G_{\rm proj}^T
\end{pmatrix}.
\]
Thus
$
\operatorname{rank}(G_DG_D^T)
=
\operatorname{rank}(G_{\rm proj}G_{\rm proj}^T).
$
The desired equality now follows from Lemma~\ref{GG^T}.
\end{proof}

\begin{cor}\label{invariance-cor}
With the notation as in Proposition \ref{invariance}, for any subcode $D\subseteq C$, we have
\[
\dim(\Hull(D))
=
\dim(\Hull(\Pi_V(D)))
+
\dim(D\cap W).
\]
\end{cor}

\begin{proof}
Since $\ker(\Pi_V|_{D})=D\cap W$, the rank-nullity theorem gives
$
\dim(D)-\dim(\Pi_V(D))=\dim(D\cap W).
$
Combining this equality with Proposition~\ref{invariance} gives the
assertion.
\end{proof}

\begin{remarknonumber}
Let $D$ be an $s$-dimensional subcode of $C$ with a $t$-dimensional hull, and put
\(
m=\dim(D\cap W).
\)
By Corollary \ref{invariance-cor}, we have
\[
\dim(\Pi_V(D))=s-m
\text{\quad and \quad}
\dim(\Hull(\Pi_V(D)))=t-m.
\]
Therefore, the enumeration of the $s$-dimensional subcodes of $C$ with a $t$-dimensional hull can be reduced to the enumeration of the $(s-m)$-dimensional subcodes of $V$ with a $(t-m)$-dimensional hull. For each prescribed projected subcode, the corresponding subcodes of $C$ are counted by Lemma \ref{proj}.
\end{remarknonumber}

The next two lemmas record the type behavior of the projection map.

\begin{lem}\label{same-type}
Assume that $q$ is odd. Let $C=V\oplus W$, where
$V\subseteq W^\perp$ and $W$ is
self-orthogonal. Then, for every subcode $D\subseteq C$, the
projection $\Pi_V(D)$ of $D$ onto $V$ along $W$ has the same type
as $D$.
\end{lem}

\begin{proof}
Suppose that $D$ has dimension $s$ and a $t$-dimensional hull, and
that $\Pi_V(D)$ has dimension $s'$ and a $t'$-dimensional hull. By
Proposition \ref{invariance},
$
s-t=s'-t'.
$
If $s=t$, then $D$ is self-orthogonal. Proposition~\ref{invariance}
gives $s'=t'$, so $\Pi_V(D)$ is also self-orthogonal. Hence both
codes are of type $+$. We may therefore assume that $s>t$.

By Proposition \ref{hull-matrix}, $D$ has a generator matrix $G_D$
satisfying
\[
G_DG_D^T
=\operatorname{diag}
(\underbrace{1,1,\ldots,\delta}_{s-t},0,\ldots,0),
\]
where $\delta=1$ or $\delta$ is a fixed nonsquare in $\F_q^*$.
Let the rows of $G_D$ be
$\bm{\alpha}_1,\ldots,\bm{\alpha}_s$, and write
\[
\bm{\alpha}_i
=\widetilde{\bm{\alpha}}_i+\mathbf{w}_i,
\qquad
\widetilde{\bm{\alpha}}_i\in V,
\quad
\mathbf{w}_i\in W.
\]
Since $V\subseteq W^\perp$ and $W$ is self-orthogonal,
$
\left\langle
\widetilde{\bm{\alpha}}_i,
\widetilde{\bm{\alpha}}_j
\right\rangle
=
\left\langle
\bm{\alpha}_i,
\bm{\alpha}_j
\right\rangle
$
for all $i,j$. Hence
\begin{equation}\label{GPI}
\begin{pmatrix}
\widetilde{\bm{\alpha}}_1\\
\vdots\\
\widetilde{\bm{\alpha}}_s
\end{pmatrix}
\begin{pmatrix}
\widetilde{\bm{\alpha}}_1\\
\vdots\\
\widetilde{\bm{\alpha}}_s
\end{pmatrix}^{T}
=
\operatorname{diag}
(\underbrace{1,1,\ldots,\delta}_{s-t},
\underbrace{0,\ldots,0}_{t}).
\end{equation}
In particular,
$\widetilde{\bm{\alpha}}_1,\ldots,
\widetilde{\bm{\alpha}}_{s-t}$ are linearly independent.
Since $s'-(s-t)=t'$, they can be extended to a basis of
$\Pi_V(D)$ by choosing $t'$ vectors among the remaining projected
vectors. Equation \eqref{GPI} then shows that this basis has the same
nonzero diagonal part as $G_DG_D^T$. Thus $\Pi_V(D)$ and $D$ have
the same type.
\end{proof}

\begin{lem}\label{even-parity-proj}
Assume that $q$ is even. Let $C=V\oplus W$, where
$V\subseteq W^\perp$ and $W$ is self-orthogonal. Then, for every
subcode $D\subseteq C$, the projection $\Pi_V(D)$ of $D$ onto $V$
along $W$ has the same parity label as $D$. That is, $D$ is odd-like,
respectively even-like, if and only if $\Pi_V(D)$ is odd-like,
respectively even-like.
\end{lem}

\begin{proof}
Since $W$ is self-orthogonal and $q$ is even, every vector in $W$ is
even-like. For any $\mathbf{x}\in D$, write
$
\mathbf{x}=\mathbf{v}+\mathbf{w},
$
where $\mathbf{v}=\Pi_V(\mathbf{x})\in\Pi_V(D)$ and $\mathbf{w}\in W$.
Since $\mathbf{w}$ is even-like, $\mathbf{x}$ and $\mathbf{v}$ have the
same coordinate sum. Hence all codewords of $D$ are even-like if and
only if all codewords of $\Pi_V(D)$ are even-like. The odd-like
assertion follows immediately.
\end{proof}
\subsection{LCD complements of hulls}
\label{subsec:lcd-complements}

We now construct the {\rm LCD} complements used to reduce enumeration
inside an arbitrary fixed code to enumeration inside an {\rm LCD} code.
Lemma~\ref{gmatrix} supplies the required canonical forms in even
characteristic, while Proposition~\ref{hull-matrix} supplies their
counterparts in odd characteristic. We begin with the complete-type
refinement needed for a parent code of type ${\bf oe}$.

\begin{lem}\label{exist-oeLCD}
Let $C$ be an $[n,k]_q$ linear code with an $\ell$-dimensional hull and suppose that $C$ has type ${\bf oe}$. Then there exists a $(k-\ell)$-dimensional {\rm LCD} subcode ${\rm L}_{\bf oe}^{k-\ell}$ such that $C=\Hull(C)\oplus {\rm L}_{\bf oe}^{k-\ell}$.
\end{lem}
\begin{proof}
Since $C^\perp$ is even-like,
Lemma~\ref{1_n} gives $\mathbf{1}_n\in C$. Since $C$ is odd-like,
we have $\mathbf{1}_n\notin C^\perp$, and hence
$\mathbf{1}_n\notin \Hull(C)$. Therefore,
$
\dim\bigl(\Hull(C)\oplus\Span\{\mathbf{1}_n\}\bigr)
=\ell+1.
$
We may choose $\bm\alpha_2,\ldots,\bm\alpha_{k-\ell}\in C$
such that
\[
C
=
\Hull(C)\oplus
\Span
\{\mathbf{1}_n,\bm\alpha_2,\ldots,\bm\alpha_{k-\ell}\}.
\]
Put
$
{\rm L}=\Span
\{\mathbf{1}_n,\bm\alpha_2,\ldots,\bm\alpha_{k-\ell}\}.
$
Then $\dim({\rm L})=k-\ell$ and $C=\Hull(C)\oplus {\rm L}$.

Let $\bm\beta_1,\ldots,\bm\beta_\ell$ be a basis of $\Hull(C)$, and
let $G_{\rm L}$ and $G_H$ be the matrices whose rows are
$\mathbf{1}_n,\bm\alpha_2,\ldots,\bm\alpha_{k-\ell}$ and
$\bm\beta_1,\ldots,\bm\beta_\ell$, respectively. Then
$
G=\begin{pmatrix}G_{\rm L}^T&G_H^T\end{pmatrix}^T
$
is a generator matrix of $C$. Since $\Hull(C)\subseteq C^\perp$, we have
\[
GG^T
=
\begin{pmatrix}
G_{\rm L}G_{\rm L}^T&O\\
O&O_\ell
\end{pmatrix}.
\]
By Lemma~\ref{GG^T},
$
\operatorname{rank}(GG^T)=k-\ell.
$
Hence
$
\operatorname{rank}(G_{\rm L}G_{\rm L}^T)=k-\ell,
$
so $G_{\rm L}G_{\rm L}^T$ is nonsingular. Applying Lemma~\ref{GG^T} to ${\rm L}$
shows that ${\rm L}$ is {\rm LCD}.

Since $\mathbf{1}_n\in {\rm L}$, Lemma~\ref{1_n} shows that ${\rm L}^\perp$
is even-like. Moreover, ${\rm L}$ cannot be even-like; otherwise
$C$ is even-like. Thus ${\rm L}$ is odd-like and has type
${\bf oe}$. Taking
$
{\rm L}_{{\bf oe}}^{k-\ell}={\rm L}
$
completes the proof.
\end{proof}

\begin{prop}\label{lcd-complement-types}
Let $C$ be an $[n,k]_q$ linear code with an $\ell$-dimensional hull.
Then the following statements hold.
\begin{enumerate}
\item Suppose that $q$ is even. If $C$ is of type ${\bf oo}$,
${\bf oe}$, or ${\bf eo}$, then there exists an {\rm LCD}
subcode ${\rm L}$ such that
\[
C={\rm L}\oplus \Hull(C),
\]
and ${\rm L}$ has the same type as $C$. If $C$ is of type ${\bf ee}$,
then there exists an {\rm LCD} subcode ${\rm L}$ of type ${\bf eo}$ such
that
\[
C={\rm L}\oplus \Hull(C),
\]
and in this case $\mathbf{1}_n\in\Hull(C)$.

\item Suppose that $q$ is odd. If $C$ is of type $+$ or type $-$, then
there exists an {\rm LCD} subcode ${\rm L}$ such that
\[
C={\rm L}\oplus \Hull(C),
\]
and ${\rm L}$ has the same type as $C$ whenever $k>\ell$. If $k=\ell$, then
we take ${\rm L}=\{\mathbf{0}\}$.
\end{enumerate}
\end{prop}

\begin{proof}
Assume first that $q$ is even. By Lemma~\ref{gmatrix}, there exists a
$(k-\ell)$-dimensional {\rm LCD} subcode ${\rm L}$ such that
$C={\rm L}\oplus\Hull(C)$, and ${\rm L}$ has the same parity as $C$.
If $C$ is of type ${\bf oe}$, the stronger assertion follows from
Lemma~\ref{exist-oeLCD}. If $C$ is of type ${\bf oo}$, then
$\mathbf{1}_n\notin C$, and hence $\mathbf{1}_n\notin {\rm L}$; by
Lemma~\ref{1_n}, the dual ${\rm L}^\perp$ is odd-like. Thus ${\rm L}$ is of type
${\bf oo}$. If $C$ is even-like, then ${\rm L}$ is even-like. Since an
even-like {\rm LCD} code cannot have an even-like dual, ${\rm L}$ is of type
${\bf eo}$. This proves the assertion for types ${\bf eo}$ and
${\bf ee}$. Moreover, if $C$ is of type ${\bf ee}$, then
Lemma~\ref{1_n} applied to both $C$ and $C^\perp$ gives
$\mathbf{1}_n\in C^\perp$ and $\mathbf{1}_n\in C$, so
$\mathbf{1}_n\in\Hull(C)$.

Now assume that $q$ is odd. If $k=\ell$, then $C$ is self-orthogonal and
we take ${\rm L}=\{\mathbf{0}\}$. If $k>\ell$, Proposition~\ref{hull-matrix}
gives a generator matrix whose Gram matrix has nonzero part
$\operatorname{diag}(1,\ldots,1,\delta)$ and zero part of size $\ell$.
The rows corresponding to the nonzero part span an {\rm LCD} subcode
${\rm L}$, while the remaining rows span $\Hull(C)$. Hence
$C={\rm L}\oplus\Hull(C)$, and the value of $\delta$ shows that
${\rm L}$ has the same type as $C$.
\end{proof}

\section{Enumeration of subcodes with prescribed hull dimensions in even characteristic}\label{sec-2}

Throughout this section, $q$ has even characteristic. We use the type
notation and the LCD-complement decompositions introduced in
Section~\ref{sec-1}.
The formulas established below show that parent codes with the same
$q$ and the same parameters $(k,\ell)_\varepsilon$ contain the same
number of subcodes with any prescribed admissible parameters
$(s,t)_\tau$. Thus the principal results are stated using the
parameterized notation introduced in Section~\ref{sec-1}; specific
representatives are retained only in the intermediate reductions.

\subsection{Enumeration of subcodes with prescribed hull dimensions in {\rm LCD} codes}

We begin with the whole-space formulas that serve as the initial values
for the recursive subcode counts. The {\rm LCD} case is due to Carlet
\emph{et al.}, and the prescribed-hull extension is due to Li
\emph{et al.}

\begin{lem}[{\cite[Theorem 15]{binaryLCD}}]\label{lem:even-lcd-mass}
Assume that $q$ is even. Let $k$ and $n$ be positive integers with $k<n$. Then
    \begin{enumerate}
        \item $\N^0_{\bf oo}(k;\F_q^n)= \begin{cases}
 q^{\frac{nk-k^2+n-1}{2}} {\frac{n}{2}-1 \brack \frac{k-1}{2}}_{q^2}, &\text{if $k$ is odd and $n$ is even};
 \\  q^{\frac{(n-k)(k-1)}{2}}   (q^{n-k}-1) {\frac{n-1}{2} \brack \frac{k-1}{2}}_{q^2}, &\text{if $k$ is odd and $n$ is odd};
 \\  q^{\frac{k(n-k-1)}{2}}   (q^{k}-1) {\frac{n-1}{2} \brack \frac{k}{2}}_{q^2}, &\text{if $k$ is even and $n$ is odd};
 \\  q^{\frac{k(n-k)}{2}}   (q^{k}-1) {\frac{n}{2}-1 \brack \frac{k}{2}}_{q^2}, &\text{if $k$ is even and $n$ is even}.
  \end{cases}$
  \item $\N^0_{\bf oe}(k;\F_q^n)= \begin{cases}
 q^{\frac{(k-1)(n-k)}{2}} {\frac{n-1}{2} \brack \frac{k-1}{2}}_{q^2}, &\text{if $k$ is odd and $n$ is odd};
 \\  q^{\frac{k(n-k)}{2}}   {\frac{n}{2}-1 \brack \frac{k}{2}-1}_{q^2}, &\text{if $k$ is even and $n$ is even};
\\  0, & \text{otherwise.}
 \end{cases}$
 \item $\N^0_{\bf eo}(k;\F_q^n)= \begin{cases}
 q^{\frac{k(n-k-1)}{2}} {\frac{n-1}{2} \brack \frac{k}{2}}_{q^2}, &\text{if $k$ is even and $n$ is odd};
 \\  q^{\frac{k(n-k)}{2}}   {\frac{n}{2}-1 \brack \frac{k}{2}}_{q^2}, &\text{if $k$ is even and $n$ is even};
\\  0, & \text{otherwise.}
 \end{cases}$
    \end{enumerate}
\end{lem}

\begin{lemma}[{\cite[Theorem 5.10]{Li-hull-even}}]\label{thm:even-whole-space-mass}
Let $\ell,k$ and $n$ be positive integers such that $\ell\leq k\leq n-\ell$. Put $k_0=k-\ell$ and
\[
A_{n,k,\ell}=\frac{q^{n-k}-q^{\ell-1}}{(q^\ell-1)q^{k-1}},
\qquad
B_{n,k,\ell}=\frac{q^{n-k}-q^\ell}{(q^\ell-1)q^{k-1}}.
\]
Then
\begin{enumerate}
    \item
    $
    \N_{\bf oo}^{\ell}(k;\F_q^n)=
    \begin{cases}
    \left(\prod_{i=1}^{\ell}A_{n,k_0+i,i}\right)
    \N^0_{\bf oo}(k_0;\F_q^n),&\text{if $n-k_0$ is odd},\\[1mm]
    \left(\prod_{i=1}^{\ell}B_{n,k_0+i,i}\right)
    \N^0_{\bf oo}(k_0;\F_q^n),&\text{if $n-k_0$ is even};
    \end{cases}
    $
    \item
    $
    \N_{\bf oe}^{\ell}(k;\F_q^n)=
    \begin{cases}
    \left(\dfrac{\prod_{i=1}^{\ell-1}B_{n,k_0+i,i}}
    {(q^\ell-1)q^{k-\ell-1}}\right)
    \N^0_{\bf oo}(k_0;\F_q^n),&\text{if $n-k_0$ is even},\\[2mm]
    0,&\text{otherwise};
    \end{cases}
    $
    \item
    $
    \N_{\bf eo}^{\ell}(k;\F_q^n)=
    \begin{cases}
    \left(\prod_{i=1}^{\ell}A_{n,k_0+i,i}\right)
    \N^0_{\bf eo}(k_0;\F_q^n),&\text{if $n$ is odd and $k_0$ is even},\\[1mm]
    \left(\prod_{i=1}^{\ell}B_{n,k_0+i,i}\right)
    \N^0_{\bf eo}(k_0;\F_q^n),&\text{if $n$ is even and $k_0$ is even},\\[1mm]
    0,&\text{otherwise};
    \end{cases}
     $
    \item
    $
    \N_{\bf ee}^{\ell}(k;\F_q^n)=
    \begin{cases}
    \left(\dfrac{\prod_{i=1}^{\ell-1}B_{n,k_0+i,i}}{q^{k-1}}\right)
    \N^0_{\bf eo}(k_0;\F_q^n),&\text{if $n$ is even and $k_0$ is even},\\[2mm]
    0,&\text{otherwise}.
    \end{cases}
    $
\end{enumerate}
\end{lemma}

For finite fields of even characteristic, Xiong and Liu \cite{x-lcd}  established mass formulas for {\rm LCD} subcodes of {\rm LCD} codes. We extend their results to subcodes with prescribed hull dimensions in {\rm LCD} codes. To this end, we first present several preliminary results.

Let $G^{k,n}_{\bf oo}$, $G^{k,n}_{\bf oe}$ and $G^{k,n}_{\bf eo}$ be the following $k\times n$ matrices:
\begin{align}\label{even}
G^{k,n}_{\bf oo}
=
\left(\begin{array}{c}
\mathbf{e}_{1}\\
\mathbf{e}_{2}\\
\vdots\\
\mathbf{e}_{k}
\end{array}\right),
\qquad
G^{k,n}_{\bf oe}
=
\left(\begin{array}{c}
\mathbf{e}_{1}\\
\vdots\\
\mathbf{e}_{k-1}\\
\sum_{i=k}^{n}\mathbf{e}_{i}
\end{array}\right),
\qquad
G^{k,n}_{\bf eo}
=
\left(\begin{array}{c}
\mathbf{e}_{1}+\mathbf{e}_{2}\\
\mathbf{e}_{2}+\mathbf{e}_{3}\\
\vdots\\
\mathbf{e}_{k}+\mathbf{e}_{k+1}
\end{array}\right),
\end{align}
where $\mathbf{e}_i=(0,\ldots,0,1,0,\ldots,0)$ denotes the vector of length $n$ with $1$ in the $i$-th coordinate. Let ${\rm L}^{k}_{\bf oo}$, ${\rm L}^{k}_{\bf oe}$ and ${\rm L}^{k}_{\bf eo}$ be the $[n,k]_q$ linear codes generated by $G^{k,n}_{\bf oo}$, $G^{k,n}_{\bf oe}$ and $G^{k,n}_{\bf eo}$, respectively. Then
\[
{\rm L}^{k}_{\bf oo}\in\mathcal{S}^0_{\bf oo}(k;\F_q^n),\qquad
{\rm L}^{k}_{\bf oe}\in\mathcal{S}^0_{\bf oe}(k;\F_q^n)
\quad\text{if $n-k$ is even},
\]
and
\[
{\rm L}^{k}_{\bf eo}\in\mathcal{S}^0_{\bf eo}(k;\F_q^n)
\quad\text{if $k$ is even}.
\]

\begin{prop}[{\cite[Theorem 11]{binaryLCD}}]\label{even-form}
Let $C_1$ and $C_2$ be two $[n,k]_q$ {\rm LCD} codes of the same type. Then there exists an orthogonal matrix $Q\in\mathbb{O}_n(q)$ such that $C_2=C_1Q$. Conversely, for any $[n,k]_q$ {\rm LCD} code $C$ and any $Q\in\mathbb{O}_n(q)$, the code $CQ$ is also {\rm LCD} and has the same type as $C$.
\end{prop}

\begin{lem}\label{same-type-count}
Let $C$ and $\widetilde{C}$ be two $[n,k]_q$ {\rm LCD} codes of the same type.
Then, for all nonnegative integers $s$, $t$ and $\tau \in \{\bf oo,\bf oe,\bf eo,\bf ee\}$,
$$\N_{\tau}^t(s;C)=\N^t_{\tau}(s;\widetilde{C}).$$
\end{lem}

\begin{proof}
By Proposition \ref{even-form}, there exists an orthogonal matrix
$Q\in \mathbb{O}_n(q)$ such that
$
\widetilde C=CQ.
$

Let $D\subseteq C$, and put $\widetilde D=DQ\subseteq\widetilde C$.
If $G_D$ is a generator matrix of $D$, then $G_DQ$ is a generator
matrix of $\widetilde D$. Since $Q Q^T=E_n$, we have
\[
G_DG_D^T=G_DQQ^TG_D^T.
\]
Thus $D$ is even-like if and only if $\widetilde D$ is even-like.
Moreover, by Lemma~\ref{GG^T},
\[
\dim\Hull(\widetilde D)
=
\dim\Hull(D).
\]
Since $(DQ)^\perp=D^\perp Q$, the same argument applied to $D^\perp$
shows that $D^\perp$ is even-like if and only if
$\widetilde D^\perp$ is even-like. Hence the complete type of $D$ is
the same as that of $\widetilde D$.

Consequently, the map
\[
D\longmapsto DQ
\]
gives a bijection between the $s$-dimensional subcodes of $C$ with
$t$-dimensional hull and any prescribed type, and the corresponding
subcodes of $\widetilde C$. This proves the four equalities.
\end{proof}

The preceding invariance result reduces the enumeration inside an {\rm LCD}
parent code to a convenient standard representative. We first treat
the case in which puncturing preserves all four complete types.

\begin{proposition}\label{oe number}
Let ${\rm L}^{k}_{\bf oe}$ be the $[n,k]_q$ {\rm LCD} code of type ${\bf oe}$. Let $s$ and $t$ be nonnegative integers with $0\leq t\leq s<k\leq n$. Then for all $\tau \in \{\bf oo,\bf oe,\bf eo,\bf ee\}$, we have  \(\N^t_{\tau}(s;{\rm L}^{k}_{\bf oe})=\N^t_{\tau}(s;\F_q^k).\)
\end{proposition}
\begin{proof}
By Proposition \ref{even-form} and Lemma \ref{same-type-count}, it is
enough to prove the assertion for the standard representative
${\rm L}_{\bf oe}^{k}$.

Let
$
S=\{k+1,k+2,\ldots,n\},
$
and denote by $({\rm L}_{\bf oe}^{k})^S$ the punctured code obtained by
deleting the coordinates in $S$. Then
$
({\rm L}_{\bf oe}^{k})^S=\F_q^k.
$

For every codeword
$
\mathbf{x}=(x_1,\ldots,x_n)\in {\rm L}_{\bf oe}^{k},
$
the last $n-k+1$ coordinates are equal. Since $n-k$ is even, it follows
that, for all $\mathbf{c}_1,\mathbf{c}_2\in {\rm L}_{\bf oe}^{k}$,
$
\langle \mathbf{c}_1,\mathbf{c}_2\rangle
=
\langle \mathbf{c}_1^S,\mathbf{c}_2^S\rangle .
$
Moreover,
$
\dim {\rm L}_{\bf oe}^{k}
=
\dim ({\rm L}_{\bf oe}^{k})^S.
$
Hence, by Lemma \ref{pun matrix}, puncturing on $S$ preserves the hull
dimension of every subcode. More precisely, for every subcode
$D\subseteq {\rm L}_{\bf oe}^{k}$,
$
\dim\Hull(D)
=
\dim\Hull(D^S).
$
Therefore
\[
D\in \mathcal{S}^t(s;{\rm L}_{\bf oe}^{k})
\Longleftrightarrow
D^S\in \mathcal{S}^t(s;\F_q^k).
\]

It remains to check that the type of a subcode is also preserved. Since
the puncturing map preserves the inner product, a subcode $D$ is
even-like if and only if its punctured code $D^S$ is even-like.

Furthermore, because the last $n-k+1$ coordinates of every codeword in
${\rm L}_{\bf oe}^{k}$ are identical, we have
\[
\mathbf{1}_n\in D
\Longleftrightarrow
\mathbf{1}_k\in D^S.
\]
By Lemma \ref{1_n}, this is equivalent to
\[
D^\perp \text{ is even-like}
\Longleftrightarrow
(D^S)^\perp \text{ is even-like}.
\]
Thus puncturing on $S$ preserves both the parity of the subcode and the
parity of its dual.

Consequently, the map
$
D\longmapsto D^S
$
gives a bijection between the $s$-dimensional subcodes of
${\rm L}_{\bf oe}^{k}$ with $t$-dimensional hull and any prescribed type,
and the corresponding subcodes of $\F_q^k$. 
This proves all four identities.
\end{proof}

The type-${\bf oo}$ case has the same puncturing reduction, but the
absence of $\mathbf{1}_n$ forces every subcode dual to be odd-like.

\begin{proposition}\label{oo number}
Let ${\rm L}^{k}_{\bf oo}$ be the $[n,k]_q$ {\rm LCD} code of type ${\bf oo}$. Let $s$ and $t$ be nonnegative integers with $0\leq t\leq s<k\leq n$. Then, 
\begin{enumerate}
    \item $\displaystyle\N^t_{\bf oo}(s;{\rm L}^{k}_{\bf oo})=\N^t_{\bf oo}(s;\F_q^k)+\N^t_{\bf oe}(s;\F_q^k);$
   
    \item $\displaystyle\N^t_{\bf eo}(s;{\rm L}^{k}_{\bf oo})=\N^t_{\bf eo}(s;\F_q^k)+\N^t_{\bf ee}(s;\F_q^k);$
    \item $\displaystyle\N^t_{\bf oe}(s;{\rm L}^{k}_{\bf oo})=\N^t_{\bf ee}(s;{\rm L}^{k}_{\bf oo})=0$
    .
\end{enumerate}
\end{proposition}
\begin{proof}
By Proposition \ref{even-form} and Lemma \ref{same-type-count}, it
suffices to consider the standard representative
${\rm L}^{k}_{\bf oo}$.
Puncturing on $S=\{k+1,\ldots,n\}$ gives $\F_q^k$ and preserves
inner products, dimensions, hull dimensions, and the even-like or
odd-like property of subcodes, by Lemma \ref{pun matrix}.

It remains only to note the difference from Proposition \ref{oe number}.
Since ${\rm L}_{\bf oo}^{k}$ has odd-like dual, Lemma \ref{1_n}
implies $\mathbf{1}_n\notin {\rm L}_{\bf oo}^{k}$. Hence no subcode contains
$\mathbf{1}_n$, and so every subcode has odd-like dual. Thus subcodes of
types ${\bf oe}$ and ${\bf ee}$ do not occur. Under the above
puncturing map, odd-like subcodes correspond to all odd-like subcodes
of $\F_q^k$, while even-like subcodes correspond to all
even-like subcodes of $\F_q^k$. This completes the proof of the three assertions.
\end{proof}

\begin{cor}\label{oe=oo}
Let ${\rm L}_{{\bf oe}}^{k}$ and ${\rm L}_{{\bf oo}}^{k}$ be
$[n,k]_q$ {\rm LCD} codes of types ${\bf oe}$ and
${\bf oo}$, respectively. Then, for all nonnegative integers
$s$ and $t$ satisfying $0\leq t\leq s\leq k$, we have
\[
\N_q\bigl[(s,t);(k,0)_{{\bf oe}}\bigr]
=
\N_q\bigl[(s,t);(k,0)_{{\bf oo}}\bigr].
\]
Moreover, for $\rho\in\{\mathbf{o},\mathbf{e}\}$,
\[
\N_q\bigl[(s,t)_\rho;(k,0)_{{\bf oe}}\bigr]
=
\N_q\bigl[(s,t)_\rho;(k,0)_{{\bf oo}}\bigr].
\]
\end{cor}

\begin{proof}
This follows from Propositions~\ref{oe number} and~\ref{oo number} by
summing the corresponding complete types with the same parity label.
\end{proof}

\begin{proposition}\label{eo number}
Let ${\rm L}^{k}_{\bf eo}$ be the $[n,k]_q$ {\rm LCD} code of type ${\bf eo}$.  Let $s$ and $t$ be nonnegative integers with $0\leq t\leq s<k$. Then
\begin{enumerate}
    \item
    $\displaystyle\N^t_{\bf eo}(s;{\rm L}^{k}_{\bf eo})=\N^t_\mathbf{e}(s;\F_q^{k+1});$
    \item
    $\displaystyle\N^t_{\bf oo}(s;{\rm L}^{k}_{\bf eo})=\N^t_{\bf oe}(s;{\rm L}^{k}_{\bf eo}) = \N^t_{\bf ee}(s;{\rm L}^{k}_{\bf eo})=0.
    $
\end{enumerate}
\end{proposition}
\begin{proof}
By Proposition \ref{even-form} and Lemma \ref{same-type-count}, it
suffices to consider the standard representative ${\rm L}^{k}_{\bf eo}$.
Let $S=\{k+2,\ldots,n\}$. Then the punctured code
$({\rm L}_{\bf eo}^{k})^S$ is
\[
E=\left\{(x_1,\ldots,x_{k+1})\in \F_q^{k+1}\mid
x_1+\cdots+x_{k+1}=0\right\}.
\]
Moreover, this puncturing preserves inner products and dimensions.
Hence, by Lemma \ref{pun matrix}, subcodes with a prescribed hull
dimension are preserved under this puncturing.

Since $E$ is precisely the set of all even-like vectors in
$\F_q^{k+1}$, every subcode of $E$ is even-like, and every
even-like subcode of $\F_q^{k+1}$ is contained in $E$. Therefore
the $s$-dimensional subcodes of ${\rm L}_{\bf eo}^{k}$ with
$t$-dimensional hull correspond exactly to the $s$-dimensional
even-like subcodes of $\F_q^{k+1}$ with $t$-dimensional hull.
Thus
\[
\N^t_{\bf eo}(s;{\rm L}_{\bf eo}^{k})
=
\N^t_\mathbf{e}(s;\F_q^{k+1}).
\]

It remains to determine the type of the duals. Since
$ {\rm L}_{\bf eo}^{k}$ is of type ${\bf eo}$, its dual is odd-like.
By Lemma \ref{1_n}, this implies
\(
\mathbf{1}_n\notin {\rm L}_{\bf eo}^{k}.
\)
Hence no subcode of ${\rm L}_{\bf eo}^{k}$ contains $\mathbf{1}_n$, and
therefore every subcode has odd-like dual, again by Lemma
\ref{1_n}. Consequently, no subcodes of types ${\bf oo}$,
${\bf oe}$, or ${\bf ee}$ occur. Hence
\[
\N^t_{\bf oo}(s;{\rm L}_{\bf eo}^{k})
=
\N^t_{\bf oe}(s;{\rm L}_{\bf eo}^{k})
=
\N^t_{\bf ee}(s;{\rm L}_{\bf eo}^{k})
=
0.
\]
This completes the proof.
\end{proof}

\subsection{Enumeration of subcodes with prescribed hull dimensions in linear codes}
In this subsection, we systematically determine the number of
subcodes with a prescribed hull dimension inside a given linear code.
Our general strategy proceeds in two phases: first, applying the
decomposition $C={\rm L}\oplus W$ from Lemma~\ref{gmatrix} to project the
subcode onto the {\rm LCD} component ${\rm L}$; second, enumerating the
possible extensions back into the hull $W$ by using the projection
techniques developed in Section~\ref{sec-1}.

\begin{lem}\label{even-projection-count}
Assume that $q$ is even. Let $C=V\oplus W$, where
$V\subseteq W^\perp$, $W$ is self-orthogonal, and $\dim W=\ell$.
Then, for $0\leq t\leq s$, one has
\[
\N^t(s;C)
=
\sum_{m=0}^{\min\{t,\ell\}}
q^{(s-m)(\ell-m)}
\genfrac{[}{]}{0pt}{}{\ell}{m}_q
\N^{t-m}(s-m;V).
\]
Moreover, for each parity label
$\epsilon\in\{\mathbf{o},\mathbf{e}\}$,
\[
\N^t_\epsilon(s;C)
=
\sum_{m=0}^{\min\{t,\ell\}}
q^{(s-m)(\ell-m)}
\genfrac{[}{]}{0pt}{}{\ell}{m}_q
\N^{t-m}_\epsilon(s-m;V).
\]
\end{lem}

\begin{proof}
Let $D\subseteq C$ be an $s$-dimensional subcode with a
$t$-dimensional hull, and put $m=\dim(D\cap W)$. Since
$W\subseteq C^\perp$, we have $D\cap W\subseteq\Hull(D)$,
so $0\leq m\leq\min\{t,\ell\}$. By Corollary~\ref{invariance-cor},
the projected subcode $\Pi_V(D)$ has dimension $s-m$ and hull
dimension $t-m$.

Conversely, fix an $(s-m)$-dimensional subcode $U\subseteq V$ with
hull dimension $t-m$. Lemma~\ref{proj} gives exactly
$
q^{(s-m)(\ell-m)}\genfrac{[}{]}{0pt}{}{\ell}{m}_q
$
subcodes $D\subseteq C$ such that $\Pi_V(D)=U$ and
$\dim(D\cap W)=m$. For every such $D$, Corollary~\ref{invariance-cor}
gives $\dim\Hull(D)=t$. Summing over all admissible
$m$ and all possible $U$ gives the first formula. The second formula
follows from the same argument together with Lemma~\ref{even-parity-proj}.
\end{proof}
The  following lemma is a direct consequence of Lemma~\ref{1_n}.
\begin{lem}\label{even-type-bookkeeping}
Assume that $q$ is even, and let $D\subseteq C$ be a subcode. Then
\begin{enumerate}
\item if $C$ is even-like, then $D$ is even-like;
\item if $C^\perp$ is odd-like, then $D^\perp$ is odd-like;
\item $D^\perp$ is even-like if and only if $\mathbf{1}_n\in D$.
\end{enumerate}
\end{lem}

\begin{theorem}\label{all-hull}
Assume that $q$ is even. Let $n,k,\ell,s$, and $t$ be
nonnegative integers such that
\[
\ell\leq k\leq n-\ell,
\qquad
0\leq t\leq s\leq k.
\]
Then the following statements hold.
\begin{enumerate}
    \item For every admissible
    $\varepsilon\in\{{\bf oo},{\bf oe}\}$,
    \begin{enumerate}
    \item $\displaystyle
    \N_q\bigl[(s,t);(k,\ell)_\varepsilon\bigr]
    =
    \sum_{m=0}^{\min\{t,\ell\}}
    q^{(s-m)(\ell-m)}
    \genfrac{[}{]}{0pt}{}{\ell}{m}_q
    \N_q\bigl[(s-m,t-m);(k-\ell,0)_\varepsilon\bigr];
   $
   \item $\displaystyle
    \N_q\bigl[(s,t)_{\mathbf{o}};(k,\ell)_\varepsilon\bigr]
    =
    \sum_{m=0}^{\min\{t,\ell\}}
    q^{(s-m)(\ell-m)}
    \genfrac{[}{]}{0pt}{}{\ell}{m}_q
    \N_q\bigl[(s-m,t-m)_{\mathbf{o}};
    (k-\ell,0)_\varepsilon\bigr];
   $
   \item $\displaystyle
    \N_q\bigl[(s,t)_{\mathbf{e}};(k,\ell)_\varepsilon\bigr]
    =
    \sum_{m=0}^{\min\{t,\ell\}}
    q^{(s-m)(\ell-m)}
    \genfrac{[}{]}{0pt}{}{\ell}{m}_q
    \N_q\bigl[(s-m,t-m)_{\mathbf{e}};
    (k-\ell,0)_\varepsilon\bigr].
   $
    \end{enumerate}
    \item For every admissible
    $\varepsilon\in\{{\bf eo},{\bf ee}\}$,
    \[
    \N_q\bigl[(s,t)_{\mathbf{e}};(k,\ell)_\varepsilon\bigr]
    =\N_q\bigl[(s,t);(k,\ell)_\varepsilon\bigr]
    =
    \sum_{m=0}^{\min\{t,\ell\}}
    q^{(s-m)(\ell-m)}
    \genfrac{[}{]}{0pt}{}{\ell}{m}_q
    \N_q\bigl[(s-m,t-m);(k-\ell,0)_{{\bf eo}}\bigr].
    \]
\end{enumerate}
A summand is understood to be zero if the indicated subcode does not
exist.
\end{theorem}

\begin{proof}
By the remark following Proposition~\ref{lcd-complement-types}, every $[n,k]_q$ linear
code $C^{k,\ell}$ with an $\ell$-dimensional hull admits a decomposition
\[
C^{k,\ell}
=
{\rm L}^{k-\ell}\oplus\Hull(C^{k,\ell}),
\qquad
{\rm L}^{k-\ell}\subseteq\Hull(C^{k,\ell})^\perp,
\]
where ${\rm L}^{k-\ell}$ is an {\rm LCD} subcode with the same parity
label as $C^{k,\ell}$. Applying Lemma~\ref{even-projection-count} with
$V={\rm L}^{k-\ell}$ and $W=\Hull(C^{k,\ell})$ gives the
total and parity-label formulas.

If $C^{k,\ell}$ is odd-like, then ${\rm L}^{k-\ell}$ is an odd-like
{\rm LCD} code. By Corollary~\ref{oe=oo}, the subcode counts with
prescribed dimension, hull dimension, and parity label are independent
of whether this odd-like LCD complement is of type ${\bf oo}$ or
type ${\bf oe}$. Thus it may be denoted by
${\rm L}_{\mathbf{o}}^{k-\ell}$, proving part {\rm (i)}.

If $C^{k,\ell}$ is even-like, then ${\rm L}^{k-\ell}$ is an even-like
{\rm LCD} code, denoted by ${\rm L}_{\mathbf{e}}^{k-\ell}$. Since every
subcode of an even-like code is even-like,
$\N_{\mathbf{e}}^t(s;C_{\mathbf{e}}^{k,\ell})
=\N^t(s;C_{\mathbf{e}}^{k,\ell})$. This proves part
{\rm (ii)}.
\end{proof}

In particular, every right-hand side in Theorem~\ref{all-hull} depends
only on $q$, the displayed subcode parameters, and the displayed parent
parameters. Hence the corresponding total and parity-label counts are
independent of the chosen parent representative and its length.

Having obtained the total number of subcodes with a prescribed hull
dimension, we next refine the enumeration by distinguishing the possible
types of these subcodes.

\begin{cor}\label{HLOO}
Let $n,k,\ell,s$ and $t$ be nonnegative integers such that $\ell\leq k\leq n-\ell$ and $0\leq t\leq s\leq k$. Then
\begin{enumerate}
    \item
    $\displaystyle
    \N_q\bigl[(s,t)_{\mathbf{o}};(k,\ell)_{{\bf oo}}\bigr]
    =
    \N_q\bigl[(s,t)_{{\bf oo}};(k,\ell)_{{\bf oo}}\bigr]
    =
    \sum_{m=0}^{\min\{t,\ell\}}
    q^{(s-m)(\ell-m)}
    \genfrac{[}{]}{0pt}{}{{\ell}}{m}_q
    \N_q\bigl[(s-m,t-m)_{{\bf oo}};
    (k-\ell,0)_{{\bf oo}}\bigr];
    $
    \item
    $\displaystyle
    \N_q\bigl[(s,t)_{\mathbf{e}};(k,\ell)_{{\bf oo}}\bigr]
    =
    \N_q\bigl[(s,t)_{{\bf eo}};(k,\ell)_{{\bf oo}}\bigr]
    =
    \sum_{m=0}^{\min\{t,\ell\}}
    q^{(s-m)(\ell-m)}
    \genfrac{[}{]}{0pt}{}{{\ell}}{m}_q
    \N_q\bigl[(s-m,t-m)_{{\bf eo}};
    (k-\ell,0)_{{\bf oo}}\bigr];
    $
    \item
    $\displaystyle
    \N_q\bigl[(s,t)_{{\bf oe}};(k,\ell)_{{\bf oo}}\bigr]
    =
    \N_q\bigl[(s,t)_{{\bf ee}};(k,\ell)_{{\bf oo}}\bigr]
    =0.
    $
\end{enumerate}
As usual, a summand is understood to be zero if the indicated subcode
does not exist.
\end{cor}
\begin{proof}
By Lemma~\ref{even-type-bookkeeping}, every subcode of
$C_{{\bf oo}}^{k,\ell}$ has odd-like dual. Hence no subcode of
types ${\bf oe}$ or ${\bf ee}$ can occur, proving part
{\rm (iii)}.

By Proposition~\ref{lcd-complement-types}, write
\[
C_{{\bf oo}}^{k,\ell}
=
{\rm L}_{{\bf oo}}^{k-\ell}
\oplus\Hull(C_{{\bf oo}}^{k,\ell}).
\]
Lemma~\ref{even-projection-count} gives the corresponding
odd-like and even-like subcode counts. Since every subcode of
${\rm L}_{{\bf oo}}^{k-\ell}$ also has odd-like dual, the odd-like
projected subcodes are exactly of type ${\bf oo}$, and the
even-like projected subcodes are exactly of type ${\bf eo}$.
This proves parts {\rm (i)} and {\rm (ii)}.
\end{proof}

\begin{cor}\label{HLEO}
Let $n,k,\ell,s$ and $t$ be nonnegative integers such that $\ell\leq k\leq n-\ell$ and $0\leq t\leq s\leq k$. Then
\begin{enumerate}
    \item $\displaystyle
    \N_q\bigl[(s,t)_{{\bf eo}};(k,\ell)_{{\bf eo}}\bigr]
    =
    \sum_{m=0}^{\min\{t,\ell\}}
    q^{(s-m)(\ell-m)}
    \genfrac{[}{]}{0pt}{}{{\ell}}{m}_q
    \N_q\bigl[(s-m,t-m)_{{\bf eo}};
    (k-\ell,0)_{{\bf eo}}\bigr];
    $
    \item
    $\displaystyle
    \N_q\bigl[(s,t)_{{\bf oo}};(k,\ell)_{{\bf eo}}\bigr]
    =
    \N_q\bigl[(s,t)_{{\bf oe}};(k,\ell)_{{\bf eo}}\bigr]
    =
    \N_q\bigl[(s,t)_{{\bf ee}};(k,\ell)_{{\bf eo}}\bigr]
    =0.
    $
\end{enumerate}
As usual, a summand is understood to be zero if the indicated subcode
does not exist.
\end{cor}

\begin{proof}
By Lemma~\ref{even-type-bookkeeping}, every subcode of
$C_{{\bf eo}}^{k,\ell}$ is even-like and has odd-like dual. Thus
only type ${\bf eo}$ can occur, proving part {\rm (ii)}.

By Proposition~\ref{lcd-complement-types}, write
\[
C_{{\bf eo}}^{k,\ell}
=
{\rm L}_{{\bf eo}}^{k-\ell}
\oplus\Hull(C_{{\bf eo}}^{k,\ell}).
\]
Applying Lemma~\ref{even-projection-count} and using the fact that
every subcode of ${\rm L}_{{\bf eo}}^{k-\ell}$ is also of type
${\bf eo}$ gives part {\rm (i)}.
\end{proof}

The type-${\bf oe}$ case requires one additional distinction:
an extension has even-like dual precisely when it contains
$\mathbf{1}_n$. Lemma~\ref{exist-oeLCD} ensures that the {\rm LCD}
complement can be chosen with the same type as the parent code.

\begin{cor}\label{HLOE}
Let $n,k,\ell,s$ and $t$ be nonnegative integers such that
$\ell\leq k\leq n-\ell$ and $0\leq t\leq s\leq k$. Then
\begin{enumerate}
    \item \(\displaystyle
    \N_q\bigl[(s,t)_{{\bf oe}};(k,\ell)_{{\bf oe}}\bigr]
    =
    \sum_{m=0}^{\min\{t,\ell,s-1\}}
    q^{(s-m-1)(\ell-m)}
    \genfrac{[}{]}{0pt}{}{{\ell}}{m}_q \cdot
    \N_q\bigl[(s-m,t-m)_{{\bf oe}};
    (k-\ell,0)_{{\bf oe}}\bigr];
    \)

    \item \(\displaystyle
    \N_q\bigl[(s,t)_{{\bf ee}};(k,\ell)_{{\bf oe}}\bigr]
    =
    \sum_{m=0}^{\min\{t,\ell,s-1\}}
    q^{(s-m-1)(\ell-m)}
    \genfrac{[}{]}{0pt}{}{{\ell}}{m}_q \cdot
    \N_q\bigl[(s-m,t-m)_{{\bf ee}};
    (k-\ell,0)_{{\bf oe}}\bigr];
    \)

    \item $\N_q\bigl[(s,t)_{{\bf oo}};(k,\ell)_{{\bf oe}}\bigr]
    =\N_q\bigl[(s,t)_{\mathbf{o}};(k,\ell)_{{\bf oe}}\bigr]
    -\N_q\bigl[(s,t)_{{\bf oe}};(k,\ell)_{{\bf oe}}\bigr];$

    \item $\N_q\bigl[(s,t)_{{\bf eo}};(k,\ell)_{{\bf oe}}\bigr]
    =\N_q\bigl[(s,t)_{\mathbf{e}};(k,\ell)_{{\bf oe}}\bigr]
    -\N_q\bigl[(s,t)_{{\bf ee}};(k,\ell)_{{\bf oe}}\bigr].$
\end{enumerate}
As usual, a summand is understood to be zero if the indicated subcode
does not exist.
\end{cor}

\begin{proof}
Let $C^{k,\ell}_{\bf oe}$ be an $[n,k]_q$ linear code of type
${\bf oe}$ with an $\ell$-dimensional hull. By
Lemma~\ref{exist-oeLCD}, write
\[
C^{k,\ell}_{\bf oe}
=
{\rm L}_{\bf oe}^{k-\ell}
\oplus
\Hull(C^{k,\ell}_{\bf oe}).
\]
As in the proof of Theorem~\ref{all-hull}, for an $s$-dimensional
subcode $D\subseteq C^{k,\ell}_{\bf oe}$ with a $t$-dimensional
hull, put
$m=\dim(D\cap\Hull(C^{k,\ell}_{\bf oe}))$. Then
\[
\dim\Pi_{{\rm L}_{\bf oe}^{k-\ell}}(D)=s-m,
\qquad
\dim\Hull
\bigl(\Pi_{{\rm L}_{\bf oe}^{k-\ell}}(D)\bigr)=t-m.
\]
Conversely, for each fixed projected subcode
$U\subseteq{\rm L}_{\bf oe}^{k-\ell}$ with these two dimensions,
Lemma~\ref{proj} gives exactly
\(
q^{(s-m)(\ell-m)}\genfrac{[}{]}{0pt}{}{\ell}{m}_q
\)
subcodes $D\subseteq C^{k,\ell}_{\bf oe}$ satisfying
$\Pi_{{\rm L}_{\bf oe}^{k-\ell}}(D)=U$ and
$\dim(D\cap\Hull(C^{k,\ell}_{\bf oe}))=m$.
By Lemma~\ref{even-parity-proj}, $D$ and $U$ have the same parity
label. Therefore,
\[
\N^t_{\mathbf{o}}(s;C^{k,\ell}_{\bf oe})
=
\sum_{m=0}^{\min\{t,\ell\}}
q^{(s-m)(\ell-m)}\genfrac{[}{]}{0pt}{}{\ell}{m}_q
\N^{t-m}_{\mathbf{o}}
(s-m;{\rm L}_{\bf oe}^{k-\ell})
\]
and
\[
\N^t_{\mathbf{e}}(s;C^{k,\ell}_{\bf oe})
=
\sum_{m=0}^{\min\{t,\ell\}}
q^{(s-m)(\ell-m)}\genfrac{[}{]}{0pt}{}{\ell}{m}_q
\N^{t-m}_{\mathbf{e}}
(s-m;{\rm L}_{\bf oe}^{k-\ell}).
\]

It remains to determine the dual parity. Lemma~\ref{1_n} gives
$\mathbf{1}_n\in{\rm L}_{\bf oe}^{k-\ell}$, and
Lemma~\ref{even-type-bookkeeping} shows that $D^\perp$ is even-like
if and only if $\mathbf{1}_n\in D$.
If $\mathbf{1}_n\in U$, then $U\neq\{\mathbf0\}$, and hence
$s-m\geq1$. By Corollary~\ref{projee}, with
$\bm\alpha=\mathbf{1}_n$, the number of extensions of $U$
containing $\mathbf{1}_n$ is
\(
q^{(s-m-1)(\ell-m)}\genfrac{[}{]}{0pt}{}{\ell}{m}_q.
\)
Projected subcodes of types ${\bf oe}$ and ${\bf ee}$,
respectively, therefore give the formulas in parts {\rm (i)} and
{\rm (ii)}. Notice that only $m\leq s-1$ occurs in these two
formulas, since a projected subcode containing $\mathbf{1}_n$
cannot be zero.

Finally, odd-like subcodes are partitioned into those of types
${\bf oo}$ and ${\bf oe}$, whereas even-like subcodes are
partitioned into those of types ${\bf eo}$ and ${\bf ee}$.
This prove parts {\rm (iii)} and {\rm (iv)}.
\end{proof}
\begin{cor}\label{HLEE}
Let $n,k,\ell,s$ and $t$ be nonnegative integers such that $1\leq \ell \leq k\leq n-\ell$ and $0\leq t\leq s\leq k$. Then
\begin{enumerate}
    \item
    $\displaystyle
    \N_q\bigl[(s,t)_{{\bf ee}};(k,\ell)_{{\bf ee}}\bigr]
    =
    \sum_{m=1}^{\min\{t,\ell\}}
    q^{(s-m)(\ell-m)}
    \genfrac{[}{]}{0pt}{}{{\ell-1}}{{m-1}}_q
    \N_q\bigl[(s-m,t-m)_{\mathbf{e}};
    (k-\ell,0)_{{\bf eo}}\bigr];
    $
    \item
    $\N_q\bigl[(s,t)_{{\bf eo}};(k,\ell)_{{\bf ee}}\bigr]
    =\N_q\bigl[(s,t);(k,\ell)_{{\bf ee}}\bigr]
    -\N_q\bigl[(s,t)_{{\bf ee}};(k,\ell)_{{\bf ee}}\bigr];$

    \item
    $\displaystyle
    \N_q\bigl[(s,t)_{{\bf oo}};(k,\ell)_{{\bf ee}}\bigr]
    =
    \N_q\bigl[(s,t)_{{\bf oe}};(k,\ell)_{{\bf ee}}\bigr]
    =0.
    $
\end{enumerate}
As usual, a summand is understood to be zero if the indicated subcode
does not exist.
\end{cor}

\begin{proof}
Let $C^{k,\ell}_{\bf ee}$ be an $[n,k]_q$ linear code of type
${\bf ee}$ with an $\ell$-dimensional hull. By
Proposition~\ref{lcd-complement-types}, write
\[
C^{k,\ell}_{\bf ee}
=
{\rm L}_{\bf eo}^{k-\ell}
\oplus
\Hull(C^{k,\ell}_{\bf ee}),
\]
and $\mathbf{1}_n\in\Hull(C^{k,\ell}_{\bf ee})$.
By Lemma~\ref{even-type-bookkeeping}, every subcode of
$C^{k,\ell}_{\bf ee}$ is even-like. Hence no subcodes of type
${\bf oo}$ or ${\bf oe}$ occur, proving part {\rm (iii)}.

We use the same projection-counting argument as in Theorem~\ref{all-hull}.
For $D\subseteq C^{k,\ell}_{\bf ee}$, put 
$m=\dim(D\cap \Hull(C^{k,\ell}_{\bf ee}))$. Then
\[
\dim \Pi_{{\rm L}_{\bf eo}^{k-\ell}}(D)=s-m,
\qquad
\dim\Hull(\Pi_{{\rm L}_{\bf eo}^{k-\ell}}(D))=t-m.
\]
and a fixed projected subcode $U$ has
$q^{(s-m)(\ell-m)}\genfrac{[}{]}{0pt}{}{\ell}{m}_q$ extensions by Lemma~\ref{proj}.
By Lemma~\ref{even-type-bookkeeping}, $D$ has even-like dual if and
only if $\mathbf{1}_n\in D$. Since
$\mathbf{1}_n\in\Hull(C^{k,\ell}_{\bf ee})$, this is
equivalent to
$\mathbf{1}_n\in D\cap\Hull(C^{k,\ell}_{\bf ee})$.

For type ${\bf ee}$ we must impose $\mathbf{1}_n\in D\cap \Hull(C^{k,\ell}_{\bf ee})$.
Corollary~\ref{projeee}, with $\bm\alpha=\mathbf{1}_n$, gives
\(
q^{(s-m)(\ell-m)}\genfrac{[}{]}{0pt}{}{{\ell-1}}{{m-1}}_q
\)
extensions, and therefore $m\geq1$. Summing over the projected even-like
subcodes $U$ gives part {\rm (i)}.

Since all subcodes are even-like, those not of type
${\bf ee}$ are of type ${\bf eo}$. Hence part
{\rm (ii)} follows.
\end{proof}
Together, Corollaries~\ref{HLOO}--\ref{HLEE} show that every
complete-type count in even characteristic is determined by
$q,(s,t)_\tau$, and $(k,\ell)_\varepsilon$.

\section{Enumeration of subcodes with prescribed hull dimensions in odd characteristic}\label{sec-3}

Throughout this section, $q$ has odd characteristic. As in the even
case, the explicit formulas show that the subcode count depends only
on $q,(s,t)_\tau$, and $(k,\ell)_\varepsilon$, and not on a chosen
representative or its length. We therefore use the parameterized
notation in the principal results and retain the signs $+$ and $-$ only
in the representative-level formulas inherited from the literature.

We use the types $+$ and $-$ introduced in Section~\ref{sec-1}. We first
recall the classification and enumeration of {\rm LCD} codes of these two
types. We
then determine the numbers of subcodes with prescribed dimension, hull
dimension, and type inside a fixed {\rm LCD} code. Finally, we extend these
formulas to an arbitrary fixed linear code by using a projection onto an
{\rm LCD} complement of its hull.

Let $\mathbf{e}_i\in\F_q^n$ denote the $i$-th standard
basis vector. Choose $a,b\in\F_q$ such that
$\delta=a^2+b^2$ is a nonsquare, and define
\[
G_+^{k,n}
=\begin{pmatrix}
\mathbf{e}_1\\
\vdots\\
\mathbf{e}_k
\end{pmatrix},
\qquad
G_-^{k,n}
=\begin{pmatrix}
\mathbf{e}_1\\
\vdots\\
\mathbf{e}_{k-1}\\
a\mathbf{e}_k+b\mathbf{e}_{k+1}
\end{pmatrix}.
\]
Here $G_+^{k,n}$ is defined for $k\leq n$, whereas $G_-^{k,n}$
requires $k<n$. The codes generated by these matrices are {\rm LCD} codes
of types $+$ and $-$, respectively.
 The code length is
suppressed whenever it is clear from the context.
\begin{prop}[{\cite{binaryLCD}}]\label{canonical}
Let $C$ be an $[n,k]_q$ {\rm LCD} code. If $C$ is of type $+$
(resp. type $-$), then there exist a generator matrix $G$ of $C$
and an orthogonal matrix $Q\in\mathbb{O}_n(q)$ such that
\[
G=G_+^{k,n}Q
\qquad
(\text{resp. }G=G_-^{k,n}Q).
\]
Conversely, the codes generated by $G_+^{k,n}Q$ and $G_-^{k,n}Q$
are {\rm LCD} codes of types $+$ and $-$, respectively.
\end{prop}

\subsection{Subcodes with prescribed hull dimensions in {\rm LCD} codes}

We first recall the numbers of {\rm LCD} subcodes of each type in an {\rm LCD}
code. These formulas will serve as the initial values in the recursive
enumeration of subcodes with positive hull dimension.

Let $\eta$ denote the quadratic
character of $\F_q$, extended by $\eta(0)=0$.
\begin{lemma}[{\cite[Theorem 4.9]{x-lcd}}]\label{SubLCD-S}
Let ${\rm L}_+^k$ be an $[n,k]_q$ {\rm LCD} code of type $+$, and let
$1\leq s<k$. Then
  \begin{align*} \N^0_+(s;{\rm L}_+^k)= \begin{cases}
 \frac{1}{2}  q^{\frac{s(k-s)-1}{2}}   \left (q^{\frac{k}{2}}- \eta \left( (-1)^{\frac{k}{2}}    \right )  \right )             {\frac{k}{2}-1 \brack \frac{s-1}{2}}_{q^2}, &\text{if $s$ is odd and $k$ is even};
 \\  \frac{1}{2}  q^{\frac{s(k-s)}{2}}   \left (q^{\frac{k-s}{2}}+ \eta \left( (-1)^{\frac{k-s}{2}}    \right )  \right )             {\frac{k-1}{2} \brack \frac{s-1}{2}}_{q^2}, &\text{if $s$ is odd and $k$ is odd};
 \\  \frac{1}{2}  q^{\frac{s(k-s)}{2}}   \left (q^{\frac{s}{2}}+ \eta \left( (-1)^{\frac{s}{2}}    \right )  \right )  {\frac{k-1}{2} \brack \frac{s}{2}}_{q^2}, &\text{if $s$ is even and $k$ is odd};
 \\  \frac{1}{2}  q^{\frac{s(k-s)}{2}}    \frac{\left (q^{\frac{s}{2}}+ \eta \left( (-1)^{\frac{s}{2}}    \right )  \right )
 \left (q^{\frac{k-s}{2}}+ \eta \left( (-1)^{\frac{k-s}{2}}    \right )  \right )  }{\left (q^{\frac{k}{2}}+ \eta \left( (-1)^{\frac{k}{2}}    \right )  \right ) }
     {\frac{k}{2} \brack \frac{s}{2}}_{q^2},&\text{if $s$ is even and $k$ is even}.
  \end{cases}
  \end{align*}
  
\begin{align*}\N^0_-(s;{\rm L}_+^k)= \begin{cases}
 \frac{1}{2}  q^{\frac{s(k-s)-1}{2}}   \left (q^{\frac{k}{2}}- \eta \left( (-1)^{\frac{k}{2}}    \right )  \right )             {\frac{k}{2}-1 \brack \frac{s-1}{2}}_{q^2}, &\text{if $s$ is odd and $k$ is even};
 \\  \frac{1}{2}  q^{\frac{s(k-s)}{2}}   \left (q^{\frac{k-s}{2}}- \eta \left( (-1)^{\frac{k-s}{2}}    \right )  \right )             {\frac{k-1}{2} \brack \frac{s-1}{2}}_{q^2}, &\text{if $s$ is odd and $k$ is odd};
 \\  \frac{1}{2}  q^{\frac{s(k-s)}{2}}   \left (q^{\frac{s}{2}}- \eta \left( (-1)^{\frac{s}{2}}    \right )  \right )  {\frac{k-1}{2} \brack \frac{s}{2}}_{q^2}, &\text{if $s$ is even and $k$ is odd};
 \\  \frac{1}{2}  q^{\frac{s(k-s)}{2}}    \frac{\left (q^{\frac{s}{2}}- \eta \left( (-1)^{\frac{s}{2}}    \right )  \right )
 \left (q^{\frac{k-s}{2}}- \eta \left( (-1)^{\frac{k-s}{2}}    \right )  \right )  }{\left (q^{\frac{k}{2}}+ \eta \left( (-1)^{\frac{k}{2}}    \right )  \right ) }
     {\frac{k}{2} \brack \frac{s}{2}}_{q^2},&\text{if $s$ is even and $k$ is even}.
     \end{cases}
     \end{align*}
\end{lemma}

\begin{lemma}[{\cite[Theorem 4.18]{x-lcd}}]\label{SubLCD-NS}
Let ${\rm L}_-^k$ be an $[n,k]_q$ {\rm LCD} code of type $-$, and let
$1\leq s<k$. Then

\[\N^0_+(s;{\rm L}_-^k)=
\begin{cases}
        \frac{1}{2}  q^{\frac{s(k-s)-1}{2}}   \left (q^{\frac{k}{2}}+ \eta \left( (-1)^{\frac{k}{2}}    \right )  \right )             {\frac{k}{2}-1 \brack \frac{s-1}{2}}_{q^2}, &\text{if $s$ is odd and $k$ is even};
        \\  \frac{1}{2}  q^{\frac{s(k-s)}{2}}   \left (q^{\frac{k-s}{2}}- \eta \left( (-1)^{\frac{k-s}{2}}    \right )  \right )             {\frac{k-1}{2} \brack \frac{s-1}{2}}_{q^2}, &\text{if $s$ is odd and $k$ is odd};
         \\  \frac{1}{2}  q^{\frac{s(k-s)}{2}}   \left (q^{\frac{s}{2}}+ \eta \left( (-1)^{\frac{s}{2}}    \right )  \right )  {\frac{k-1}{2} \brack \frac{s}{2}}_{q^2}, &\text{if $s$ is even and $k$ is odd};
         \\ \frac{1}{2}  q^{\frac{s(k-s)}{2}} \left( \frac{\left(q^{\frac{s}{2}}+ \eta \left( (-1)^{\frac{s}{2}}  \right )  \right ) \left (q^{\frac{k-s}{2}}- \eta \left( (-1)^{\frac{k-s}{2}}    \right )  \right ) }{q^{\frac{k}{2}}- \eta \left( (-1)^{\frac{k}{2}}    \right )}\right) {\frac{k}{2} \brack \frac{s}{2}}_{q^2}, &\text{if $s$ is even and $k$ is even},
    \end{cases}\]

\[\N^0_-(s;{\rm L}_-^k)=
   \begin{cases}
        \frac{1}{2}  q^{\frac{s(k-s)-1}{2}}   \left (q^{\frac{k}{2}}+ \eta \left( (-1)^{\frac{k}{2}}    \right )  \right )             {\frac{k}{2}-1 \brack \frac{s-1}{2}}_{q^2}, &\text{if $s$ is odd and $k$ is even};
         \\  \frac{1}{2}  q^{\frac{s(k-s)}{2}}   \left (q^{\frac{k-s}{2}}+ \eta \left( (-1)^{\frac{k-s}{2}}    \right )  \right )             {\frac{k-1}{2} \brack \frac{s-1}{2}}_{q^2}, &\text{if $s$ is odd and $k$ is odd};
         \\  \frac{1}{2}  q^{\frac{s(k-s)}{2}}   \left (q^{\frac{s}{2}}- \eta \left( (-1)^{\frac{s}{2}}    \right )  \right )  {\frac{k-1}{2} \brack \frac{s}{2}}_{q^2}, &\text{if $s$ is even and $k$ is odd};
         \\  \frac{1}{2}  q^{\frac{s(k-s)}{2}} \left( \frac{\left(q^{\frac{s}{2}}-\eta \left( (-1)^{\frac{s}{2}} \right)   \right ) \left (q^{\frac{k-s}{2}}+ \eta \left( (-1)^{\frac{k-s}{2}}    \right )  \right ) }{q^{\frac{k}{2}}- \eta \left( (-1)^{\frac{k}{2}}    \right )}\right) {\frac{k}{2} \brack \frac{s}{2}}_{q^2}, &\text{if $s$ is even and $k$ is even}.
    \end{cases}\]
\end{lemma}

Before treating subcodes of an {\rm LCD} parent code, we record the
corresponding enumeration in the whole vector space. This is a direct
reformulation of Theorem 4.11 in \cite{Li-hull-odd} using the notation
of this paper.

\begin{lem}[{\cite[Theorem 4.11]{Li-hull-odd}}]
\label{whole-space-count-odd}
Let $1\leq t\leq s\leq k-t$, and put $s_0=s-t$. Then
\[
\N_+^t(s;\F_q^k)
=
\begin{cases}
\displaystyle
\left(
\prod_{i=1}^{t}
\frac{q^{k-s_0-2i+1}-1}{q^{s_0}(q^i-1)}
\right)
\N_+^0(s_0;\F_q^k),
&\text{if $k-s_0$ is odd},\\[4mm]
\displaystyle
A_1(k,s_0,t)\N_+^0(s_0;\F_q^k),
&\text{if $k-s_0$ is even},
\end{cases}
\]
and
\[
\N_-^t(s;\F_q^k)
=
\begin{cases}
\displaystyle
\left(
\prod_{i=1}^{t}
\frac{q^{k-s_0-2i+1}-1}{q^{s_0}(q^i-1)}
\right)
\N_-^0(s_0;\F_q^k),
&\text{if $k-s_0$ is odd},\\[4mm]
\displaystyle
A_2(k,s_0,t)\N_-^0(s_0;\F_q^k),
&\text{if $k-s_0$ is even},
\end{cases}
\]
where
\[
A_1(k,s_0,t)
=
\frac{
q^{\frac{k-s_0}{2}-t}
+\eta\left((-1)^{\frac{k-s_0}{2}}\right)
}{
q^{\frac{k-s_0}{2}}
+\eta\left((-1)^{\frac{k-s_0}{2}}\right)
}
\prod_{i=1}^{t}
\frac{q^{k-s_0-2i+2}-1}{q^{s_0}(q^i-1)}
\]
and
\[
A_2(k,s_0,t)
=
\frac{
q^{\frac{k-s_0}{2}-t}
-\eta\left((-1)^{\frac{k-s_0}{2}}\right)
}{
q^{\frac{k-s_0}{2}}
-\eta\left((-1)^{\frac{k-s_0}{2}}\right)
}
\prod_{i=1}^{t}
\frac{q^{k-s_0-2i+2}-1}{q^{s_0}(q^i-1)}.
\]
\end{lem}

Since $\F_q^k$ is an {\rm LCD} code of type $+$, the two initial values
appearing in Lemma \ref{whole-space-count-odd} are given explicitly
by Lemma \ref{SubLCD-S}.

The preceding result immediately gives the enumeration for a type
$+$ {\rm LCD} parent code.

\begin{prop}\label{S-subcode-odd}
Let ${\rm L}_+^k$ be an $[n,k]_q$ {\rm LCD} code of type $+$. For
$1\leq t\leq s\leq k$, we have
\[
\N_+^t(s;{\rm L}_+^k)
= \N_+^t(s;\F_q^k)
\]
and
\[
\N_-^t(s;{\rm L}_+^k)
= \N_-^t(s;\F_q^k).
\]
\end{prop}

\begin{proof}
The proof is parallel to that of Proposition~\ref{oo number} in the
even-characteristic type-${\bf oo}$ case.
\end{proof}

The next two results prepare the recursive construction used below.
They show how, inside an {\rm LCD} parent code, an $s$-dimensional
subcode with a $t$-dimensional hull can be extended by a suitable
self-orthogonal vector to an $s+1$-dimensional subcode with a
$t+1$-dimensional hull, and they also determine the exact multiplicity with
which each resulting subcode is obtained.

\begin{prop}\label{hullequal}
Let $C$ be an $[n,k]_q$ {\rm LCD} code, and let $D$ be an
$[n,s]_q$ subcode of $C$ with a $t$-dimensional hull. Then
$D^\perp\cap C$ is an $[n,k-s]_q$ code with a $t$-dimensional
hull. Moreover,
\[
\Hull(D^\perp\cap C)
=\Hull(D).
\]
\end{prop}

\begin{proof}
Since $D\subseteq C$, we have $C^\perp\subseteq D^\perp$.
As $C$ is {\rm LCD},
$
\F_q^n=C\oplus C^\perp,
$
and hence $D^\perp+C=\F_q^n$. Therefore,
\begin{align*}
\dim(D^\perp\cap C)
&=\dim(D^\perp)+\dim(C)-\dim(D^\perp+C)=k-s.
\end{align*}

Moreover, since $D\subseteq C$ and
$C\cap C^\perp=\{\mathbf{0}\}$, the modular law gives
\[
(D+C^\perp)\cap C
=D+(C^\perp\cap C)=D.
\]
Using
$
(D^\perp\cap C)^\perp=D+C^\perp,
$
we obtain
\begin{align*}
\Hull(D^\perp\cap C)
&=(D^\perp\cap C)\cap(D+C^\perp)\\
&=D^\perp\cap\bigl((D+C^\perp)\cap C\bigr)\\
&=D^\perp\cap D.
\end{align*}
Thus
$\Hull(D^\perp\cap C)
=
\Hull(D)$,
and hence $D^\perp\cap C$ has a $t$-dimensional hull.
\end{proof}

\begin{lem}\label{cons}
Let $C$ be an {\rm LCD} code, and let $D_0$ be an
$[n,s-1]_q$ subcode of $C$ with a $(t-1)$-dimensional hull.
Suppose that $\mathbf{x}$ is a self-orthogonal vector satisfying
\[
\mathbf{x}\in
(D_0^\perp\cap C)
\setminus\Hull(D_0^\perp\cap C).
\]
Then
$
D=D_0\oplus\Span\{\mathbf{x}\}
$
is an $[n,s]_q$ subcode of $C$ with a $t$-dimensional hull,
and
\[
\Hull(D)
=
\Hull(D_0)
\oplus\Span\{\mathbf{x}\}.
\]
Moreover, $D_0$ and $D$ have the same type.

Conversely, every fixed $[n,s]_q$ subcode $D\subseteq C$ with a
$t$-dimensional hull is obtained from exactly
$
(q^t-1)q^{s-1}
$
pairs $(D_0,\mathbf{x})$ satisfying the above conditions.
\end{lem}

\begin{proof}
By Proposition \ref{hullequal},
$
\Hull(D_0^\perp\cap C)
=\Hull(D_0).
$
Therefore,
\begin{equation}
(D_0^\perp\cap C)
\setminus\Hull(D_0^\perp\cap C)
=
\left(
D_0^\perp\setminus\Hull(D_0^\perp)
\right)\cap C.
\label{eq:extension-set-in-C}
\end{equation}

Thus $\mathbf{x}$ satisfies the condition in
{\cite[Theorem 2.2]{Li-hull-odd}} and also belongs to $C$.
It follows that
$
D=D_0\oplus\Span\{\mathbf{x}\}
$
is an $[n,s]_q$ subcode of $C$ with a $t$-dimensional hull,
and
$
\Hull(D)
=
\Hull(D_0)
\oplus\Span\{\mathbf{x}\}.
$

Let $G_{D_0}$ be a generator matrix of $D_0$. Since
$\mathbf{x}\in D_0^\perp$ and $\mathbf{x}$ is self-orthogonal,
\[
\begin{pmatrix}
G_{D_0}\\
\mathbf{x}
\end{pmatrix}
\begin{pmatrix}
G_{D_0}\\
\mathbf{x}
\end{pmatrix}^{T}
=
\begin{pmatrix}
G_{D_0}G_{D_0}^T&O\\
O&O
\end{pmatrix}.
\]
Hence $D_0$ and $D$ have the same type.

Finally, let $D\subseteq C$ be a fixed $[n,s]_q$ subcode with a
$t$-dimensional hull. By
{\cite[Proposition 2.3]{Li-hull-odd}}, there are exactly
$
(q^t-1)q^{s-1}
$
pairs $(D_0,\mathbf{x})$ such that
$
D=D_0\oplus\Span\{\mathbf{x}\},
$
where $D_0$ has dimension $s-1$ and hull dimension $t-1$, and
$\mathbf{x}$ is self-orthogonal with
$
\mathbf{x}\in
D_0^\perp\setminus\Hull(D_0^\perp).
$
Since $D_0\subseteq D\subseteq C$ and
$\mathbf{x}\in D\subseteq C$, Equation
\eqref{eq:extension-set-in-C} shows that these pairs satisfy the
conditions in the statement. This completes the proof.
\end{proof}
We shall need to determine the type of $D^\perp\cap C$ inside an
{\rm LCD} parent code $C$. For this purpose, we first recall the following
type relation from \cite{Li-hull-odd}.
\begin{lem}[{\cite[Proposition 4.6]{Li-hull-odd}}]
\label{otype}
Let $C$ be an $[n,k]_q$ linear code with an $\ell$-dimensional
hull, and let $D\subseteq C$ be a $(k-r)$-dimensional subcode
with an $(\ell+r)$-dimensional hull. Then:
\begin{enumerate}
\item if $q\equiv1\pmod4$, then $D$ has the same type as $C$;
\item if $q\equiv3\pmod4$ and $r$ is even, then $D$ has the same
type as $C$;
\item if $q\equiv3\pmod4$ and $r$ is odd, then $D$ and $C$ have
different types.
\end{enumerate}
\end{lem}

\begin{remarknonumber}
The preceding statement is the intermediate type relation established
in the proof of Proposition 4.6 in \cite{Li-hull-odd}. We isolate it
because it will be applied to orthogonal complements taken inside a
fixed {\rm LCD} parent code.
\end{remarknonumber}

\begin{prop}\label{otype2}
Let $C$ be an $[n,k]_q$ {\rm LCD} code of type $-$, and let
$D\subseteq C$ be an $s$-dimensional subcode with a
$t$-dimensional hull. Then:
\begin{enumerate}
\item if $q\equiv1\pmod4$, then $D^\perp\cap C$ and $D$ have
different types;
\item if $q\equiv3\pmod4$ and $t$ is even, then
$D^\perp\cap C$ and $D$ have different types;
\item if $q\equiv3\pmod4$ and $t$ is odd, then
$D^\perp\cap C$ and $D$ have the same type.
\end{enumerate}
\end{prop}
\begin{proof}
By Proposition \ref{hull-matrix}, $D$ contains an
$(s-t)$-dimensional {\rm LCD} subcode ${\rm L}^{s-t}$ of the same type as $D$.
Then
$
D^\perp\cap C
\subseteq
({\rm L}^{s-t})^\perp\cap C.
$
Since $C$ is of type $-$, the {\rm LCD} codes
$({\rm L}^{s-t})^\perp\cap C$ and ${\rm L}^{s-t}$ have different types.
Moreover, Proposition \ref{hullequal} shows that
$D^\perp\cap C$ has a $t$-dimensional hull. Applying Lemma
\ref{otype} to
$
D^\perp\cap C
\subseteq
({\rm L}^{s-t})^\perp\cap C
$
gives the three asserted type relations. Since $D$ and ${\rm L}^{s-t}$
have the same type, the proof is complete.
\end{proof}

\begin{cor}\label{dual-type}
Let $C$ be an $[n,k]_q$ {\rm LCD} code of type $+$, and let
$D\subseteq C$ be an $s$-dimensional subcode with a
$t$-dimensional hull. Then:
\begin{enumerate}
\item if $q\equiv1\pmod4$, then $D^\perp\cap C$ and $D$ have the
same type;
\item if $q\equiv3\pmod4$ and $t$ is even, then
$D^\perp\cap C$ and $D$ have the same type;
\item if $q\equiv3\pmod4$ and $t$ is odd, then
$D^\perp\cap C$ and $D$ have different types.
\end{enumerate}
\end{cor}

\begin{lem}[{\cite[Proposition 4.10]{Li-hull-odd}}]
\label{sonumber}
Assume that $C$ is an $[n,k]_q$ linear code with an
$\ell$-dimensional hull, where $k>\ell$. If $C$ is of type $+$
(resp. type $-$), then the number of self-orthogonal codewords in
$C\setminus\Hull(C)$ is
\[
q^\ell M_+(k-\ell)
\quad
\left(\text{resp. }q^\ell M_-(k-\ell)\right),
\]
where
\[
M_+(r)=
\begin{cases}
q^{r-1}-1,
    & \text{if $r$ is odd},\\[1mm]
\left(q^{\frac{r}{2}}
-\eta\left((-1)^{\frac{r}{2}}\right)\right)
\left(q^{\frac{r}{2}-1}
+\eta\left((-1)^{\frac{r}{2}}\right)\right),
    & \text{if $r$ is even},
\end{cases}
\]
and
\[
M_-(r)=
\begin{cases}
q^{r-1}-1,
    & \text{if $r$ is odd},\\[1mm]
\left(q^{\frac{r}{2}}
+\eta\left((-1)^{\frac{r}{2}}\right)\right)
\left(q^{\frac{r}{2}-1}
-\eta\left((-1)^{\frac{r}{2}}\right)\right),
    & \text{if $r$ is even}.
\end{cases}
\]
\end{lem}

With the preceding preparations, we can now derive the corresponding
enumeration formula for subcodes of an {\rm LCD} code of type $-$.
\begin{prop}\label{NS-subcode-odd}
Let ${\rm L}_-^k$ be an $[n,k]_q$ {\rm LCD} code of type $-$. For
$1\leq t\leq s\leq k-t$, put $s_0=s-t$. Then
\[
\N_+^t(s;{\rm L}_-^k)
=
\begin{cases}
\displaystyle
\left(
\prod_{i=1}^{t}
\frac{
q^{k-s_0-2i+1}-1
}{
q^{s_0}(q^i-1)
}
\right)
\N_+^0(s_0;{\rm L}_-^k),
&
\text{if $k-s_0$ is odd},
\\[4mm]
\displaystyle
A_2(k,s_0,t)\N_+^0(s_0;{\rm L}_-^k),
&
\text{if $k-s_0$ is even},
\end{cases}
\]
and
\[
\N_-^t(s;{\rm L}_-^k)
=
\begin{cases}
\displaystyle
\left(
\prod_{i=1}^{t}
\frac{
q^{k-s_0-2i+1}-1
}{
q^{s_0}(q^i-1)
}
\right)
\N_-^0(s_0;{\rm L}_-^k),
&
\text{if $k-s_0$ is odd},
\\[4mm]
\displaystyle
A_1(k,s_0,t)\N_-^0(s_0;{\rm L}_-^k),
&
\text{if $k-s_0$ is even},
\end{cases}
\]
Here $A_1(k,s_0,t)$ and $A_2(k,s_0,t)$ are as defined in
Lemma \ref{whole-space-count-odd}.
\end{prop}

\begin{proof}
For each $1\leq i\leq t$ and $\tau\in\{+,-\}$, let
$
D\in\mathcal{S}_\tau^{i-1}(s_0+i-1;{\rm L}_-^k)
$
and put
$
D_{{\rm L}_-^k}^\perp=D^\perp\cap {\rm L}_-^k.
$
Since ${\rm L}_-^k$ is {\rm LCD}, we have
$
\dim(D_{{\rm L}_-^k}^\perp)=k-(s_0+i-1).
$
Moreover, Proposition \ref{hullequal} gives
$
\dim\bigl(\Hull(D_{{\rm L}_-^k}^\perp)\bigr)=i-1.
$
It follows that
\[
\dim(D_{{\rm L}_-^k}^\perp)
-
\dim\bigl(\Hull(D_{{\rm L}_-^k}^\perp)\bigr)
=
k-s_0-2i+2.
\]

By Lemma \ref{cons}, every self-orthogonal vector
$
\mathbf{x}\in
D_{{\rm L}_-^k}^\perp
\setminus
\Hull(D_{{\rm L}_-^k}^\perp)
$
extends $D$ to
$
D\oplus\Span\{\mathbf{x}\},
$
which is an $(s_0+i)$-dimensional subcode of ${\rm L}_-^k$ with an
$i$-dimensional hull and the same type as $D$. Furthermore,
every fixed subcode in
$\mathcal{S}_\tau^i(s_0+i;{\rm L}_-^k)$ is counted exactly
$
(q^i-1)q^{s_0+i-1}
$
times in this construction.

Assume first that $q\equiv1\pmod4$. Since ${\rm L}_-^k$ is of type $-$,
Proposition \ref{otype2} shows that $D_{{\rm L}_-^k}^\perp$ is of type $-$
when $D$ is of type $+$, and is of type $+$ when $D$ is of
type $-$. Hence Lemma \ref{sonumber} gives
\[
\N_+^t(s;{\rm L}_-^k)
=
\left(
\prod_{i=1}^{t}
\frac{
M_-(k-s_0-2i+2)
}{
q^{s_0}(q^i-1)
}
\right)
\N_+^0(s_0;{\rm L}_-^k)
\]
and
\[
\N_-^t(s;{\rm L}_-^k)
=
\left(
\prod_{i=1}^{t}
\frac{
M_+(k-s_0-2i+2)
}{
q^{s_0}(q^i-1)
}
\right)
\N_-^0(s_0;{\rm L}_-^k).
\]

Now assume that $q\equiv3\pmod4$. For $1\leq i\leq t$, define
\[
M_{1,i}(m)
=
\begin{cases}
M_-(m), & \text{if $i$ is even},\\
M_+(m), & \text{if $i$ is odd},
\end{cases}
\]
and
\[
M_{2,i}(m)
=
\begin{cases}
M_+(m), & \text{if $i$ is even},\\
M_-(m), & \text{if $i$ is odd}.
\end{cases}
\]
In this case, Proposition \ref{otype2} shows that $D_{{\rm L}_-^k}^\perp$
has the same type as $D$ when $i$ is even and the opposite type
when $i$ is odd. Therefore,
\[
\N_+^t(s;{\rm L}_-^k)
=
\left(
\prod_{i=1}^{t}
\frac{
M_{2,i}(k-s_0-2i+2)
}{
q^{s_0}(q^i-1)
}
\right)
\N_+^0(s_0;{\rm L}_-^k)
\]
and
\[
\N_-^t(s;{\rm L}_-^k)
=
\left(
\prod_{i=1}^{t}
\frac{
M_{1,i}(k-s_0-2i+2)
}{
q^{s_0}(q^i-1)
}
\right)
\N_-^0(s_0;{\rm L}_-^k).
\]

It remains only to simplify the products above. Replace the whole-space
dimension $n$ in the cited theorem by $k$, and replace its parameters
$(k,\ell,k_0)$ by $(s,t,s_0)$. The resulting products are precisely
those simplified in the proof of
{\cite[Theorem 4.11, equations (6)--(9)]{Li-hull-odd}}.
More precisely, the products involving $M_+$ and $M_{1,i}$
give the factor $A_1(k,s_0,t)$, whereas those involving $M_-$ and
$M_{2,i}$ give the factor $A_2(k,s_0,t)$. Since the parent code
${\rm L}_-^k$ is of type $-$, the factor $A_2(k,s_0,t)$ occurs in
$\N_+^t(s;{\rm L}_-^k)$, while $A_1(k,s_0,t)$ occurs in
$\N_-^t(s;{\rm L}_-^k)$. Applying the product identities
established there gives the asserted formulas.
\end{proof}
\subsection{Subcodes with prescribed hull dimensions in arbitrary linear codes}

We next extend the preceding enumeration from {\rm LCD} parent codes to
arbitrary linear codes by using the LCD-complement decomposition in
Proposition~\ref{lcd-complement-types}.

\begin{theorem}\label{linear-code-count}
Assume that $q$ is odd. Let $C$ be an $[n,k]_q$ linear code with an
$\ell$-dimensional hull and type $\varepsilon$, where
$0\leq\ell\leq\min\{k,n-k\}$. Let $0\leq t\leq s\leq k$, and let
$\tau\in\{+,-\}$ be an admissible type. Then
\[
\N_q\bigl[(s,t)_\tau;(k,\ell)_\varepsilon\bigr]
=
\sum_{m=0}^{\min\{t,\ell\}}
q^{(s-m)(\ell-m)}
\genfrac{[}{]}{0pt}{}{\ell}{m}_q
\N_q\bigl[(s-m,t-m)_\tau;
(k-\ell,0)_\varepsilon\bigr].
\]
As usual, a summand is understood to be zero if the indicated subcode
does not exist.
\end{theorem}

\begin{proof}
By Proposition~\ref{lcd-complement-types}, each of the codes
$C_+^{k,\ell}$ and $C_-^{k,\ell}$ admits a decomposition
\[
C_\pm^{k,\ell}
=
{\rm L}_\pm^{k-\ell}\oplus \Hull(C_\pm^{k,\ell}),
\]
where the LCD complement has the same type as the parent code. Applying
the projection-counting argument of Theorem~\ref{all-hull}, together
with Lemma~\ref{same-type}, gives the asserted formula for every
admissible pair of types.
\end{proof}

The right-hand side of Theorem~\ref{linear-code-count} depends only on
$q,(s,t)_\tau$, and $(k,\ell)_\varepsilon$. This establishes the
representative and length independence asserted at the beginning of
this section.

\section{Duality, Containment Enumeration, and Asymptotic Distributions}\label{sec-4}

This section studies containment enumeration and asymptotic distributions. Increasing the dimension of a known code is a common step in code constructions and motivates counting codes that contain a prescribed subcode while having a prescribed hull dimension and type. Duality identifies this containment problem with a subcode enumeration problem in the corresponding dual code, thereby providing explicit conditional mass formulas. We then determine the asymptotic distributions of prescribed-hull subcodes and of extensions containing a prescribed subcode. The common sequential convention for all asymptotic limits is introduced after the exact containment formula. The proofs of the main results in
Subsections~\ref{subsec-lcd-finite} and~\ref{subsec-asymptotic-subcodes}
are deferred to Appendix~A.

\subsection{Containment enumeration via duality}\label{subsec-containment-duality}

Let $D$ be a fixed subcode of $\F_q^n$. Motivated by the extension of a
given code to a code of larger dimension, we consider the family of
$[n,k]_q$ codes $C$ such that $D\subseteq C$ and $C$ has a prescribed
hull dimension and type. The identity
\[
D\subseteq C \quad\Longleftrightarrow\quad C^\perp\subseteq D^\perp,
\]
shows that the containment enumeration problem is dual to the subcode
enumeration problem solved above. We record the action of duality on
code types before stating the resulting correspondence.

\begin{Definition}\label{def-dual-type}
    Let $C$ be a linear code with $\dim(\operatorname{Hull}(C))=\ell$ and type $\tau$. We denote the type of $C^\perp$ by $\tau^\perp$.
    \begin{itemize}
        \item If $q$ is even, then
        \[
        ({\bf oo})^\perp={\bf oo},\qquad
        ({\bf oe})^\perp={\bf eo},\qquad
        ({\bf eo})^\perp={\bf oe},\qquad
        ({\bf ee})^\perp={\bf ee}.
        \]
        \item If $q$ is odd, Lemma~\ref{dual-type} gives the following correspondence:
        \begin{itemize}
            \item if $q\equiv 1\pmod 4$, or if $q\equiv 3\pmod 4$ and $\ell$ is even, then $\tau^\perp=\tau$;
            \item if $q\equiv 3\pmod 4$ and $\ell$ is odd, then the two types are interchanged, that is, $(+)^\perp=-$ and $(-)^\perp=+$.
        \end{itemize}
    \end{itemize}
\end{Definition}

Let $D$ have parameters $(s,t)_\tau$. In the notation of
Section~\ref{sec-1}, the number of extensions of type $\varepsilon$
and parameters $(k,\ell)$ is
$\N_{q,n}[(k,\ell)_\varepsilon\mid(s,t)_\tau]$. The next theorem
expresses this count as a parameterized subcode count in $D^\perp$.

The following theorem makes the containment--subcode correspondence
type-preserving after the dual-type involution is taken into account.

\begin{theorem}\label{thm-contain}
Let $D$ be an $[n,s]_q$ linear code with a $t$-dimensional hull and
type $\tau$, where $s\leq k\leq n$, and let
$0\leq\ell\leq\min\{k,n-k\}$. For every extension type
$\varepsilon$, one has
\[
\N_{q,n}\bigl[(k,\ell)_\varepsilon\bigm|(s,t)_\tau\bigr]
=
\N_q\bigl[(n-k,\ell)_{\varepsilon^\perp};
          (n-s,t)_{\tau^\perp}\bigr].
\]
Here $\tau^\perp$ and $\varepsilon^\perp$ are defined in
Definition~\ref{def-dual-type}. Combining this identity with
Corollaries~\ref{HLOO}, \ref{HLEO}, \ref{HLOE}, and \ref{HLEE} when
$q$ is even, and with Theorem~\ref{linear-code-count} when $q$ is odd,
gives explicit formulas for the containment count.
\end{theorem}
\begin{proof}
The equivalence $D\subseteq C$ if and only if
$C^\perp\subseteq D^\perp$ shows that the map $C\mapsto C^\perp$
sends the codes counted on the left to $(n-k)$-dimensional subcodes of
the $(n-s)$-dimensional code $D^\perp$. Moreover,
\[
\Hull(C^\perp)=C^\perp\cap C=\Hull(C),
\qquad
\Hull(D^\perp)=\Hull(D),
\]
so the counted subcode has hull dimension $\ell$ and its parent has
hull dimension $t$. By Definition~\ref{def-dual-type}, $C^\perp$ has
type $\varepsilon^\perp$, while $D^\perp$ has type $\tau^\perp$.
Since duality is an involution, this map is a bijection between the two
families. This proves the identity and also shows that the containment
count depends on the fixed code only through
$q,n,(s,t)_\tau$.
\end{proof}

\subsection{Finite-dimensional hull and type distributions in {\rm LCD} codes}
\label{subsec-lcd-finite}

Throughout the remainder of this section, the field size $q$ is fixed.
All limits are understood along admissible sequences of parameter
tuples. More precisely, a sequence
\[
\bigl((s_\nu,t_\nu)_{\tau_\nu};
      (k_\nu,\ell_\nu)_{\varepsilon_\nu}\bigr)_{\nu\geq1}
\]
is called admissible if, for every $\nu$, there exists a parent code of
type $\varepsilon_\nu$ with parameters $(k_\nu,\ell_\nu)$ containing
a subcode with the indicated parameters and type, and all dimension,
hull, type, and parity conditions in the statement under consideration
are satisfied. If a statement explicitly concerns a type or parity
class that is forced to be empty, its count is interpreted as zero and
only the parent tuple is required to be realizable. Whenever a limit is
displayed without the index $\nu$, it is understood in this sequential
sense, along every admissible sequence satisfying the stated
conditions. Parameters and types declared fixed in a result are
constant along the sequence.

By the enumeration results of Sections~\ref{sec-2} and~\ref{sec-3},
the value of
\[
\N_q\bigl[(s_\nu,t_\nu)_{\tau_\nu};
          (k_\nu,\ell_\nu)_{\varepsilon_\nu}\bigr]
\]
is independent of the particular representatives and their lengths.
Consequently, the asymptotic statements below depend only on the
displayed parameter tuples. In particular, expressions such as
\[
\lim_{\substack{s\to\infty\\k-\ell-s\to\infty}}
A(k,s,\ell)
\]
denote joint sequential limits, not iterated limits.

For every real number $q>1$ and integer $n\geq1$, define
\[
g_{q,n}:=\prod_{i=1}^n\left(1-\frac1{q^i}\right),
\qquad
h_{q,n}:=\prod_{i=1}^n(q^i-1),
\]
with the conventions $g_{q,0}=h_{q,0}=1$. The sequence
$(g_{q,n})_{n\geq1}$ is strictly decreasing and converges to a positive
limit, denoted by $g_{q,\infty}$.

The asymptotic results of Carlet \emph{et al.}~\cite{binaryLCD} and the
prescribed-hull extensions of Li \emph{et al.}~\cite{Li-hull-even,Li-hull-odd}
describe the regime in which both the subcode dimension and its
codimension in the parent tend to infinity. We first give a
finite-dimensional refinement for subcodes of a high-dimensional
{\rm LCD} parent: the integers $s$ and $t$ may remain fixed, and only
$k-s$ is required to tend to infinity. When $s$ is subsequently allowed
to grow, the resulting limits agree with the previously known
large-dimensional constants. In particular, for an even-characteristic
parent of type ${\bf oe}$ and an odd-characteristic parent of type $+$,
the joint large-dimensional limits recover the corresponding results of
Carlet \emph{et al.}; the formulas for arbitrary prescribed hull
dimension recover the extensions of Li \emph{et al.}

For $0\leq t\leq s$, put $s_0=s-t$.

\begin{prop}
\label{prop:asymptotic_limits}
Fix a nonnegative integer $t$. In parts~{\rm (i)} and~{\rm (ii)}, also
fix a nonnegative integer $s$ with $t\leq s$. All limits below are
taken according to the admissible-sequence convention above.
\begin{enumerate}
\item Suppose that $q$ is even.
\begin{enumerate}
    \item If $\varepsilon\in\{{\bf oo},{\bf oe}\}$, then
    \[
    \lim_{k-s\to\infty}
    \frac{\N_q\bigl[(s,t);(k,0)_\varepsilon\bigr]}
    {\genfrac{[}{]}{0pt}{}{k}{s}_q}
    =
    \frac{g_{q,s}}
    {h_{q,t}g_{q^2,\lfloor s_0/2\rfloor}}.
    \]
    Moreover,
    \[
    \lim_{\substack{k-s\to\infty\\s_0\ \text{is odd}}}
    \frac{\N_q\bigl[(s,t)_{{\bf oo}};(k,0)_\varepsilon\bigr]}
    {\genfrac{[}{]}{0pt}{}{k}{s}_q}
    =
    \frac{g_{q,s}}
    {h_{q,t}g_{q^2,(s_0-1)/2}},
    \]
    \[
    \lim_{\substack{k-s\to\infty\\s_0\ \text{is even}}}
    \frac{\N_q\bigl[(s,t)_{{\bf oo}};(k,0)_\varepsilon\bigr]}
    {\genfrac{[}{]}{0pt}{}{k}{s}_q}
    =
    \frac{(1-q^{-s_0})g_{q,s}}
    {h_{q,t}g_{q^2,s_0/2}},
    \]
   
    \[
    \lim_{\substack{k-s\to\infty\\s_0\ \text{is even}}}
    \frac{\N_q\bigl[(s,t)_{{\bf eo}};(k,0)_\varepsilon\bigr]}
    {\genfrac{[}{]}{0pt}{}{k}{s}_q}
    =
    \frac{q^{-s_0}g_{q,s}}
    {h_{q,t}g_{q^2,s_0/2}},
    \]
    and
    \[
    \lim_{k-s\to\infty}
    \frac{\N_q\bigl[(s,t)_{{\bf oe}};(k,0)_\varepsilon\bigr]}
    {\genfrac{[}{]}{0pt}{}{k}{s}_q}
    =
    \lim_{k-s\to\infty}
    \frac{\N_q\bigl[(s,t)_{{\bf ee}};(k,0)_\varepsilon\bigr]}
    {\genfrac{[}{]}{0pt}{}{k}{s}_q}
    =0.
    \]

    \item If the parent type is ${\bf eo}$, then $s_0$ must be even, and
    \[
    \lim_{\substack{k-s\to\infty\\s_0\ \text{is even}}}
    \frac{\N_q\bigl[(s,t);(k,0)_{{\bf eo}}\bigr]}
    {\genfrac{[}{]}{0pt}{}{k}{s}_q}
    =
    \lim_{\substack{k-s\to\infty\\s_0\ \text{is even}}}
    \frac{\N_q\bigl[(s,t)_{{\bf eo}};(k,0)_{{\bf eo}}\bigr]}
    {\genfrac{[}{]}{0pt}{}{k}{s}_q}
    =
    \frac{q^t g_{q,s}}
    {h_{q,t}g_{q^2,s_0/2}},
    \]
    
    while the other three complete-type proportions are zero.
\end{enumerate}

\item Suppose that $q$ is odd and
$\varepsilon\in\{+,-\}$. Then
\[
\lim_{k-s\to\infty}
\frac{\N_q\bigl[(s,t);(k,0)_\varepsilon\bigr]}{\genfrac{[}{]}{0pt}{}{k}{s}_q}
=
\frac{g_{q,s}}
{h_{q,t}g_{q^2,\lfloor s_0/2\rfloor}}.
\]
\[
\lim_{\substack{k-s\to\infty\\s_0\ \text{is odd}}}
\frac{\N_q\bigl[(s,t)_+;(k,0)_\varepsilon\bigr]}{\genfrac{[}{]}{0pt}{}{k}{s}_q}
=
\lim_{\substack{k-s\to\infty\\s_0\ \text{is odd}}}
\frac{\N_q\bigl[(s,t)_{-};(k,0)_\varepsilon\bigr]}{\genfrac{[}{]}{0pt}{}{k}{s}_q}
=
\frac{g_{q,s}}
{2h_{q,t}g_{q^2,(s_0-1)/2}}.
\]
\[
\lim_{\substack{k-s\to\infty\\s_0\ \text{is even}}}
\frac{\N_q\bigl[(s,t)_+;(k,0)_\varepsilon\bigr]}{\genfrac{[}{]}{0pt}{}{k}{s}_q}
=
\frac{\bigl(1+\eta((-1)^{s_0/2})q^{-s_0/2}\bigr)g_{q,s}}
{2h_{q,t}g_{q^2,s_0/2}},
\]
and
\[
\lim_{\substack{k-s\to\infty\\s_0\ \text{is even}}}
\frac{\N_q\bigl[(s,t)_{-};(k,0)_\varepsilon\bigr]}{\genfrac{[}{]}{0pt}{}{k}{s}_q}
=
\frac{\bigl(1-\eta((-1)^{s_0/2})q^{-s_0/2}\bigr)g_{q,s}}
{2h_{q,t}g_{q^2,s_0/2}}.
\]

\item Fix $t$ and let both $s$ and $k-s$ tend to infinity. Then the
finite-dimensional distributions pass to the familiar
large-dimensional limits:
\begin{enumerate}
    \item if $q$ is even and
    $\varepsilon\in\{{\bf oo},{\bf oe}\}$, then
    \[
    \lim_{\substack{s\to\infty\\k-s\to\infty}}
    \frac{\N_q\bigl[(s,t);(k,0)_\varepsilon\bigr]}
    {\genfrac{[}{]}{0pt}{}{k}{s}_q}
    =
    \lim_{\substack{s\to\infty\\k-s\to\infty}}
    \frac{\N_q\bigl[(s,t)_{{\bf oo}};(k,0)_\varepsilon\bigr]}
    {\genfrac{[}{]}{0pt}{}{k}{s}_q}
    =
    \frac{g_{q,\infty}}
    {h_{q,t}g_{q^2,\infty}};
    \]
    \item if $q$ is even and the parent type is ${\bf eo}$, then
    \[
    \lim_{\substack{s\to\infty\\k-s\to\infty\\s-t\ \text{is even}}}
    \frac{\N_q\bigl[(s,t);(k,0)_{{\bf eo}}\bigr]}
    {\genfrac{[}{]}{0pt}{}{k}{s}_q}
    =
    \lim_{\substack{s\to\infty\\k-s\to\infty\\s-t\ \text{is even}}}
    \frac{\N_q\bigl[(s,t)_{{\bf eo}};(k,0)_{{\bf eo}}\bigr]}
    {\genfrac{[}{]}{0pt}{}{k}{s}_q}
    =
    \frac{q^t g_{q,\infty}}
    {h_{q,t}g_{q^2,\infty}};
    \]
    \item if $q$ is odd and $\varepsilon\in\{+,-\}$, then
    \[
    \lim_{\substack{s\to\infty\\k-s\to\infty}}
    \frac{\N_q\bigl[(s,t);(k,0)_\varepsilon\bigr]}{\genfrac{[}{]}{0pt}{}{k}{s}_q}
    =
    \frac{g_{q,\infty}}
    {h_{q,t}g_{q^2,\infty}},
    \]
    and
    \[
    \lim_{\substack{s\to\infty\\k-s\to\infty}}
    \frac{\N_q\bigl[(s,t)_+;(k,0)_\varepsilon\bigr]}{\genfrac{[}{]}{0pt}{}{k}{s}_q}
    =
    \lim_{\substack{s\to\infty\\k-s\to\infty}}
    \frac{\N_q\bigl[(s,t)_{-};(k,0)_\varepsilon\bigr]}{\genfrac{[}{]}{0pt}{}{k}{s}_q}
    =
    \frac{g_{q,\infty}}
    {2h_{q,t}g_{q^2,\infty}}.
    \]
\end{enumerate}
\end{enumerate}
\end{prop}

\begin{remarknonumber}
The first two parts of Proposition~\ref{prop:asymptotic_limits} keep
$s$ finite and require only $k-s\to\infty$. Thus more than one type
may retain a nonzero limiting proportion. Part~{\rm (iii)} concerns the
different joint limit in which $s\to\infty$ as well; in that regime
type ${\bf oo}$ is dominant for an odd-like parent, type ${\bf eo}$ is
the only possible type for an even-like {\rm LCD} parent, and the two
odd-characteristic types are asymptotically balanced.
\end{remarknonumber}

\subsection{Asymptotic distribution of subcodes with prescribed hull dimensions}\label{subsec-asymptotic-subcodes}

We now combine the {\rm LCD} limits of
Proposition~\ref{prop:asymptotic_limits} with the contribution of the
hull of an arbitrary fixed parent code.

For an {\rm LCD} parent code, $\ell=0$ and hence
$a=k-\ell-s=k-s$. Thus Proposition~\ref{prop:asymptotic_limits}
describes the regime $a\to+\infty$. The corresponding result for a
general parent code is given below.

In the exact enumeration formulas, the summation index $m$ records the
dimension of the intersection between the subcode and the hull of the
fixed code. The following estimate isolates the asymptotically dominant
projection term.

\begin{lem}\label{lem:dominant-projection-term}
Let $m$ be a fixed nonnegative integer with $0\leq m\leq\ell$, and
define
\[
W_m(k,s,\ell)
:=
\frac{q^{(s-m)(\ell-m)}\genfrac{[}{]}{0pt}{}{\ell}{m}_q
\genfrac{[}{]}{0pt}{}{k-\ell}{s-m}_q}{\genfrac{[}{]}{0pt}{}{k}{s}_q}.
\]
If both $s$ and $k-\ell-s$ tend to infinity, then
\[
\lim_{\substack{s\to\infty\\k-\ell-s\to\infty}}
W_0(k,s,\ell)=1,
\qquad
\lim_{\substack{s\to\infty\\k-\ell-s\to\infty}}
W_m(k,s,\ell)=0\quad\text{if $m\geq1$}.
\]
\end{lem}

 Since $t$ is fixed in the
applications below, the number of terms
$0\leq m\leq\min\{t,\ell\}$ is uniformly bounded, and the termwise
limits may be summed. Thus, after normalization by
$\genfrac{[}{]}{0pt}{}{k}{s}_q$, the term $m=0$ has weight tending to one, whereas every
fixed term with $m\geq1$ has weight tending to zero whenever
$s\to\infty$ and $k-\ell-s\to\infty$.

This estimate reflects a basic finite-geometric phenomenon. Let $V$ be
an $n$-dimensional vector space over $\F_q$, and let $W$ be a fixed
$k$-dimensional subspace of $V$. The proportion of $s$-dimensional
subspaces $U\subseteq V$ having a nonzero intersection with $W$ tends
to zero, even when $k$ varies, provided that $n-k-s\to\infty$. In our
setting, this phenomenon is applied with
$V=C$ and $W=\operatorname{Hull}(C)$; hence the corresponding condition
is $k-\ell-s\to\infty$. Thus, after normalization by $\genfrac{[}{]}{0pt}{}{k}{s}_q$,
the leading contribution comes from the terms with
$D\cap\operatorname{Hull}(C)=\{\mathbf 0\}$. We can now combine this
observation with the asymptotic limits for {\rm LCD} codes and the
exact projection formulas.

\begin{theorem}\label{thm:asymptotic}
Fix a nonnegative integer $t$.  Suppose that $t\leq s\leq k$. When both $s$ and $k-\ell-s$ tend to infinity, the following statements hold.

\begin{enumerate}
    \item Suppose that $q$ is even.
    \begin{enumerate}
        \item Let $\varepsilon\in\{{\bf oo},{\bf oe}\}$. Then
        \[
        \lim_{\substack{s\to\infty\\k-\ell-s\to\infty}}
        \frac{\N_q\bigl[(s,t);(k,\ell)_\varepsilon\bigr]}
        {\genfrac{[}{]}{0pt}{}{k}{s}_q}
        =
        \lim_{\substack{s\to\infty\\k-\ell-s\to\infty}}
        \frac{\N_q\bigl[(s,t)_{{\bf oo}};
        (k,\ell)_\varepsilon\bigr]}
        {\genfrac{[}{]}{0pt}{}{k}{s}_q}
        =
        \frac{g_{q,\infty}}
        {h_{q,t}\cdot g_{q^2,\infty}}.
        \]

        \item Let $\varepsilon\in\{{\bf eo},{\bf ee}\}$. Then
        \[
        \lim_{\substack{s\to\infty\\k-\ell-s\to\infty\\s-t\ \text{is even}}}
        \frac{\N_q\bigl[(s,t);(k,\ell)_\varepsilon\bigr]}
        {\genfrac{[}{]}{0pt}{}{k}{s}_q}
        =
        \lim_{\substack{s\to\infty\\k-\ell-s\to\infty\\s-t\ \text{is even}}}
        \frac{\N_q\bigl[(s,t)_{{\bf eo}};
        (k,\ell)_\varepsilon\bigr]}
        {\genfrac{[}{]}{0pt}{}{k}{s}_q}
        =
        \frac{q^t\cdot g_{q,\infty}}
        {h_{q,t}\cdot g_{q^2,\infty}}.
        \]
    \end{enumerate}

    \item Suppose that $q$ is odd and
    $\varepsilon\in\{+,-\}$. Then
    \[
    \lim_{\substack{s\to\infty\\k-\ell-s\to\infty}}
    \frac{\N_q\bigl[(s,t)_+;(k,\ell)_\varepsilon\bigr]}
    {\genfrac{[}{]}{0pt}{}{k}{s}_q}
    =
    \lim_{\substack{s\to\infty\\k-\ell-s\to\infty}}
    \frac{\N_q\bigl[(s,t)_{-};(k,\ell)_\varepsilon\bigr]}
    {\genfrac{[}{]}{0pt}{}{k}{s}_q}
    =
    \frac{g_{q,\infty}}
    {2h_{q,t}\cdot g_{q^2,\infty}}.
    \]
\end{enumerate}

\end{theorem}

\begin{remarknonumber}
Theorem~\ref{thm:asymptotic} describes the asymptotic hull proportions
along an admissible sequence of parent codes. For each member of the
sequence, the quantity
$
\frac{\N_q\bigl[(s,t);(k,\ell)_\varepsilon\bigr]}
{\genfrac{[}{]}{0pt}{}{k}{s}_q}
$
is the proportion of its $s$-dimensional subcodes whose hull dimension
is $t$. The
theorem shows that, under the condition
\[
s\to\infty,\qquad k-\ell-s\to\infty,
\]
this proportion has the same leading limit as in the corresponding
{\rm LCD} component of $C$.

In even characteristic, the dominant type is determined by the parity
of the parent code: for odd-like parent codes, almost all prescribed-hull
subcodes are of type ${\bf oo}$, while for even-like parent codes,
almost all prescribed-hull subcodes are of type ${\bf eo}$. In odd
characteristic, the two types $+$ and $-$ are asymptotically balanced.
Thus the hull part of the parent code affects the exact finite formulas,
but it does not change the leading asymptotic distribution, provided
that the {\rm LCD} component has sufficiently large remaining dimension.
\end{remarknonumber}
With the unified parameter $a=k-\ell-s$, the preceding theorem treats
the regime $a\to+\infty$. We now consider $a\leq0$. For every
$s$-dimensional subcode $D\subseteq C$, the projection identity gives
\[
s-\dim\operatorname{Hull}(D)
=
\dim\Pi_{\rm L}(D)-\dim\operatorname{Hull}(\Pi_{\rm L}(D))
\leq k-\ell.
\]
Thus $\dim\operatorname{Hull}(D)\geq s-(k-\ell)=-a$. The following theorem
counts the subcodes attaining this minimum and describes their
asymptotic distribution.

\begin{theorem}
\label{thm:asymptotic-finite}
Put
\[
k_0=k-\ell,
\qquad
a=k_0-s=k-\ell-s,
\]
and suppose that $a\leq0$.
Let $\varepsilon$ be any admissible parent type.
\begin{enumerate}
    \item If $a$ is fixed and both $k_0$ and $\ell$ tend to infinity,
    then
    \[
    \lim_{\substack{k_0\to\infty\\\ell\to\infty\\s=k_0-a}}
    \frac{\N_q\bigl[(s,-a);(k,\ell)_\varepsilon\bigr]}
    {\genfrac{[}{]}{0pt}{}{k}{s}_q}
    =
    \frac{g_{q,\infty}}{g_{q,-a}}.
    \]
    In particular, in the case $s=k-\ell$,
    \[
    \lim_{\substack{k-\ell\to\infty\\\ell\to\infty\\s=k-\ell}}
    \frac{\N_q\bigl[(s,0);(k,\ell)_\varepsilon\bigr]}
    {\genfrac{[}{]}{0pt}{}{k}{s}_q}
    =g_{q,\infty}.
    \]
 \item If $a=k-\ell-s\to-\infty$, then
    \[
    \lim_{a\to-\infty}
    \frac{\N_q\bigl[(s,-a);(k,\ell)_\varepsilon\bigr]}
    {\genfrac{[}{]}{0pt}{}{k}{s}_q}
    =1.
    \]

    \item Under the assumptions of part~{\rm (ii)}, assume in addition
    that $k-s\to\infty$. If $q$ is even, the normalized count of
    type ${\bf oo}$ tends to one for
    $\varepsilon\in\{{\bf oo},{\bf oe}\}$, while that of type
    ${\bf eo}$ tends to one for
    $\varepsilon\in\{{\bf eo},{\bf ee}\}$. If $q$ is odd, the
    dominant minimum-hull subcodes have type $\varepsilon$.
\end{enumerate}
\end{theorem}

Theorem~\ref{thm:asymptotic-finite} treats $a\leq0$, including the
regime $a\to-\infty$. We next consider fixed $a>0$. Although the
complete hull distribution may involve several projection terms in
this range, the minimum-hull contribution has a simple limiting form.

\begin{theorem}
\label{thm:bounded-gap-minimum}
Put
\[
k_0=k-\ell,
\qquad
a=k_0-s=k-\ell-s,
\]
and suppose that $a$ is a fixed positive integer and that both $k_0$ and
$\ell$ tend to infinity.
\begin{enumerate}
    \item Suppose that $q$ is even.
    \begin{enumerate}
        \item If $\varepsilon\in\{{\bf oo},{\bf oe}\}$, then the
        minimum possible hull dimension is zero and
        \[
        \lim_{\substack{k_0\to\infty\\\ell\to\infty\\s=k_0-a}}
        \frac{\N_q\bigl[(s,0);(k,\ell)_\varepsilon\bigr]}
        {\genfrac{[}{]}{0pt}{}{k}{s}_q}
        =
        \lim_{\substack{k_0\to\infty\\\ell\to\infty\\s=k_0-a}}
        \frac{\N_q\bigl[(s,0)_{{\bf oo}};
        (k,\ell)_\varepsilon\bigr]}
        {\genfrac{[}{]}{0pt}{}{k}{s}_q}
        =
        \frac{g_{q,\infty}}
        {g_{q^2,\lfloor a/2\rfloor}}.
        \]

        \item If $\varepsilon\in\{{\bf eo},{\bf ee}\}$ and $a$ is
        even, then the minimum possible hull dimension is zero and
        \[
        \lim_{\substack{k_0\to\infty\\\ell\to\infty\\s=k_0-a}}
        \frac{\N_q\bigl[(s,0);(k,\ell)_\varepsilon\bigr]}
        {\genfrac{[}{]}{0pt}{}{k}{s}_q}
        =
        \lim_{\substack{k_0\to\infty\\\ell\to\infty\\s=k_0-a}}
        \frac{\N_q\bigl[(s,0)_{{\bf eo}};
        (k,\ell)_\varepsilon\bigr]}
        {\genfrac{[}{]}{0pt}{}{k}{s}_q}
        =
        \frac{g_{q,\infty}}
        {g_{q^2,a/2}}.
        \]

        \item If $\varepsilon\in\{{\bf eo},{\bf ee}\}$ and $a$ is
        odd, then the minimum possible hull dimension is one and
        \[
        \lim_{\substack{k_0\to\infty\\\ell\to\infty\\s=k_0-a}}
        \frac{\N_q\bigl[(s,1);(k,\ell)_\varepsilon\bigr]}
        {\genfrac{[}{]}{0pt}{}{k}{s}_q}
        =
        \lim_{\substack{k_0\to\infty\\\ell\to\infty\\s=k_0-a}}
        \frac{\N_q\bigl[(s,1)_{{\bf eo}};
        (k,\ell)_\varepsilon\bigr]}
        {\genfrac{[}{]}{0pt}{}{k}{s}_q}
        =
        \frac{g_{q,\infty}}
        {g_{q,1}g_{q^2,(a+1)/2}}.
        \]
    \end{enumerate}

    \item Suppose that $q$ is odd and $\varepsilon\in\{+,-\}$. Then
    the minimum possible hull dimension is zero and
    \[
    \lim_{\substack{k_0\to\infty\\\ell\to\infty\\s=k_0-a}}
    \frac{\N_q\bigl[(s,0);(k,\ell)_\varepsilon\bigr]}
    {\genfrac{[}{]}{0pt}{}{k}{s}_q}
    =
    \frac{g_{q,\infty}}
    {g_{q^2,\lfloor a/2\rfloor}}.
    \]
\end{enumerate}
\end{theorem}

\begin{remarknonumber}
When $a=k-\ell-s$ is fixed and positive, the full prescribed-hull distribution can
still be obtained from the exact projection formulas. In general,
however, several intersection dimensions have nonzero limiting
contributions. For the minimum hull dimension the sum reduces to one
term, except in the odd-dimensional alternating case, where exactly
two terms contribute. This is why the minimum-hull distribution admits
the simpler formulas in Theorem~\ref{thm:bounded-gap-minimum}.

In odd characteristic there is, in general, no corresponding single
dominant type when $a>0$ is fixed. Internal orthogonal complementation
reduces the type distribution to the finite-dimensional type limits
in Proposition~\ref{prop:asymptotic_limits}. Both types may therefore
retain nonzero limiting proportions; when the relevant nondegenerate
dimension is odd they are balanced, while in even dimension their two
limits contain the nonzero correction factors
$1\pm\eta((-1)^{a/2})q^{-a/2}$, with the appropriate parity adjustment.
\end{remarknonumber}

\subsection{Asymptotic distribution under a fixed containment constraint}
\label{subsec-asymptotic-containment}

We now transfer the preceding asymptotic results to the containment
problem. Let $D$ be an $[n,s]_q$ linear code with a $t$-dimensional
hull and type $\tau$, and let $s\leq k\leq n$. The preceding sequence
convention extends to containment tuples
$(n,s,t,k,\ell,\tau,\varepsilon)$: the length $n$ may vary, and for
each index the fixed contained code and the counted extension
parameters must satisfy the stated feasibility conditions. By
Theorem~\ref{thm-contain}, for every admissible hull dimension $u$ and
every extension type $\varepsilon$,
\[
\N_{q,n}\bigl[(k,u)_\varepsilon\bigm|(s,t)_\tau\bigr]
=
\N_q\bigl[(n-k,u)_{\varepsilon^\perp};
          (n-s,t)_{\tau^\perp}\bigr],
\]
and
\[
\genfrac{[}{]}{0pt}{}{n-s}{k-s}_q
=
\genfrac{[}{]}{0pt}{}{n-s}{n-k}_q.
\]
The code $D^\perp$ has dimension $n-s$, hull dimension $t$, and an
{\rm LCD} complement of dimension $n-s-t$. Under the
containment--subcode correspondence, the selected subcode $C^\perp$
has dimension $n-k$. Consequently, the difference parameter for the
subcode problem inside $D^\perp$ is
\[
b=(n-s-t)-(n-k)=k-s-t.
\]
Thus $b=(k-s)-t$ measures the difference between the extension
dimension $k-s$ and the hull dimension $t$ of $D$.

When $q$ is even, $D^\perp$ is odd-like if
$\mathbf{1}_n\notin D$ and even-like if
$\mathbf{1}_n\in D$. The relevant complete types are transferred by
\[
({\bf oo})^\perp={\bf oo},
\qquad
({\bf eo})^\perp={\bf oe}.
\]
Thus the three ranges $b\to+\infty$, fixed $b>0$, and $b\leq0$
follow directly from
Theorems~\ref{thm:asymptotic}, \ref{thm:asymptotic-finite} and
\ref{thm:bounded-gap-minimum}
 respectively. We state the resulting
containment distributions without separate proofs.

\begin{cor}
\label{cor-asymptotic-contain}
Let $D$ be an $[n,s]_q$ linear code with a $t$-dimensional hull and
type $\tau$, and fix a nonnegative integer $\ell$. Suppose that both $n-k$ and
$k-s-t$ tend to infinity.
\begin{enumerate}
    \item Suppose that $q$ is even.
    \begin{enumerate}
        \item If $\mathbf{1}_n\notin D$, then
        \[
        \lim_{\substack{n-k\to\infty\\
                        k-s-t\to\infty}}
        \frac{\N_{q,n}\bigl[(k,\ell)\bigm|(s,t)_\tau\bigr]}
        {\genfrac{[}{]}{0pt}{}{n-s}{k-s}_q}
        =
        \lim_{\substack{n-k\to\infty\\
                        k-s-t\to\infty}}
        \frac{\N_{q,n}\bigl[(k,\ell)_{{\bf oo}}
        \bigm|(s,t)_\tau\bigr]}
        {\genfrac{[}{]}{0pt}{}{n-s}{k-s}_q}
        =
        \frac{g_{q,\infty}}
        {h_{q,\ell}g_{q^2,\infty}}.
        \]

        \item If $\mathbf{1}_n\in D$, then
        \[
        \lim_{\substack{n-k\to\infty\\
                        k-s-t\to\infty\\
                        n-k-\ell\ \text{is even}}}
        \frac{\N_{q,n}\bigl[(k,\ell)\bigm|(s,t)_\tau\bigr]}
        {\genfrac{[}{]}{0pt}{}{n-s}{k-s}_q}
        =
        \lim_{\substack{n-k\to\infty\\
                        k-s-t\to\infty\\
                        n-k-\ell\ \text{is even}}}
        \frac{\N_{q,n}\bigl[(k,\ell)_{{\bf oe}}
        \bigm|(s,t)_\tau\bigr]}
        {\genfrac{[}{]}{0pt}{}{n-s}{k-s}_q}
        =
        \frac{q^\ell g_{q,\infty}}
        {h_{q,\ell}g_{q^2,\infty}}.
        \]
    \end{enumerate}

    \item Suppose that $q$ is odd. Whether $D$ is of type $+$ or
    type $-$,
    \[
    \lim_{\substack{n-k\to\infty\\
                    k-s-t\to\infty}}
    \frac{\N_{q,n}\bigl[(k,\ell)\bigm|(s,t)_\tau\bigr]}
    {\genfrac{[}{]}{0pt}{}{n-s}{k-s}_q}
    =
    \frac{g_{q,\infty}}
    {h_{q,\ell}g_{q^2,\infty}},
    \]
    and the two extension types are asymptotically balanced:
    \[
    \lim_{\substack{n-k\to\infty\\
                    k-s-t\to\infty}}
    \frac{\N_{q,n}\bigl[(k,\ell)_+\bigm|(s,t)_\tau\bigr]}
    {\genfrac{[}{]}{0pt}{}{n-s}{k-s}_q}
    =
    \lim_{\substack{n-k\to\infty\\
                    k-s-t\to\infty}}
    \frac{\N_{q,n}\bigl[(k,\ell)_{-}\bigm|(s,t)_\tau\bigr]}
    {\genfrac{[}{]}{0pt}{}{n-s}{k-s}_q}
    =
    \frac{g_{q,\infty}}
    {2h_{q,\ell}g_{q^2,\infty}}.
    \]
\end{enumerate}
\end{cor}

The preceding corollary treats the range in which the number of added
dimensions exceeds $t$ by an unbounded amount. We next consider the
transition range in which this excess is fixed.

\begin{cor}
\label{cor:bounded-gap-containment}
Let $D$ be an $[n,s]_q$ linear code with a $t$-dimensional hull and
type $\tau$, and let $s\leq k\leq n$. Suppose that
\[
b=k-s-t
\]
is a fixed positive integer and that both $n-s-t$ and $t$ tend to
infinity.
\begin{enumerate}
    \item Suppose that $q$ is even.
    \begin{enumerate}
        \item If $\mathbf{1}_n\notin D$, then the minimum possible hull
        dimension of an extension is zero and
        \[
        \lim_{\substack{n-s-t\to\infty\\
                        t\to\infty\\
                        k-s-t=b}}
        \frac{\N_{q,n}\bigl[(k,0)\bigm|(s,t)_\tau\bigr]}
        {\genfrac{[}{]}{0pt}{}{n-s}{k-s}_q}
        =
        \lim_{\substack{n-s-t\to\infty\\
                        t\to\infty\\
                        k-s-t=b}}
        \frac{\N_{q,n}\bigl[(k,0)_{{\bf oo}}
        \bigm|(s,t)_\tau\bigr]}
        {\genfrac{[}{]}{0pt}{}{n-s}{k-s}_q}
        =
        \frac{g_{q,\infty}}
        {g_{q^2,\lfloor b/2\rfloor}}.
        \]

        \item If $\mathbf{1}_n\in D$ and $b$ is even, then the minimum
        possible hull dimension of an extension is zero and
        \[
        \lim_{\substack{n-s-t\to\infty\\
                        t\to\infty\\
                        k-s-t=b}}
        \frac{\N_{q,n}\bigl[(k,0)\bigm|(s,t)_\tau\bigr]}
        {\genfrac{[}{]}{0pt}{}{n-s}{k-s}_q}
        =
        \lim_{\substack{n-s-t\to\infty\\
                        t\to\infty\\
                        k-s-t=b}}
        \frac{\N_{q,n}\bigl[(k,0)_{{\bf oe}}
        \bigm|(s,t)_\tau\bigr]}
        {\genfrac{[}{]}{0pt}{}{n-s}{k-s}_q}
        =
        \frac{g_{q,\infty}}
        {g_{q^2,b/2}}.
        \]

        \item If $\mathbf{1}_n\in D$ and $b$ is odd, then the minimum
        possible hull dimension of an extension is one and
        \[
        \lim_{\substack{n-s-t\to\infty\\
                        t\to\infty\\
                        k-s-t=b}}
        \frac{\N_{q,n}\bigl[(k,1)\bigm|(s,t)_\tau\bigr]}
        {\genfrac{[}{]}{0pt}{}{n-s}{k-s}_q}
        =
        \lim_{\substack{n-s-t\to\infty\\
                        t\to\infty\\
                        k-s-t=b}}
        \frac{\N_{q,n}\bigl[(k,1)_{{\bf oe}}
        \bigm|(s,t)_\tau\bigr]}
        {\genfrac{[}{]}{0pt}{}{n-s}{k-s}_q}
        =
        \frac{g_{q,\infty}}
        {g_{q,1}g_{q^2,(b+1)/2}}.
        \]
    \end{enumerate}

    \item Suppose that $q$ is odd. Then the minimum possible hull
    dimension of an extension is zero and
    \[
    \lim_{\substack{n-s-t\to\infty\\
                    t\to\infty\\
                    k-s-t=b}}
    \frac{\N_{q,n}\bigl[(k,0)\bigm|(s,t)_\tau\bigr]}
    {\genfrac{[}{]}{0pt}{}{n-s}{k-s}_q}
    =
    \frac{g_{q,\infty}}
    {g_{q^2,\lfloor b/2\rfloor}}.
    \]
\end{enumerate}
\end{cor}

We finally consider $b\leq0$. In this range, every extension
$C\supseteq D$ satisfies
\[
\dim\operatorname{Hull}(C)\geq -b.
\]
Thus $-b$ is the minimum possible hull dimension. The following
corollary gives the proportion of extensions attaining this minimum.

\begin{cor}
\label{cor-asymptotic-contain-finite}
Let $D$ be an $[n,s]_q$ linear code with a $t$-dimensional hull and
type $\tau$, and let $s\leq k\leq n$. Put
\[
b=k-s-t,
\]
and suppose that $b\leq0$.
\begin{enumerate}
    \item If $b$ is fixed and both $n-s-t$ and $t$ tend to infinity,
    then
    \[
    \lim_{\substack{n-s-t\to\infty\\
                    t\to\infty\\
                    k-s-t=b}}
    \frac{\N_{q,n}\bigl[(k,-b)\bigm|(s,t)_\tau\bigr]}
    {\genfrac{[}{]}{0pt}{}{n-s}{k-s}_q}
    =
    \frac{g_{q,\infty}}
    {g_{q,-b}}.
    \]

    \item If $k-s-t\to-\infty$, then
    \[
    \lim_{k-s-t\to-\infty}
    \frac{\N_{q,n}\bigl[(k,-b)\bigm|(s,t)_\tau\bigr]}
    {\genfrac{[}{]}{0pt}{}{n-s}{k-s}_q}
    =1.
    \]

    \item Under the assumptions of part~{\rm (ii)}, assume in addition
    that $k-s\to\infty$. If $q$ is even, then
    \[
    \lim_{\substack{k-s-t\to-\infty\\
                    k-s\to\infty}}
    \frac{\N_{q,n}\bigl[(k,-b)_{{\bf oo}}
    \bigm|(s,t)_\tau\bigr]}
    {\genfrac{[}{]}{0pt}{}{n-s}{k-s}_q}
    =1
    \]
    if $\mathbf{1}_n\notin D$, whereas
    \[
    \lim_{\substack{k-s-t\to-\infty\\
                    k-s\to\infty}}
    \frac{\N_{q,n}\bigl[(k,-b)_{{\bf oe}}
    \bigm|(s,t)_\tau\bigr]}
    {\genfrac{[}{]}{0pt}{}{n-s}{k-s}_q}
    =1
    \]
    if $\mathbf{1}_n\in D$.

    If $q$ is odd and $D$ has type $\tau$, then the dominant extensions
    have type $\tau$ when $q\equiv1\pmod4$. When
    $q\equiv3\pmod4$, they have type $\tau$ if $k-s$ is even and the
    opposite type if $k-s$ is odd.
\end{enumerate}
\end{cor}
The parameter \(b=(k-s)-t\) therefore organizes the principal extension regimes considered above.
When \(b\to+\infty\), the extension distribution is governed by the {\rm LCD} part of \(D^\perp\).
For fixed \(b>0\), the minimum-hull proportion has a nonzero explicit limit, with an additional
parity obstruction for even-like codes in even characteristic. When \(b\leq0\), at least \(-b\)
hull dimensions must remain, and as \(b\to-\infty\), almost all extensions attain this minimum.
\section{Conclusion and perspectives}\label{sec-conclusion}
In this paper, we determined the number of subcodes with prescribed
dimensions and hull dimensions inside a fixed linear code over finite
fields of both even and odd characteristic.
In even characteristic, the enumeration was further refined according
to the four complete types, whereas in odd characteristic the types
$+$ and $-$ were considered.
The resulting parameterized count
\[
\N_q\bigl[(s,t)_\tau;(k,\ell)_\varepsilon\bigr]
\]
depends only on $q,(s,t)_\tau$, and $(k,\ell)_\varepsilon$, and not on
the particular parent code or its length.

By decomposing the fixed code into an {\rm LCD} part and its hull, we
reduced the enumeration problem to the corresponding formulas for
{\rm LCD} codes. This approach recovers the usual whole-space mass
formulas as a special case. Motivated by dimension-extending procedures
in code constructions, we also formulated the problem of counting codes
with a prescribed hull dimension that contain a fixed code. The
containment--subcode duality then yields explicit conditional mass
formulas for
$\N_{q,n}[(k,\ell)_\varepsilon\mid(s,t)_\tau]$.

The asymptotic analysis is organized by the single difference
$a=k-\ell-s$. For an {\rm LCD} parent, fixed $s$ and $t$ with
$k-s\to\infty$ give explicit limiting formulas for the hull and
complete-type proportions. When $s$ also tends to infinity, these
formulas recover the corresponding large-dimensional constants of
Carlet \emph{et al.} and Li \emph{et al.} For a general parent, when
both $s$ and $a$ grow, the limiting distribution is
governed by the {\rm LCD} component; in even characteristic one type is
dominant, whereas in odd characteristic the types $+$ and $-$ are
asymptotically balanced. When $a>0$ is fixed and both $k-\ell$ and
$\ell$ grow, the minimum-hull proportion has the explicit limits in
Theorem~\ref{thm:bounded-gap-minimum}. When $a\leq0$, the minimum hull dimension is $-a$; for fixed $a$,
when both $k-\ell$ and $\ell$ tend to infinity, its proportion tends to
$g_{q,\infty}/g_{q,-a}$, while for $a\to-\infty$ almost all subcodes
attain this minimum. The corresponding fixed-difference and
concentration results for extensions are governed, through duality, by
the parameter $b=k-s-t$.

It would be natural to extend this framework to Hermitian and
symplectic hulls. Such extensions would require the corresponding
decompositions, type classifications, and mass formulas under these
inner products.

\appendix
\section{Proofs}
In this appendix, we give the proofs of the main results stated in
Subsections~\ref{subsec-lcd-finite} and~\ref{subsec-asymptotic-subcodes}.
\begin{proof}
The case $s=0$ is immediate, so we may assume $s\geq1$.
We first establish the total finite-dimensional limit in the whole
vector space. Put $s_0=s-t$. The exact prescribed-hull formulas in
Lemma~\ref{thm:even-whole-space-mass} when $q$ is even and
Lemma~\ref{whole-space-count-odd} when $q$ is odd give
\[
\lim_{k-s\to\infty}
\frac{\operatorname{N}^t(s;\F_q^k)}
{\operatorname{N}^0(s_0;\F_q^k)}
q^{-t(k-2s_0)+t^2}
=\frac1{h_{q,t}}.
\]
The corresponding {\rm LCD} formulas and the Gaussian product formula give
\[
\lim_{k-s\to\infty}
\frac{\operatorname{N}^0(s_0;\F_q^k)}
{q^{s_0(k-s_0)}}
=
\frac1{g_{q^2,\lfloor s_0/2\rfloor}},
\quad
\lim_{k-s\to\infty}
\frac{\genfrac{[}{]}{0pt}{}{k}{s}_q}{q^{s(k-s)}}
=\frac1{g_{q,s}}.
\]
Since
\[
t(k-2s_0)-t^2+s_0(k-s_0)=s(k-s),
\]
these three limits yield
\[
\lim_{k-s\to\infty}
\frac{\operatorname{N}^t(s;\F_q^k)}{\genfrac{[}{]}{0pt}{}{k}{s}_q}
=
\frac{g_{q,s}}
{h_{q,t}g_{q^2,\lfloor s_0/2\rfloor}}.
\]

Suppose that $q$ is even. Proposition~\ref{oe number} identifies all
four complete-type counts in ${\rm L}_{\mathbf{oe}}^k$ with the
corresponding whole-space counts, while Corollary~\ref{oe=oo} gives the
same total and parity-label counts for ${\rm L}_{\mathbf{oo}}^k$.
The even-like whole-space formulas give
\[
\lim_{\substack{k-s\to\infty\\s_0\ \text{is even}}}
\frac{\operatorname{N}_{\mathbf e}^t(s;\F_q^k)}{\genfrac{[}{]}{0pt}{}{k}{s}_q}
=
\frac{q^{-s_0}g_{q,s}}
{h_{q,t}g_{q^2,s_0/2}},
\quad
\text{and}
\quad
\lim_{\substack{k-s\to\infty\\s_0\ \text{is odd}}}
\frac{\operatorname{N}_{\mathbf e}^t(s;\F_q^k)}{\genfrac{[}{]}{0pt}{}{k}{s}_q}
=0.
\]
A subcode has even-like dual precisely when it contains the all-one
vector, and
\[
\lim_{k-s\to\infty}
\frac{\genfrac{[}{]}{0pt}{}{k-1}{s-1}_q}{\genfrac{[}{]}{0pt}{}{k}{s}_q}
=
\lim_{k-s\to\infty}
\frac{q^s-1}{q^k-1}
=0.
\]
Consequently, the types ${\bf oe}$ and ${\bf ee}$ have limiting
proportion zero. Subtracting the even-like contribution from the total
limit proves part~{\rm (i)(a)} for a parent of type ${\bf oe}$.
For a parent of type ${\bf oo}$, Proposition~\ref{oo number} shows that
the same two limiting contributions have types ${\bf oo}$ and
${\bf eo}$, respectively, while the other two types are absent. This
proves part~{\rm (i)(a)} in both odd-like cases.

For an even-like {\rm LCD} parent, Proposition~\ref{eo number} gives
\(
\operatorname{N}_{\mathbf{eo}}^t
(s;{\rm L}_{\mathbf{eo}}^k)
=
\operatorname{N}_{\mathbf e}^t(s;\F_q^{k+1}),
\)
and all other complete-type counts are zero. Moreover,
\[
\lim_{k-s\to\infty}
\frac{\genfrac{[}{]}{0pt}{}{k+1}{s}_q}{\genfrac{[}{]}{0pt}{}{k}{s}_q}=q^s.
\]
Applying the preceding even-like whole-space limit with $k+1$ in
place of $k$ gives the factor $q^sq^{-s_0}=q^t$, and hence proves
part~{\rm (i)(b)}.

Suppose next that $q$ is odd. The {\rm LCD}-subcode formulas in
Lemmas~\ref{SubLCD-S}--\ref{SubLCD-NS}, together with the positive-hull
recurrences in Propositions~\ref{S-subcode-odd} and
\ref{NS-subcode-odd}, give the same three-factor calculation as above.
The factors involving the parent dimension are of the form
$q^M+1$ or $q^M-1$, with $M\to\infty$, and their ratios to $q^M$
have limit one. When $s_0$ is odd, the two subcode types therefore have
equal limiting weights. When $s_0$ is even, the two remaining fixed
factors are
\[
\frac12\left(1+\eta((-1)^{s_0/2})q^{-s_0/2}\right)
\quad\text{and}\quad
\frac12\left(1-\eta((-1)^{s_0/2})q^{-s_0/2}\right).
\]
Multiplying these factors by
$g_{q,s}/(h_{q,t}g_{q^2,\lfloor s_0/2\rfloor})$ proves
part~{\rm (ii)}. This argument is identical for parent types $+$ and
$-$, since their remaining parent-type factors have limit one.

Finally, the joint limits in part~{\rm (iii)} follow from the same
product calculations with both $s$ and $k-s$ tending to infinity.
For the zero-hull initial terms these are the limits in
\cite[Corollaries 3.18 and 4.20]{x-lcd}; the positive-hull factors are
given by Lemma~\ref{thm:even-whole-space-mass} and
Lemma~\ref{whole-space-count-odd}. The calculations are the same as
those above, with
\[
\lim_{s\to\infty}g_{q,s}=g_{q,\infty},
\qquad
\lim_{s\to\infty}g_{q^2,\lfloor s_0/2\rfloor}
=g_{q^2,\infty},
\]
and, in the type refinements,
\[
\lim_{s\to\infty}q^{-s_0}=0,
\qquad
\lim_{s\to\infty}q^{-s_0/2}=0.
\]
This proves all the stated joint limits.
\end{proof}

\subsection{Proof of  Lemma~\ref{lem:dominant-projection-term}}
\begin{proof}
For $m=0$, the product formula for Gaussian coefficients gives
\[
W_0(k,s,\ell)
=
\frac{\prod_{j=k-\ell-s+1}^{k-\ell}(1-q^{-j})}
{\prod_{j=k-\ell-s+1}^{k-\ell}(1-q^{-(j+\ell)})}.
\]
Using
\[
\prod_j(1-x_j)\geq1-\sum_jx_j
\qquad\text{if $0\leq x_j\leq1$},
\]
we obtain
\[
1-\sum_{j=k-\ell-s+1}^{k-\ell}q^{-j}
\leq
\prod_{j=k-\ell-s+1}^{k-\ell}(1-q^{-j})
\leq1.
\]
Moreover,
\[
\sum_{j=k-\ell-s+1}^{k-\ell}q^{-j}
\leq
\sum_{j=k-\ell-s+1}^{\infty}q^{-j}
=\frac{q^{-(k-\ell-s)}}{q-1}
.
\]
Since
\[
\lim_{k-\ell-s\to\infty}
\frac{q^{-(k-\ell-s)}}{q-1}=0,
\]
the numerator product tends to one. The denominator product tends
to one by the same argument, which proves the asserted limit for
$W_0(k,s,\ell)$.

Now let $m\geq1$ be fixed. The product formula gives
\[
\frac{W_m(k,s,\ell)}{W_0(k,s,\ell)}
=q^{-m(s+\ell-m)}\genfrac{[}{]}{0pt}{}{\ell}{m}_q
\prod_{j=0}^{m-1}
\frac{q^{s-j}-1}{q^{k-\ell-s+j+1}-1}.
\]
First,
\[
\genfrac{[}{]}{0pt}{}{\ell}{m}_q
=q^{m(\ell-m)}
\prod_{i=0}^{m-1}
\frac{1-q^{-(\ell-i)}}{1-q^{-(m-i)}}
\leq
\frac{q^{m(\ell-m)}}{\prod_{i=1}^{m}(1-q^{-i})}.
\]
For all sufficiently large parameters,
\[
q^{k-\ell-s+j+1}-1
\geq(1-q^{-1})q^{k-\ell-s+j+1}
\qquad\text{if $0\leq j\leq m-1$}.
\]
Consequently,
\[
0\leq
\frac{W_m(k,s,\ell)}{W_0(k,s,\ell)}
\leq
\frac{q^{-m(k-\ell-s)-m^2}}
{(1-q^{-1})^m\prod_{i=1}^{m}(1-q^{-i})},
\]
which tends to zero because $m$ is fixed and
$k-\ell-s\to\infty$. Combining this with the asserted limit for
$W_0(k,s,\ell)$ proves the result.
\end{proof}
\subsection{Proof of  Theorem~\ref{thm:asymptotic}}
\begin{proof}
We prove only part {\rm (i)(a)}, since the other cases are handled similarly.

If the dual of $C_{\bf o}^{k,\ell}$ is even-like, then
by Corollary~\ref{HLOE},
\begin{align*}
\frac{
\operatorname{N}_{{\bf oe}}^t
(s;C_{{\bf oe}}^{k,\ell})
}{
\genfrac{[}{]}{0pt}{}{k}{s}_q
}&=
\frac{
\displaystyle\sum_{m=0}^{\min\{t,\ell\}}
q^{(s-m-1)(\ell-m)}
\genfrac{[}{]}{0pt}{}{\ell}{m}_q\operatorname{N}_{{\bf oe}}^{t-m}
(s-m;{\rm L}_{{\bf oe}}^{k-\ell})
}{
\genfrac{[}{]}{0pt}{}{k}{s}_q
}\\
&=\displaystyle\sum_{m=0}^{\min\{t,\ell\}}\frac{W_m(k,s,\ell)}{q^{\ell-m}} \cdot \frac{\operatorname{N}_{{\bf oe}}^{t-m}
(s-m;{\rm L}_{{\bf oe}}^{k-\ell})}{\genfrac{[}{]}{0pt}{}{k-\ell}{s-m}_q}\\
&\leq \displaystyle\sum_{m=0}^{\min\{t,\ell\}}W_m(k,s,\ell) \cdot \frac{\operatorname{N}_{{\bf oe}}^{t-m}
(s-m;{\rm L}_{{\bf oe}}^{k-\ell})}{\genfrac{[}{]}{0pt}{}{k-\ell}{s-m}_q}.
\end{align*}

Hence, by Proposition \ref{prop:asymptotic_limits} and Lemma \ref{lem:dominant-projection-term}, we have
\[
        \lim_{\substack{s\to\infty\\k-\ell-s\to\infty}}
        \frac{\operatorname{N}^t_{{\bf oe}}(s;C_{\bf oe}^{k,\ell})}
        {\genfrac{[}{]}{0pt}{}{k}{s}_q}=0.
\]

Similarly, we deduce that
\[
        \lim_{\substack{s\to\infty\\k-\ell-s\to\infty}}
        \frac{\operatorname{N}^t_{{\bf ee}}(s;C_{\bf oe}^{k,\ell})}
        {\genfrac{[}{]}{0pt}{}{k}{s}_q}=
        \lim_{\substack{s\to\infty\\k-\ell-s\to\infty}}
        \frac{\operatorname{N}^t_{{\bf eo}}(s;C_{\bf oe}^{k,\ell})}
        {\genfrac{[}{]}{0pt}{}{k}{s}_q}=0
\]
and 
\begin{align*}
\frac{
\operatorname{N}^t
(s;C_{{\bf oe}}^{k,\ell})
}{
\genfrac{[}{]}{0pt}{}{k}{s}_q
}&=
\frac{
\displaystyle\sum_{m=0}^{\min\{t,\ell\}}
q^{(s-m)(\ell-m)}
\genfrac{[}{]}{0pt}{}{\ell}{m}_q\operatorname{N}^{t-m}
(s-m;{\rm L}_{{\bf oe}}^{k-\ell})
}{
\genfrac{[}{]}{0pt}{}{k}{s}_q
}\\
&= \displaystyle\sum_{m=0}^{\min\{t,\ell\}}W_m(k,s,\ell) \cdot \frac{\operatorname{N}^{t-m}
(s-m;{\rm L}_{{\bf oe}}^{k-\ell})}{\genfrac{[}{]}{0pt}{}{k-\ell}{s-m}_q}.
\end{align*}
Hence
\[
        \lim_{\substack{s\to\infty\\k-\ell-s\to\infty}}
        \frac{\operatorname{N}^t(s;C_{\bf o}^{k,\ell})}
        {\genfrac{[}{]}{0pt}{}{k}{s}_q}=\frac{g_{q,\infty}}{g_{q^2,\infty}\cdot h_{q,t}}.
\]

If the dual of $C_{\bf o}^{k,\ell}$ is odd-like, by Theorem \ref{all-hull}, we have $\operatorname{N}_{\bf o}^{t}(s;C_{\bf oe}^{k,\ell})=\operatorname{N}_{\bf o}^{t}(s;C_{\bf oo}^{k,\ell})$. Therefore, \[
        \lim_{\substack{s\to\infty\\k-\ell-s\to\infty}}
        \frac{\operatorname{N}^t(s;C_{\bf oo}^{k,\ell})}
        {\genfrac{[}{]}{0pt}{}{k}{s}_q}
        =
        \lim_{\substack{s\to\infty\\k-\ell-s\to\infty}}
        \frac{\operatorname{N}_{{\bf oo}}^t
        (s;C_{\bf oo}^{k,\ell})}
        {\genfrac{[}{]}{0pt}{}{k}{s}_q}
        =
        \frac{g_{q,\infty}}
        {h_{q,t}\cdot g_{q^2,\infty}}.
        \]

This proves the assertion in part (i)(a).
\end{proof}

\subsection{Proof of  Theorem~\ref{thm:asymptotic-finite}}
\begin{proof}
By Proposition~\ref{lcd-complement-types}, write
\[
C^{k,\ell}={\rm L}^{k_0}\oplus\operatorname{Hull}(C^{k,\ell}).
\]
Let $D\subseteq C^{k,\ell}$ be an $s$-dimensional subcode with
$(-a)$-dimensional hull and put
$m=\dim(D\cap\operatorname{Hull}(C^{k,\ell}))$. The projection formula
gives $m\leq -a$, while
\[
\dim\Pi_{\rm L}(D)=s-m\leq k_0
\]
gives $m\geq s-k_0=-a$. Hence $m=-a$ and
$\Pi_{\rm L}(D)={\rm L}^{k_0}$. The projection-counting formula therefore yields
\[
\operatorname{N}^{-a}(s;C^{k,\ell})
=q^{k_0(k-s)}\genfrac{[}{]}{0pt}{}{\ell}{-a}_q.
\]
Using the product formula for Gaussian coefficients, we obtain
\[
\frac{\operatorname{N}^{-a}(s;C^{k,\ell})}{\genfrac{[}{]}{0pt}{}{k}{s}_q}
=
\frac{g_{q,\ell}g_{q,s}}{g_{q,-a}g_{q,k}}
=
\frac{\prod_{j=-a+1}^{\ell}(1-q^{-j})}
{\prod_{j=s+1}^{k}(1-q^{-j})}.
\]
If $a$ is fixed and $k_0,\ell\to\infty$, then
$s=k_0-a\to\infty$ and $k=k_0+\ell\to\infty$, which proves part~{\rm (i)}.

 From the product expression above,
\[
\frac{\operatorname{N}^{-a}(s;C^{k,\ell})}{\genfrac{[}{]}{0pt}{}{k}{s}_q}
=
\frac{\prod_{j=-a+1}^{\ell}(1-q^{-j})}
{\prod_{j=s+1}^{k}(1-q^{-j})}.
\]
This quotient is a proportion and hence is at most one. Since its
denominator is at most one, while
\[
\prod_{j=-a+1}^{\ell}(1-q^{-j})
\geq
1-\sum_{j=-a+1}^{\ell}q^{-j}
\geq
1-\sum_{j=-a+1}^{\infty}q^{-j}
=
1-\frac{q^{a}}{q-1},
\]
we obtain
\[
1-\frac{q^a}{q-1}
\leq
\frac{\operatorname{N}^{-a}(s;C^{k,\ell})}{\genfrac{[}{]}{0pt}{}{k}{s}_q}
\leq1.
\]
Letting $a\to-\infty$ proves part~{\rm (ii)}.

Now assume that $q$ is even. If $C^{k,\ell}$ is of type
${\bf oo}$, then Lemma~\ref{even-type-bookkeeping} and
Lemma~\ref{even-parity-proj} show that every subcode counted above is
of type ${\bf oo}$. If $C^{k,\ell}$ is of type ${\bf oe}$,
then every subcode counted above is odd-like. Moreover, by
Corollary~\ref{HLOE},
\[
\frac{
\operatorname{N}_{{\bf oe}}^{-a}(s;C_{{\bf oe}}^{k,\ell})
}
{\genfrac{[}{]}{0pt}{}{k}{s}_q}
=
\frac{
q^{(k_0-1)(k-s)}
\genfrac{[}{]}{0pt}{}{\ell}{-a}_q
}
{\genfrac{[}{]}{0pt}{}{k}{s}_q}
=
q^{-(k-s)}
\frac{g_{q,\ell}g_{q,s}}
{g_{q,-a}g_{q,k}},
\]
which tends to $0$. Since every subcode counted above is odd-like in
this case, the remaining counted subcodes are of type
${\bf oo}$. Hence, for every odd-like parent code,
\[
\lim_{\substack{a\to-\infty\\k-s\to\infty}}
\frac{\operatorname{N}_{{\bf oo}}^{-a}
(s;C_{\mathbf{o}}^{k,\ell})}
{\genfrac{[}{]}{0pt}{}{k}{s}_q}
=1.
\]

If $C^{k,\ell}$ is of type ${\bf eo}$, then every subcode of
$C^{k,\ell}$ is of type ${\bf eo}$. If $C^{k,\ell}$ is of type
${\bf ee}$, then every subcode counted above is even-like, and
Corollary~\ref{HLEE} gives
\[
\frac{
\operatorname{N}_{{\bf ee}}^{-a}(s;C_{{\bf ee}}^{k,\ell})
}
{\genfrac{[}{]}{0pt}{}{k}{s}_q}
=
\frac{
q^{k_0(k-s)}
\genfrac{[}{]}{0pt}{}{\ell-1}{-a-1}_q
}
{\genfrac{[}{]}{0pt}{}{k}{s}_q}.
\]
Since
\[
\lim_{\substack{a\to-\infty\\k-s\to\infty}}
\frac{\genfrac{[}{]}{0pt}{}{\ell-1}{-a-1}_q}{\genfrac{[}{]}{0pt}{}{\ell}{-a}_q}
=
\lim_{\substack{a\to-\infty\\k-s\to\infty}}
\frac{q^{-a}-1}{q^\ell-1}
=0,
\]
the type-${\bf ee}$ contribution is negligible. Since every subcode
counted above is even-like in this case, we obtain, for every even-like
parent code,
\[
\lim_{\substack{a\to-\infty\\k-s\to\infty}}
\frac{\operatorname{N}_{{\bf eo}}^{-a}
(s;C_{\mathbf{e}}^{k,\ell})}
{\genfrac{[}{]}{0pt}{}{k}{s}_q}
=1.
\]

When $q$ is odd, Lemma~\ref{same-type} shows that the type of each
subcode counted above is the same as the type of the projected subcode
${\rm L}^{k_0}$. By Proposition~\ref{lcd-complement-types}, ${\rm L}^{k_0}$ has the same
type as $C^{k,\ell}$ whenever $k_0>0$. When $k_0=0$, the parent code is
self-orthogonal and the assertion is immediate. This proves
part~{\rm (iii)}.
\end{proof}
\subsection{Proof of  Theorem~\ref{thm:bounded-gap-minimum}}
\begin{proof}
By Proposition~\ref{lcd-complement-types}, write
$
C^{k,\ell}={\rm L}^{k_0}\oplus\operatorname{Hull}(C^{k,\ell}).
$ 
Since ${\rm L}^{k_0}$ is {\rm LCD}, internal orthogonal complementation in
${\rm L}^{k_0}$ preserves the hull dimension. Hence, for every $u$ and $v$,
$
\operatorname{N}^v(u;{\rm L}^{k_0})
=
\operatorname{N}^v(k_0-u;{\rm L}^{k_0}).
$

The fixed-dimensional limits needed below are precisely those in
Proposition~\ref{prop:asymptotic_limits}. In the nonalternating
even-characteristic case, and in either type in odd characteristic,
that proposition gives, for every fixed nonnegative integer $u$,
\[
\lim_{k_0\to\infty}
\frac{\operatorname{N}^0(u;{\rm L}^{k_0})}{\genfrac{[}{]}{0pt}{}{k_0}{u}_q}
=\frac{g_{q,u}}{g_{q^2,\lfloor u/2\rfloor}}.
\]
For an even-like {\rm LCD} parent in even characteristic it gives
\[
\lim_{k_0\to\infty}
\frac{\operatorname{N}^0(2j;{\rm L}_{\mathbf{eo}}^{k_0})}
{\genfrac{[}{]}{0pt}{}{k_0}{2j}_q}
=\frac{g_{q,2j}}{g_{q^2,j}}
\quad
\text{and}
\quad
\lim_{k_0\to\infty}
\frac{\operatorname{N}^1(2j+1;{\rm L}_{\mathbf{eo}}^{k_0})}
{\genfrac{[}{]}{0pt}{}{k_0}{2j+1}_q}
=\frac{g_{q,2j+1}}{g_{q,1}g_{q^2,j}}.
\]

Suppose first that the restricted form on ${\rm L}^{k_0}$ is
nonalternating, or that $q$ is odd. A subcode with zero-dimensional
hull has zero intersection with $\operatorname{Hull}(C^{k,\ell})$.
The projection-counting formula and internal orthogonal complementation
therefore give
\[
\operatorname{N}^0(s;C^{k,\ell})
=q^{s\ell}\operatorname{N}^0(s;{\rm L}^{k_0})
=q^{s\ell}\operatorname{N}^0(a;{\rm L}^{k_0}).
\]
Moreover,
\[
\lim_{\substack{k_0\to\infty\\\ell\to\infty\\s=k_0-a}}
\frac{q^{s\ell}\genfrac{[}{]}{0pt}{}{k_0}{s}_q}{\genfrac{[}{]}{0pt}{}{k}{s}_q}
=
\lim_{\substack{k_0\to\infty\\\ell\to\infty}}
\frac{g_{q,k_0}g_{q,\ell+a}}
{g_{q,a}g_{q,k}}
=
\frac{g_{q,\infty}}{g_{q,a}}.
\]
Combining this limit with the first finite-dimensional limit proves
part~{\rm (i)(a)} and part~{\rm (ii)}. The same argument, using
the second limit, proves part~{\rm (i)(b)}.

It remains to consider an even-like parent code when $q$ is even and
$a=2j+1$ is odd. In this case $k_0$ is even and $s=k_0-a$ is odd. Every
subcode of $C_{\mathbf{e}}^{k,\ell}$ is even-like, so its Gram matrix
is alternating and has even rank. Its hull dimension therefore has
the same parity as $s$, and hence is at least one.

For hull dimension one, the projection formula has only the two terms
corresponding to intersection dimensions zero and one. Thus
\begin{align*}
\operatorname{N}^1(s;C_{\mathbf{e}}^{k,\ell})
={}q^{s\ell}\operatorname{N}^1(s;{\rm L}_{\mathbf{eo}}^{k_0})
+q^{(s-1)(\ell-1)}\genfrac{[}{]}{0pt}{}{\ell}{1}_q
\operatorname{N}^0(s-1;{\rm L}_{\mathbf{eo}}^{k_0}).
\end{align*}
Internal orthogonal complementation gives
\[
\operatorname{N}^1(s;{\rm L}_{\mathbf{eo}}^{k_0})
=\operatorname{N}^1(a;{\rm L}_{\mathbf{eo}}^{k_0}),
\qquad
\operatorname{N}^0(s-1;{\rm L}_{\mathbf{eo}}^{k_0})
=\operatorname{N}^0(a+1;{\rm L}_{\mathbf{eo}}^{k_0}).
\]
The two corresponding projection weights satisfy
\[
\lim_{\substack{k_0\to\infty\\\ell\to\infty\\s=k_0-a}}
\frac{q^{s\ell}\genfrac{[}{]}{0pt}{}{k_0}{a}_q}{\genfrac{[}{]}{0pt}{}{k}{s}_q}
=
\frac{g_{q,\infty}}{g_{q,a}}
\quad
\text{and}
\quad
\lim_{\substack{k_0\to\infty\\\ell\to\infty\\s=k_0-a}}
\frac{q^{(s-1)(\ell-1)}\genfrac{[}{]}{0pt}{}{\ell}{1}_q
\genfrac{[}{]}{0pt}{}{k_0}{a+1}_q}{\genfrac{[}{]}{0pt}{}{k}{s}_q}
=
\frac{q^{-(a+1)}g_{q,\infty}}
{g_{q,1}g_{q,a+1}}.
\]
Using the two even-like limits from
Proposition~\ref{prop:asymptotic_limits}, the required limit is therefore
\[
\frac{g_{q,\infty}}{g_{q,1}g_{q^2,j}}
+
\frac{q^{-(a+1)}g_{q,\infty}}
{g_{q,1}g_{q^2,j+1}}=\frac{g_{q,\infty}}
{g_{q,1}g_{q^2,j+1}},
\]
which proves part~{\rm (i)(c)}.

It remains to prove the complete-type assertions in even
characteristic. First let the parent be odd-like. If it has type
${\bf oo}$, Corollary~\ref{HLOO} shows that only types ${\bf oo}$ and
${\bf eo}$ occur among the subcodes counted above. If $s$ is odd, the
type-${\bf eo}$ count is zero. If $s$ is even, Lemma~\ref{lem:even-lcd-mass}
and Proposition~\ref{oo number} give
\[
\frac{\operatorname{N}_{\mathbf{eo}}^0
(s;{\rm L}_{\mathbf{oo}}^{k_0})}
{\operatorname{N}_{\mathbf{oo}}^0
(s;{\rm L}_{\mathbf{oo}}^{k_0})}
\leq\frac1{q^s-1}.
\]
Since $s=k_0-a\to\infty$, the type-${\bf eo}$ contribution tends to
zero. If the parent has type ${\bf oe}$, Corollary~\ref{oe=oo} gives
the same parity-label proportions. In addition,
Corollary~\ref{HLOE} shows that the type-${\bf oe}$ contribution has an
extra factor $q^{-\ell}$ relative to the corresponding total
projection term. Proposition~\ref{prop:asymptotic_limits}, applied after
internal orthogonal complementation to the fixed dimension $a$, shows
that the normalized {\rm LCD} count has a positive limit. Hence the
type-${\bf oe}$ contribution also tends to zero. Thus type ${\bf oo}$
is dominant for every odd-like parent.

If the parent has type ${\bf eo}$, Corollary~\ref{HLEO} shows that every
subcode has type ${\bf eo}$. If the parent has type ${\bf ee}$ and $a$
is even, Corollary~\ref{HLEE} gives
$\operatorname{N}_{\mathbf{ee}}^0(s;C_{\mathbf{ee}}^{k,\ell})=0$.
Finally, suppose that the parent has type ${\bf ee}$ and $a$ is odd.
For hull dimension one, Corollary~\ref{HLEE} and the total projection
formula give
\[
\frac{\operatorname{N}_{\mathbf{ee}}^1
(s;C_{\mathbf{ee}}^{k,\ell})}
{\operatorname{N}^1(s;C_{\mathbf{ee}}^{k,\ell})}
\leq\frac1{\genfrac{[}{]}{0pt}{}{\ell}{1}_q},
\]
which tends to zero as $\ell\to\infty$. Consequently type ${\bf eo}$
is dominant for every even-like parent. This completes the proof.
\end{proof}

\section{Examples}
The following binary and ternary examples illustrate the prescribed-hull subcode distributions and compare the best minimum distances attained in the resulting subcode layers with the corresponding global fixed-hull optima.

\begin{example}\label{ex:binary-ee-distance}
Let $C_{\bf ee}^{7,3}$ be the binary linear code generated by
\[
G=
\begin{pmatrix}
1&0&0&0&0&0&0&1&1&1&0&0\\
0&1&0&0&0&0&0&0&1&1&0&1\\
0&0&1&0&0&0&0&1&1&0&0&1\\
0&0&0&1&0&0&0&1&0&1&0&1\\
0&0&0&0&1&0&0&0&1&1&1&0\\
0&0&0&0&0&1&0&1&1&0&1&0\\
0&0&0&0&0&0&1&1&0&1&1&0
\end{pmatrix}.
\]
A direct computation gives
\[
C_{\bf ee}^{7,3}=[12,7,4]_2,
\qquad
\dim\Hull(C_{\bf ee}^{7,3})=3.
\]

Both $C_{\bf ee}^{7,3}$ and its dual are even-like, so the parent code
has type ${\bf ee}$.

We enumerated all proper nonzero subcodes of $C_{\bf ee}^{7,3}$ in
Magma \cite{magma} and classified them according to their dimensions,
hull dimensions, and complete types. The five counting columns in
Table~\ref{tab:binary-ee-distance} agree entry by entry with
Corollary~\ref{HLEE} and the corresponding total-count formula. We also
record the global fixed-hull optimum $d_2(s,t;\F_2^{12})$ and the largest
minimum distance $d_2(s,t;C_{\bf ee}^{7,3})$ attained inside the fixed
parent. The global fixed-hull optimal distances in the last column are
taken from the classification of binary codes with prescribed positive
hull dimensions in Li and Shi~\cite{Li-hull-even}, together with the
known optimal distances of binary {\rm LCD} codes for $t=0$ in
Galvez \emph{et al.}~\cite{Galvez-LCD}.
A superscript ${}^\ast$
means
\[
d_2(s,t;C_{\bf ee}^{7,3})=d_2(s,t;\F_2^{12}).
\]
All unlisted values are zero.

\begin{table}[htbp]
\centering
\caption{Subcodes of $C_{\bf ee}^{7,3}$ classified by hull dimension
and complete type, together with fixed-hull optimal-distance
comparisons.}
\label{tab:binary-ee-distance}

\renewcommand{\arraystretch}{1.15}
\scriptsize

\begin{tabular*}{\linewidth}{
@{}
c@{\hspace{0.55em}}
c@{\hspace{4.0em}}
r@{\extracolsep{\fill}}
rrrr
@{\extracolsep{0pt}\hspace{5.0em}}
rr
@{}
}
\hline
$s$ & $t$
& ${\rm N}_{\bf oo}^t(s;C_{\bf ee}^{7,3})$
& ${\rm N}_{\bf oe}^t(s;C_{\bf ee}^{7,3})$
& ${\rm N}_{\bf eo}^t(s;C_{\bf ee}^{7,3})$
& ${\rm N}_{\bf ee}^t(s;C_{\bf ee}^{7,3})$
& ${\rm N}^t(s;C_{\bf ee}^{7,3})$
& $d_2(s,t;C_{\bf ee}^{7,3})$
& $d_2(s,t;\F_2^{12})$\\
\hline

1&1&0&0&126&1&127&$12^\ast$&12\\
\hline

2&0&0&0&1280&0&1280&6&7\\
2&2&0&0&1324&63&1387&$8^\ast$&8\\
\hline

3&1&0&0&9600&320&9920&$6^\ast$&6\\
3&3&0&0&1560&331&1891&$6^\ast$&6\\
\hline

4&0&0&0&4096&0&4096&4&5\\
4&2&0&0&6080&1200&7280&$6^\ast$&6\\
4&4&0&0&240&195&435&$4^\ast$&4\\
\hline

5&1&0&0&1536&256&1792&$4^\ast$&4\\
5&3&0&0&480&380&860&$4^\ast$&4\\
5&5&0&0&0&15&15&$4^\ast$&4\\
\hline

6&2&0&0&64&48&112&$4^\ast$&4\\
6&4&0&0&0&15&15&$4^\ast$&4\\
\hline

\end{tabular*}
\end{table}
There are thirteen nonzero $(s,t)$ layers among the proper subcodes of
$C_{\bf ee}^{7,3}$. Eleven of them attain the corresponding global
fixed-hull optimum. The only exceptions are
\[
(s,t)=(2,0)
\qquad\text{and}\qquad
(s,t)=(4,0),
\]
and in both cases the gap is exactly one:
\[
d_2(2,0;C_{\bf ee}^{7,3})=6=d_2(2,0;\F_2^{12})-1,
\]
\[
d_2(4,0;C_{\bf ee}^{7,3})=4=d_2(4,0;\F_2^{12})-1.
\]

\end{example}

\begin{example}\label{ex:ternary-ns-12}
Let $C_{-}^{6,1}$ be the ternary linear code generated by
\[
G=
\begin{pmatrix}
1&0&0&0&0&0&1&0&1&0&2&2\\
0&1&0&0&0&0&0&0&1&1&1&1\\
0&0&1&0&0&0&0&1&2&1&1&0\\
0&0&0&1&0&0&1&1&1&1&0&0\\
0&0&0&0&1&0&1&2&1&0&1&0\\
0&0&0&0&0&1&1&2&0&1&0&1
\end{pmatrix}.
\]
A direct computation gives
\[
C_{-}^{6,1}=[12,6,5]_3,
\qquad
\dim\Hull(C_{-}^{6,1})=1.
\]
The determinant square class of the induced nondegenerate part is
represented by $2$. Since $2$ is a nonsquare in $\F_3$, the parent code
has type $-$. 

We exhaustively enumerated all
\[
\sum_{s=1}^{5}\genfrac{[}{]}{0pt}{}{6}{s}_3=56630
\]
proper nonzero subcodes of $C_{-}^{6,1}$ in Magma \cite{magma}.
For every occurring pair $(s,t)$, Table~\ref{tab:ternary-ns-types}
lists the numbers of type-$+$ and type-$-$ subcodes and their sum.
These counts agree entry by entry with Theorem~\ref{linear-code-count}.
In particular, if ${\rm L}_{-}^{5}$ is an LCD complement of the hull,
the computation verifies
\[
\N_3\bigl[(s,t)_\tau;(6,1)_{-}\bigr]
=
\sum_{m=0}^{\min\{t,1\}}
3^{(s-m)(1-m)}
\genfrac{[}{]}{0pt}{}{1}{m}_3
\N_3\bigl[(s-m,t-m)_\tau;(5,0)_{-}\bigr]
\]
for both $\tau=+$ and $\tau=-$.

The last two columns compare the global fixed-hull optimum
$d_3(s,t;\F_3^{12})$ with $d_3(s,t;C_{-}^{6,1})$. The benchmark values
are obtained from the known optimal distances of ternary {\rm LCD}
codes for $t=0$~\cite{A3}, the classification of ternary
one-dimensional-hull codes for $t=1$~\cite{Tong-hull-one}, and, for
$t\geq2$, the known hull distributions of optimal ternary
codes~\cite{optimal-hulls}, together with standard upper bounds and the
explicit subcodes found in our computation. For the remaining case
$d_3(5,2;\F_3^{12})$, the classification of optimal $[12,5,6]_3$ codes
in \cite{optimal-hulls} gives the upper bound $d_3(5,2;\F_3^{12})\leq5$,
whereas our exhaustive computation exhibits a hull-$2$ $[12,5,5]_3$
subcode of $C_{-}^{6,1}$. Hence
\[
d_3(5,2;\F_3^{12})=5.
\]
As before, a superscript ${}^\ast$ marks equality with the global
fixed-hull optimum, and all unlisted values are zero.

\begin{table}[htbp]
\centering
\caption{Subcodes of $C_{-}^{6,1}$ classified by hull dimension and
type, together with fixed-hull optimal-distance comparisons.}
\label{tab:ternary-ns-types}

\renewcommand{\arraystretch}{1.15}
\scriptsize

\begin{tabular*}{\linewidth}{
@{}
c@{\hspace{0.55em}}
c@{\hspace{10em}}
r@{\extracolsep{\fill}}
rr
@{\extracolsep{0pt}\hspace{10em}}
rr
@{}
}
\hline
$s$ & $t$
& ${\rm N}_{+}^t(s;C_{-}^{6,1})$
& ${\rm N}_{-}^t(s;C_{-}^{6,1})$
& ${\rm N}^t(s;C_{-}^{6,1})$
& $d_3(s,t;C_{-}^{6,1})$
& $d_3(s,t;\F_3^{12})$\\
\hline
1&0&108&135&243&$11^\ast$&11\\
1&1&121&0&121&$12^\ast$&12\\
\hline

2&0&2430&4860&7290&$8^\ast$&8\\
2&1&1116&2205&3321&$8^\ast$&8\\
2&2&400&0&400&$9^\ast$&9\\
\hline

3&0&14580&7290&21870&$7^\ast$&7\\
3&1&3510&7020&10530&$7^\ast$&7\\
3&2&120&1320&1440&$8^\ast$&8\\
3&3&40&0&40&$6^\ast$&6\\
\hline

4&0&3645&2916&6561&$6^\ast$&6\\
4&1&3780&270&4050&$6^\ast$&6\\
4&2&120&240&360&$6^\ast$&6\\
4&3&0&40&40&5&6\\
\hline

5&0&0&243&243&$5^\ast$&5\\
5&1&45&36&81&5&6\\
5&2&40&0&40&$5^\ast$&5\\
\hline

\end{tabular*}
\end{table}

Among the sixteen nonzero $(s,t)$ layers occurring in the proper
subcodes of $C_{-}^{6,1}$, fourteen attain the global fixed-hull
optimum. The only exceptions are
\[
(s,t)=(4,3)
\qquad\text{and}\qquad
(s,t)=(5,1),
\]
for which
\[
d_3(4,3;C_{-}^{6,1})=5=d_3(4,3;\F_3^{12})-1,
\]
\[
d_3(5,1;C_{-}^{6,1})=5=d_3(5,1;\F_3^{12})-1.
\]

\end{example}

\bigskip 

\noindent\textbf{Acknowledgements.}\\
The work of Linhan Xiong and Hongwei Liu was supported by NSFC (Grant Nos. 12441102 and 12271199). The work of Liren Lin was supported by NSFC (Grant Nos. 12201192).

\end{document}